\documentclass[letterpaper,11pt]{article}
\usepackage[letterpaper,margin=1.0in]{geometry}

\usepackage{microtype}
\usepackage{xcolor}
\usepackage{amsmath}
\usepackage{amssymb}
\usepackage{amsthm}
\usepackage{mathrsfs}
\usepackage{enumitem}
\usepackage{textgreek}
\usepackage{graphicx}
\usepackage{framed}
\usepackage{doi}
\usepackage[textsize=tiny]{todonotes}
\usepackage[framemethod=tikz]{mdframed}
\usepackage{thm-restate}
\usepackage{mathtools}
\usepackage{xspace}
\usepackage[capitalize, nameinlink]{cleveref}
\usepackage{multirow}
\usepackage{cellspace}
\usepackage{array}

\Crefname{remark}{Remark}{Remarks}
\Crefname{observation}{Observation}{Observations}

\theoremstyle{plain}
\newtheorem{theorem}{Theorem}[section]
\newtheorem{lemma}[theorem]{Lemma}

\newtheorem{corollary}[theorem]{Corollary}
\newtheorem{claim}[theorem]{Claim}
\newtheorem{observation}[theorem]{Observation}

\theoremstyle{definition}
\newtheorem{definition}[theorem]{Definition}

\theoremstyle{plain}

\newcounter{open}
\newtheorem{oq}[open]{Open Problem}

\theoremstyle{remark}

\newcommand{\namedref}[2]{\hyperref[#2]{#1~\ref*{#2}}}

\newcommand{\eps}{\varepsilon}

\newcommand{\s}{\mspace{1mu}}
\newcommand{\mybox}[1]{\mspace{2mu}{\setlength{\fboxsep}{1.5pt}\color{lightgray}\boxed{\color{black}\scriptstyle #1}}\mspace{2mu}}
\newcommand{\A}{\mathsf{A}}
\newcommand{\B}{\mathsf{B}}
\newcommand{\C}{\mathsf{C}}
\renewcommand{\L}{\mathsf{L}}

\newcommand{\U}{\mathsf{U}}

\newcommand{\M}{\mathsf{M}}
\renewcommand{\P}{\mathsf{P}}
\renewcommand{\O}{\mathsf{O}}
\newcommand{\D}{\mathsf{D}}
\newcommand{\Q}{\mathsf{Q}}

\newcommand{\X}{\mathsf{X}}

\newcommand{\bM}{\mybox{\M}}
\newcommand{\bP}{\mybox{\P}}
\newcommand{\bO}{\mybox{\O}}

\newcommand{\bX}{\mybox{\X}}

\newcommand{\fA}{\mathcal{A}}
\newcommand{\fB}{\mathcal{B}}

\newcommand{\fQ}{\mathcal{Q}}

\newcommand{\colour}{\mathsf{color}}
\newcommand{\rank}{\mathrm{rank}}

\DeclareMathOperator{\re}{\mathcal R}
\DeclareMathOperator{\rere}{\overline{\mathcal R}}

\newcommand{\nodeconst}{\ensuremath{\mathcal{N}}}
\newcommand{\edgeconst}{\ensuremath{\mathcal{E}}}
\newcommand{\gen}[1]{\langle #1 \rangle}
\newcommand{\repi}{\ensuremath{\re(\Pi)}}
\newcommand{\rerepi}{\ensuremath{\rere(\re(\Pi))}}

\newcommand{\epi}{\ensuremath{\edgeconst_{\Pi}}}
\newcommand{\npi}{\ensuremath{\nodeconst_{\Pi}}}
\newcommand{\spi}{\ensuremath{\Sigma_{\Pi}}}
\newcommand{\ezs}{\ensuremath{\edgeconst(z,s)}}
\newcommand{\nzs}{\ensuremath{\nodeconst(z,s)}}
\newcommand{\szs}{\ensuremath{\Sigma(z,s)}}
\newcommand{\ere}{\ensuremath{\edgeconst_{\re}}}
\newcommand{\nre}{\ensuremath{\nodeconst_{\re}}}
\newcommand{\sre}{\ensuremath{\Sigma_{\re}}}
\newcommand{\erere}{\ensuremath{\edgeconst_{\rere}}}
\newcommand{\nrere}{\ensuremath{\nodeconst_{\rere}}}
\newcommand{\srere}{\ensuremath{\Sigma_{\rere}}}
\newcommand{\repizs}{\ensuremath{\re(\Pi(z,s))}}
\newcommand{\edda}{\ensuremath{\edgeconst}}
\newcommand{\neda}{\ensuremath{\nodeconst}}
\newcommand{\rerepizs}{\ensuremath{\rere(\re(\Pi(z,s)))}}
\newcommand{\pistar}{\ensuremath{\Pi^*_q}}
\newcommand{\estar}{\ensuremath{\edgeconst^*_q}}
\newcommand{\nstar}{\ensuremath{\nodeconst^*_q}}
\newcommand{\sstar}{\ensuremath{\Sigma^*_q}}
\newcommand{\fminus}{\mathcal{F}^-_q}
\newcommand{\fplus}{\mathcal{F}^+_q}

\newcommand{\markedi}{\mathsf{m_i}}

\DeclareMathOperator{\poly}{poly}

\newcommand{\LOCAL}{\ensuremath{\mathsf{LOCAL}}\xspace}
\newcommand{\CONGEST}{\ensuremath{\mathsf{CONGEST}}\xspace}

\newcommand{\set}[1]{\left\{#1\right\}}

\newcommand{\ccs}{\mathfrak{C}}
\newcommand{\ccc}{\mathcal{C}}

\newcommand{\len}{\mathrm{len}}

\newenvironment{myabstract}
 {\list{}{\listparindent 1.5em%
         \itemindent    \listparindent
         \leftmargin    1cm
         \rightmargin   1cm
         \parsep        0pt}%
     \item\relax}
 {\endlist}

 \newenvironment{mycover}
 {\list{}{\listparindent 0pt
         \itemindent    \listparindent
         \leftmargin    1cm
         \rightmargin   1cm
         \parsep        0pt}%
     \raggedright
     \item\relax}
 {\endlist}

\newcommand{\myemail}[1]{\,$\cdot$\, {\small #1}}
\newcommand{\myaff}[1]{\,$\cdot$\, {\small #1}\par\smallskip}

\definecolor{darkgreen}{rgb}{0,0.5,0}
\definecolor{darkred}{rgb}{0.4,0,0}
\hypersetup{
    colorlinks=true,
    linkcolor=darkred,
    citecolor=darkgreen,
    filecolor=black,
    urlcolor=[rgb]{0,0.1,0.5},
    pdftitle={Distributed Maximal Matching and Maximal Independent Set on Hypergraphs},
    pdfauthor={}
}

\begin{document}

 \setcounter{page}{0}
 \thispagestyle{empty}

 \begin{mycover}
     {\huge\bfseries\boldmath Distributed Maximal Matching and Maximal Independent Set on Hypergraphs \par}
     \bigskip
     \bigskip
     \bigskip
     \textbf{Alkida Balliu}
     \myemail{alkida.balliu@gssi.it}
     \myaff{Gran Sasso Science Institute}

     \textbf{Sebastian Brandt}
     \myemail{brandt@cispa.de}
     \myaff{CISPA Helmholtz Center for Information Security}
    
     \textbf{Fabian Kuhn}
     \myemail{kuhn@cs.uni-freiburg.de}
     \myaff{University of Freiburg}

     \textbf{Dennis Olivetti}
     \myemail{dennis.olivetti@gssi.it}
     \myaff{Gran Sasso Science Institute}
 \end{mycover}
 \bigskip

\begin{myabstract}
  We investigate the distributed complexity of maximal matching and maximal independent set (MIS) in hypergraphs in the \LOCAL model. A maximal matching of a hypergraph $H=(V_H,E_H)$ is a maximal disjoint set $M\subseteq E_H$ of hyperedges and an MIS $S\subseteq V_H$ is a maximal set of nodes such that no hyperedge is fully contained in $S$. Both problems can be solved by a simple sequential greedy algorithm, which can be implemented na\"ively in $O(\Delta r + \log^* n)$ rounds, where $\Delta$ is the maximum degree, $r$ is the rank, and $n$ is the number of nodes of the hypergraph.

  We show that for maximal matching, this  na\"{i}ve algorithm is optimal in the following sense. Any deterministic algorithm for solving the problem requires $\Omega(\min\set{\Delta r, \log_{\Delta r} n})$ rounds, and any randomized one requires $\Omega(\min\set{\Delta r, \log_{\Delta r} \log n})$ rounds. Hence, for any algorithm with a complexity of the form $O(f(\Delta, r) + g(n))$, we have $f(\Delta, r) \in \Omega(\Delta r)$ if $g(n)$ is not too large, and in particular if $g(n) = \log^* n$ (which is the optimal asymptotic dependency on $n$ due to Linial's lower bound [FOCS'87]). Our lower bound proof is based on the round elimination framework, and its structure is inspired by a new round elimination fixed point that we give for the $\Delta$-vertex coloring problem in hypergraphs, where nodes need to be colored such that there are no monochromatic hyperedges.
  
  %Our lower bound proof is based on the round elimination framework.
  %Its structure is inspired by a new round elimination fixed point that we give for the $\Delta\cdot(r-1)$-vertex coloring problem in hypergraphs, where any two nodes that share a hyperedge need to be colored with different colors. The fixed point implies that computing such a vertex coloring requires $\Omega(\log_{\Delta r} n)$ rounds deterministically and $\Omega(\log_{\Delta r}\log n)$ rounds with randomization. Note that coloring with just one more color can be done in deterministic time $O\big(\sqrt{\Delta r \log(\Delta r)} + \log^* n\big)$ by using an existing distributed $(\Delta+1)$-vertex coloring algorithm.
  
  For the MIS problem on hypergraphs, we show that for $\Delta\ll r$, there are significant improvements over the na\"ive $O(\Delta r + \log^* n)$-round algorithm. We give two deterministic algorithms for the problem. We show that a hypergraph MIS can be computed in $O(\Delta^2\cdot\log r + \Delta\cdot\log r\cdot \log^* r + \log^* n)$ rounds. We further show that at the cost of a much worse dependency on $\Delta$, the dependency on $r$ can be removed almost entirely, by giving an algorithm with round complexity $\Delta^{O(\Delta)}\cdot\log^* r + O(\log^* n)$.  
\end{myabstract}

\section{Introduction and Related Work}
\label{sec:intro}

In the area of distributed graph algorithms, we have a network that is represented by a graph $G = (V,E)$, where nodes represent machines and edges represent communication links, and the goal is to solve some graph problem on the graph representing the network. The nodes of $G$ can communicate if they are neighbors, and in the \LOCAL model of distributed computing the computation proceeds in synchronous rounds and messages are allowed to be arbitrarily large. The complexity of a problem is the minimum number of rounds required, in the worst case, to solve it.

Given a graph $G=(V,E)$, a \emph{maximal independent set} (MIS) of $G$ is an inclusion-wise maximal set $S\subseteq V$ of nodes such that no two neighbors are in $S$ and a \emph{maximal matching} of $G$ is an inclusion-wise maximal set $M\subseteq E$ of edges such that no two adjacent edges are in $M$. Maximal matching and MIS are two of the classic local symmetry breaking problems considered in the area of distributed graph algorithms.

%The distributed computation of an MIS or a maximal matching of a given network graph $G$ are two of the classic local distributed symmetry breaking problems considered in the area of distributed graph algorithms. 
The distributed complexity of the two problems has been studied in a long line of research, see, e.g., \cite{Alon1986,Luby1986,Israeli1986,Linial1987,Awerbuch89,panconesi96decomposition,Hanckowiak1998,Hanckowiak2001,panconesi01simple,KuhnMW04,Kuhn2009,LenzenW11,BEPS,barenboim14distributed,ghaffari16improved,KuhnMW16,fischer17improved,Balliu2019,Rozhon2020,trulytight,balliurules,GGR2020,BBKOmis,hideandseek}. Researchers have investigated the computational time complexity of these two problems in two different ways: (i) as a function of the maximum degree $\Delta$ of the graph, studying complexities of the form $f(\Delta) + g(n)$, where $g$ is a much more slowly growing function than $f$; (ii) as a function of the total number of nodes $n$ in the graph, studying complexities of the form $f(n)$. The best known randomized distributed algorithms for MIS and maximal matching run in $O(\log\Delta + \log^5\log n)$ rounds~\cite{ghaffari16improved,GGR2020} and $O(\log\Delta + \log^3\log n)$ rounds~\cite{BEPS,fischer17improved}, respectively.
As a function of $n$, the fastest known deterministic algorithms for MIS and maximal matching have time complexities $O(\log^5 n)$~\cite{GGR2020} and $O(\log^2\Delta\cdot\log n) = O(\log^3 n)$~\cite{fischer17improved}, respectively. For both problems, there is a deterministic algorithm with a time complexity of $O(\Delta+\log^* n)$~\cite{panconesi01simple,barenboim14distributed}. This algorithm is optimal as a function of $n$, as we know that it is not possible to solve any of these problems in $f(\Delta) + o(\log^* n)$ rounds for any function $f$~\cite{Linial1992, Naor1991}. Further, for both problems, there are randomized time lower bounds of $\Omega\Big(\min\Big\{\frac{\log\Delta}{\log\log\Delta},\sqrt{\frac{\log n}{\log\log n}}\Big\}\Big)$~\cite{KuhnMW04, BGKO2022} and $\Omega\big(\min\big\{\Delta,\frac{\log\log n}{\log\log\log n}\big\}\big)$~\cite{Balliu2019,hideandseek}, which hold even on tree topologies. Hence, although there are certainly some interesting remaining questions, the distributed complexity of MIS and maximal matching in graphs is relatively well understood.

\paragraph{Generalization to Hypergraphs.} In this paper, we consider the distributed complexity of the natural generalization of MIS and maximal matchings to hypergraphs. A hypergraph $H=(V_H,E_H)$ consists of a set of nodes $V_H$ and a set of hyperedges $E_H\subseteq 2^{V_H}$. 
We say that a hyperedge $e \in E_H$ is incident to a node $v \in V_H$ if $v \in e$. Similarly, a node $v$ is incident to a hyperedge $e$ if $v\in e$. 
We will omit the subscript ``$H$'' from the notation if $H$ is clear from the context.
The degree $\deg_H(v)$ of a node $v$ in $H$ is its number of incident hyperedges, while the rank $\rank_{H}(e)$ of an edge $e\in E_H$ is the number of its incident nodes (i.e., $\rank_{H}(e)=|e|$). 
The rank $r_H$ of $H$ is the maximum hyperedge cardinality, i.e., $r_H=\max_{e\in E_H}|e|$. We denote with $\Delta_H$ the maximum degree in $H$.
An independent set $S\subseteq V_H$ of a hypergraph is a set of nodes such that no hyperedge is fully contained in $S$, i.e., $\forall e\in E_H\,:\,e\cap S\neq e$. A matching $M\subseteq E_H$ is a set of edges such that every node $v\in V_H$ is contained in at most one of the sets in $M$. Maximal independent sets and maximal matchings are then again inclusion-wise maximal sets with the respective property.

Generally, while graphs model pairwise dependencies or interactions, hypergraphs model dependencies or interactions between more than $2$ nodes or entities. Graphs are therefore natural to model the pairwise communication in traditional wired networks. However, in general, interactions between the participants of a distributed system can certainly be more complex. Hypergraphs have a  richer structure and are generally less understood than graphs (not only in the distributed context). Various problems that can naturally be formulated as problems on hypergraphs have been studied and found applications in the context of distributed graph algorithms (e.g., \cite{CzygrinowHSW12,KuttenNPR14,lotker15,GhaffariKM17-slocal,FischerGK17,spaa21_tokendropping}).
%\fatodo{Is something important missing here? Esp.\ something that is not cited later?} 
In this context, one can also mention the beautiful line of work that uses methods from algebraic topology to prove distributed impossibility results and also develop distributed algorithms, e.g., \cite{HerlihyS99,SaksZ00,CastanedaR12,topology_book,CastanedaFPRRT21}. There, simplicial complexes (which from a combinatorial point of view are hypergraphs with the property that the set of hyperedges forms a downward-closed family of subsets of nodes) are used to express the development of the state of a distributed system throughout the execution of a distributed algorithm. We hope that a better understanding of the fundamental limitations and possibilities of distributed algorithms in hypergraphs will lead to insights that prove useful in the design of future distributed systems. In the following, we discuss some concrete reasons why hypergraph MIS and maximal matching specifically are interesting and worthwhile objects to being studied from a distributed algorithms point of view.

\paragraph{Distributed MIS in Hypergraphs.} Assume that we have a shared resource so that in each local neighborhood in the network, only a limited number of nodes can concurrently access the resource. The shared resource could for example be the common communication channel in a wireless setting that allows some network-level coding. Maybe a node can still decode a received linear combination of messages, as long as the signal consists at most $\kappa$ messages for some $\kappa>1$. The possible sets of nodes that can access the channel concurrently can then be expressed by the set of independent sets of some hypergraph. A further direct application of the MIS problem in hypergraphs is the problem of computing a minimal set cover and the problems of computing minimal dominating sets with certain properties~\cite{HarrisMPRS16}. Given a set cover instance, we can define a hypergraph $H$ in the natural way by creating a node for each set and hyperedge for each element consisting of the sets that contain this element. The vertex covers of $H$ correspond to the set covers of our set cover instance and the set of minimal vertex covers (and thus the set of minimal set covers) are exactly the complements of the set of maximal independent sets of $H$. To see this, note that clearly the component of a vertex cover (at least one node per hyperedge) is an independent set (not all nodes per hyperedge) and vice versa. Clearly minimality of a vertex cover implies maximality of the independent set and vice versa. In \cite{HarrisMPRS16,KuttenNPR14}, it is further shown that the minimal set cover problem can be used as a subroutine to compute sparse minimal dominating sets\footnote{A minimal dominating set is called sparse if the average degree of the nodes in the dominating set is close to the average degree of the graph.} and minimal connected dominating sets.

From a theoretical point of view, the distributed MIS problem in hypergraphs is also interesting because many of the techniques that work efficiently in the graph setting seem to fail or become much less efficient. It is for example not even clear how to obtain a randomized hypergraph MIS algorithm that is similarly efficient and simple as Luby's algorithm~\cite{Alon1986,Luby1986} for graphs. To the best of our knowledge, the first paper to explicitly study the problem of computing an MIS in a hypergraph is by Kutten, Nanongkai, Pandurangan, and Robinson~\cite{KuttenNPR14}. In the \LOCAL model (i.e., with arbitrarily large messages), it is relatively straightforward to compute an MIS by using a network decomposition. By using the network decomposition algorithm of \cite{LinialS93}, one in this way obtains  a randomized $O(\log^2 n)$-round algorithm and by using the recent deterministic network decomposition of \cite{Rozhon2020,GGR2020}, one obtains a deterministic $O(\log^5 n)$-round algorithm. The focus of \cite{KuttenNPR14} was therefore to obtain efficient algorithms in the \CONGEST model, i.e., algorithms that only use small messages.\footnote{Note that there is more than one way in which one can define the \CONGEST model for hypergraphs. As we do not focus on \CONGEST algorithms in this paper, we refer to \cite{KuttenNPR14} for a discussion of this issue.} In \cite{KuttenNPR14}, it is shown that an MIS in a hypergraph of maximum degree $\Delta$ and rank $r$ can be computed in the \CONGEST model in time $O(\log^{(r+4)!+4} n)$ if $r$ is small and in time $\Delta^{\eps}\cdot\poly\log n$ for any constant $\eps>0$. This upper bound has been later improved to $(\log n)^{2^{r+3} + O(1)}$ in \cite{phmis2}. Distributed \CONGEST algorithm for a closely related problem have also been studied in \cite{KuhnZheng}. All the mentioned bounds are obtained by randomized algorithms. In the present paper, we focus on deterministic algorithms and we focus on the dependency on $\Delta$ and $r$, while keeping the dependency on $n$ to $O(\log^* n)$ and thus as small as possible (due to Linial's $\Omega(\log^* n)$ lower bound~\cite{Linial1992}).
While there is not much work on distributed algorithms for computing an MIS in a hypergraph, the problem has been studied from a parallel algorithms point of view, see, e.g.~\cite{KarpUW88,BeameL90,Kelsen92,LuczakS97,phmis,phmis2}.

\paragraph{Distributed Maximal Matching in Hypergraphs.} First note that the set of matchings on a hypergraph $H=(V_H,E_H)$ is equal to the set of independent sets on the line graph of $H$, that is, on the graph containing a node for every hyperedge in $E_H$ and an edge for every intersecting pair of hyperedges in $E_H$. The maximal matching problem on hypergraphs of maximum degree $\Delta$ and rank $r$ is therefore a special case of the MIS problem on graphs of maximum degree at most $r(\Delta-1)$. Maximal matchings in hypergraphs have several applications as subroutines for solving standard problems in graphs, as we discuss next.

The maximum matching problem in graphs can be approximated arbitrarily well by using a classic approach of Hopcroft and Karp~\cite{HopcroftK73}. When starting with some matching and augmenting along a maximal set of disjoint shortest augmenting paths, one obtains a matching for which the shortest augmenting path length is strictly larger. As soon as the shortest augmenting path length is at least $2/\eps$, the matching is guaranteed to be within a factor $(1-\eps)$ of an optimal matching. The step of finding a maximal set of disjoint shortest augmenting paths can be directly interpreted as a maximal matching problem in the hypergraph defined by the set of shortest augmenting paths. In the distributed context, this idea has been used, e.g., in \cite{czygrinow03,lotker15,FischerGK17,FOCS18-derand,Harris20}. In \cite{HougardyV06,GhaffariKMU18}, it is further shown how maximal hypergraph matching can also be used in a similar way to obtain fast distributed approximations for the weighted maximum matching problem. In \cite{GhaffariKMU18}, this was then used to deterministically compute an edge coloring with only $(1+\eps)\Delta$ colors. In \cite{FischerGK17}, it was further shown that the problem of computing a $(2\Delta-1)$-edge coloring in graphs can be directly reduced to the problem of computing a maximal matching in a hypergraph of rank $3$. Additionally, \cite{FischerGK17} also shows how to compute an edge orientation of out-degree $(1+\eps)$ times the arboricity of a given graph by reducing to maximal matching in low-rank hypergraphs. Finally, when viewing a hypergaph as a bipartite graph between nodes and hyperedges, the maximal hypergraph matching problem is related to the MIS problem on the square of this bipartite graph and understanding the complexity of maximal matching in hypergraphs can be a step towards understanding the complexity of MIS in $G^2$.

Since the maximal hypergraph matching problem is a special case of the MIS problem in graphs, there are quite efficient randomized algorithms for the problem and the focus in the literature therefore so far has been on developing deterministic distributed algorithms for maximal matchings in hypergraphs~\cite{FischerGK17,FOCS18-derand,Harris20}. Prior to the new efficient deterministic network decomposition algorithm of Rozho\v{n} and Ghaffari~\cite{Rozhon2020}, those papers lead to the first deterministic distributed polylogarithmic-time $(2\Delta-1)$-edge coloring algorithms. The best known deterministic distributed algorithm for computing a maximal matching in a hypergraph of maximum degree $\Delta$ and rank $r$ is due to Harris~\cite{Harris20} and it has a time complexity of $\tilde{O}\big((r^2\log\Delta + r\log^2\Delta)\cdot\log n\big)$, where $\tilde{O}(x)$ hides polylogarithmic factors in $x$.

% some old notes: hypergraphs are a generalization of graphs, problems
% on hypergraphs seem to appear naturally in the distributed world
% when studying problems on graphs (cite what? hypergraph MM by
% Fischer/Ghaffari? this SLOCAL-complete alternative to network
% decomposition? what else? how to fit in Harris?), but are also
% interesting in their own right, in particular in the disguise of
% bipartite problems (cite what? binary LCLs, semi-matching and
% related, what else?)

\paragraph{The Trivial Algorithm.}
Both an MIS and a maximal matching of an $n$-node hypergraph $H=(V_H,E_H)$ with maximum degree $\Delta$ and rank $r$ can be computed in time $O(\Delta r + \log^* n)$ in a straightforward way.
The trivial algorithm for hypergraph MIS proceeds as follows.
We first color the nodes in $V_H$ with $\Delta(r-1)+1$ colors such that no two nodes that share a hyperedge are colored with the same color. Such a coloring can be computed in $O(\Delta r + \log^* n)$ rounds, e.g., by using the algorithm\footnote{We could achieve a smaller runtime of $O(\sqrt{\Delta r \log(\Delta r)} + \log^* n)$ rounds for computing the coloring by using a state-of-the-art distributed $(\mathit{degree}+1)$-coloring algorithm~\cite{fraigniaud16local,BEG18,MausTonoyan20}; however, as, subsequently, we iterate through the obtained color classes, this would not change the overall asymptotic runtime.} of \cite{barenboim14distributed}.
Then, we iterate through the $O(\Delta r)$ color classes and greedily add nodes to the solution if the addition does not violate the maximality condition.
The trivial algorithm for hypergraph maximal matching proceeds analogously, where we color hyperedges instead of nodes, and there is a different maximality condition.

\paragraph{Some Notation.}
Before stating our contributions, we briefly discuss some terminology regarding hypergraphs.
A hypergraph is called \emph{$r$-uniform} if all its hyperedges are of cardinality exactly $r$ and it is called \emph{$\Delta$-regular} if all nodes have degree $\Delta$. A hypergraph is called \emph{linear} if any two hyperedges intersect in at most one node. Further, a \emph{hypertree} is a connected hypergraph $H=(V_H,E_H)$ such that there exists an underlying tree $T$ on the nodes $V_H$ for which every hyperedge consists of the nodes of a connected subtree of $T$. Finally, there is a natural representation of a hypergraph $H=(V_H,E_H)$ as a \emph{bipartite graph} consisting of the nodes $V_H \cup E_H$ and an edge between $v\in V_H$ and $e\in E_H$ if and only if $v\in e$. We refer to this bipartite graph as the \emph{bipartite representation of $H$}. Note that $H$ is a linear hypertree if and only if its bipartite representation is a tree. We say that a linear hypertree is $\Delta$-regular if every node is either of degree $\Delta$ or of degree $1$.

\subsection{Our Contributions}
While there is a clear relation between \emph{graph} MIS and \emph{hypergraph} maximal matching, or between \emph{graph} MIS and \emph{graph} maximal matching (it is possible to use an algorithm for the first problem to solve the second ones, as discussed earlier), observe that hypergraph maximal matching and hypergraph MIS cannot be easily compared: in contrast to the situation that we have in graphs, where it is possible to use an algorithm for MIS to solve maximal matching in the same asymptotic runtime, on hypergraphs no reduction of this kind is known.
In fact, a priori it is not clear at all which problem is the easier of the two.

In order to better understand our results, it is useful to compare the complexities of hypergraph MIS and hypergraph maximal matching with the complexity of the trivial algorithm that solves these two problems.
As discussed above, the trivial algorithm has a runtime of $O(\Delta r + \log^* n)$ rounds, i.e., it has a very low dependency on $n$, but is quite slow in terms of $\Delta$ and $r$.
The question that we study in our work is whether the trivial algorithm is optimal or whether it is possible to improve on it.
\vspace{2pt}
\begin{mdframed}[backgroundcolor=gray!20, topline=false, rightline=false, leftline=false, bottomline=false] 
	\textbf{Question 1.}
	\noindent Is it possible to solve hypergraph maximal matching and/or hypergraph MIS in $O(\log^* n)+ o(\Delta r)$ rounds?
\end{mdframed}
\vspace{2pt}

Informally, we show the following.
\begin{itemize}
	\item For hypergraph maximal matching, the trivial algorithm is best possible, unless we spend a much higher dependency on $n$.
	\item For hypergraph MIS, it is possible to improve on the trivial algorithm when $r \gg \Delta$.
%	\item For hypergraph MIS, it is possible to improve on the trivial algorithm, at least when $r \gg \Delta$.
\end{itemize}
In the following, we will discuss our results and their relation to the trivial algorithm more formally.

\paragraph{Maximal Matching Lower Bound.}
As our main result, we show that for the hypergraph maximal matching problem, the trivial algorithm---which is simply a na\"{i}ve implementation of the sequential greedy algorithm---is best possible.

\begin{restatable}{theorem}{matchinglower}\label{thm:matchinglower}
	Assume that $\Delta\geq 2$ and $r\geq 2$.
	Then any deterministic distributed algorithm in the \LOCAL model for computing a maximal matching in hypergraphs with maximum degree $\Delta$, rank $r$, and $n$ nodes requires $\Omega\big(\min\big\{\Delta r, \log_{\Delta r} n\big\}\big)$ rounds.\footnote{As lower bounds in several parameters are sometimes difficult to understand, we would also like to provide, as an example, the exact quantification of this statement, which is as follows: There is a constant $c$ such that, for any $\Delta \ge 2$ and $r \ge 2$, and any deterministic algorithm $\fA$, there are infinitely many hypergraphs with maximum degree $\Delta$, rank $r$, and $n$ nodes on which $\fA$ has a runtime of at least $c \big(\min\big\{\Delta r, \log_{\Delta r} n\big\}\big)$ rounds.} Any randomized such algorithm requires at least $\Omega\big(\min\big\{\Delta r, \log_{\Delta r} \log n\big\}\big)$ rounds. Moreover, our lower bounds hold already on $\Delta$-regular $r$-uniform linear hypertrees.
\end{restatable}

We remark that, while in general hypergraphs the number of hyperedges can be much larger than the number of nodes, in linear hypertrees (which is the case where our lower bound applies) the sum of the nodes and the hyperedges is linear in $n$.
%\todo{There is one technical issue that the authors should address. They describe lower bounds for deterministic algorithms, as well as randomized algorithms which succeed with high probability, i.e. probability 1 - 1/poly(n). In many cases, hypergraph algorithms are used as primitives for graph algorithms. In order to get a deterministic *graph* algorithm, it suffices to use a randomized hypergraph algorithm with failure probability $2^{-n^2}$; by contrast, a randomized hypergraph algorithm would have to have a lower failure probability $2^{-n^r}$ to be converted into a deterministic hypergraph algorithm.  Sometimes, it is more profitable to build deterministic graph algorithms in this way. The authors should describe lower bounds for round complexity in order to get failure probability $\delta$, where $\delta$ is a much lower value than 1/poly(n).  Presumably, the runtime would have a $\log \log (1/\delta) + \log^* n$ term instead of $\log \log n$?  I think that adding this would be needed to fully understand the limitations of the hypergraph algorithms.}

\Cref{thm:matchinglower} shows that, for hypergraph maximal matching, we answer Question 1 negatively: the theorem essentially states that we need to spend $\Omega(\Delta r)$ rounds, unless we spend a much higher dependency on $n$ than the one of the trivial algorithm, i.e., $\Omega(\log_{\Delta r} n)$ for deterministic algorithms and $\Omega(\log_{\Delta r} \log n)$ for randomized ones. In other words, the trivial greedy algorithm is optimal if we insist on having a small dependency on $n$. Observe that there exist algorithms for hypergraph maximal matching that can beat the $\Omega(\Delta r)$ lower bound that we provide, but they all spend substantially more than $\log^*$ as a function of $n$, and our results show that this is indeed necessary. For example, the algorithm by Harris \cite{Harris20} has a dependency on $\Delta$ that is only polylogarithmic, but it also has a multiplicative $\log n$ dependency.  

In fact, as a corollary, we obtain that the complexity of any algorithm for hypergraph maximal matching, when expressed solely as a function of $n$, must be large. This is known by prior work for $r = 2$, and we get that the same holds even when $r = \Delta$. This can be observed by setting $ r = \Delta = \sqrt{\log n / \log \log n}$ in the deterministic lower bound and $ r = \Delta = \sqrt{\log \log n / \log \log \log n}$ in the randomized one.
\begin{corollary}
	Any deterministic distributed algorithm in the \LOCAL model for computing a maximal matching in $n$-node hypergraphs requires at least $\Omega\big(\frac{\log n}{\log\log n}\big)$ rounds. Any randomized such algorithm requires at least $\Omega\big( \frac{\log \log n}{\log\log\log n}\big)$ rounds. Moreover, our lower bounds hold already on regular uniform linear $n$-node hypertrees that satisfy $r = \Delta$.
\end{corollary}

\paragraph{Hypergraph Coloring.}
 On a high level, the proof of \Cref{thm:matchinglower} is based on a lower bound proof that we provide for a seemingly unrelated problem, namely, hypergraph coloring. In fact, as a byproduct, we also obtain lower bounds for variants of this problem. There are two natural ways to generalize the standard coloring problem to hypergraphs:
 \begin{itemize}
 	\item Color the nodes of the hypergraph with $c$ colors such that, for every hyperedge of rank at least $2$, it holds that at least $2$ incident nodes have different colors. This is the standard definition of hypergraph coloring.
 	\item Color the nodes of the hypergraph with $c$ colors such that, for every hyperedge, all incident nodes have different colors. This variant is sometimes called \emph{strong coloring} \cite{strongcoloring}, \emph{colorful coloring}, or \emph{rainbow coloring}\footnote{In the literature, the term rainbow coloring is also used to refer to a different, harder variant of coloring.} \cite{uniquemaximum}. We will refer to this variant as \emph{colorful}.
 \end{itemize}
 For the first problem, we show that it cannot be solved fast when $c \le \Delta$, and for the second one, we show that it cannot be solved fast when $c \le \Delta(r-1)$. While these results can also be shown with non-trivial reductions from the known hardness of a different problem called sinkless orientation, the proof that we present for these lower bounds is essential as a building block for the proof of our hypergraph maximal matching lower bound (and, in fact, can be seen as a simplified version of that proof).

 \begin{theorem}
 	In the \LOCAL model, the $\Delta$-hypergraph coloring and the $(r-1)\Delta$-hypergraph colorful coloring require $\Omega(\log_{\Delta r} n)$ rounds for deterministic algorithms and $\Omega(\log_{\Delta r} \log n)$ rounds for randomized ones.
 \end{theorem}

\paragraph{MIS Upper Bounds.}
In graphs, it is known that both maximal matching and MIS require a number of rounds that is linear in $\Delta$ \cite{Balliu2019,hideandseek}, unless we spend a high dependency on $n$. Hence, a natural question to ask is whether, like in the case of hypergraph maximal matching, also hypergraph MIS requires $\Omega(\Delta r)$ rounds if we keep the dependency on $n$ as small as possible.

We show that hypergraph MIS behaves differently than hypergraph maximal matching: in fact, Question 1 can be answered affirmatively for a certain parameter range.
In particular, we consider the case where $r\gg\Delta$ and design algorithms with a round complexity of the form $O(\log^* n) + f(\Delta, r)$ where $f(\Delta,r)=o(\Delta r)$ if $\Delta$ is sufficiently small compared to $r$. More concretely, we prove the following two results.

\begin{restatable}{theorem}{MISalgI}\label{thm:mis-slowInDelta}
  In the LOCAL model, the MIS problem on $n$-node hypergraphs of maximum degree $\Delta$ and rank $r$ can be solved in $O(\Delta^2\log r + \Delta \log r \log^* r+ \log^* n)$ deterministic rounds.
\end{restatable}

\begin{restatable}{theorem}{MISalgII}\label{thm:MIS-independent-of-r}
  In the LOCAL model, the MIS problem on $n$-node hypergraphs of maximum degree $\Delta$ and rank $r$ can be solved in $2^{O(\Delta\log\Delta)}\log^* r + O(\log^* n)$ deterministic rounds.
\end{restatable}

Note that the bound in the second theorem is almost independent of the rank $r$.
In particular, Theorem~\ref{thm:MIS-independent-of-r} implies that for bounded-degree hypergraphs, an MIS can be computed in time $O(\log^* n)$, even if the rank $r$ is not bounded. Observe that, while the dependency on $\Delta$ is high, by prior work we know that it has to be at least linear \cite{Balliu2019,hideandseek}.

We further note that the bounds that we provide are incomparable to the upper bounds that were obtained in \cite{KuttenNPR14}, where the authors tried to optimize the round complexity as a function of $n$ alone. The runtimes of the algorithms in \cite{KuttenNPR14} are all at least polylogarithmic in $n$. In our algorithms, we keep the dependency on $n$ to $O(\log^* n)$ and thus as small as it can be. Our algorithms are faster than the $\poly\log n$-time algorithms based on network decomposition in the realistic scenario where $\Delta$ and $r$ are much smaller than $n$.

\paragraph{Open Questions.}
Our work raises the following natural open question.
Is it possible to solve hypergraph MIS in just $O(\Delta + \log^* n)$ rounds, independently of $r$?
Given the mentioned lower bounds on graphs, such a complexity would be optimal.
In this regard, we show that there is a variant of hypergraph coloring that, if solved fast, would provide such an algorithm. This variant is typically called $c$-unique-maximum coloring \cite{uniquemaximum}, and requires to color the nodes of a hypergraph with $c$ colors such that, for each hyperedge $e$, the maximum color appearing at the nodes incident to $e$ occurs exactly once among the incident nodes. We show that finding an $O(\Delta)$-unique-maximum coloring in $O(\Delta + \log^* n)$ rounds would allow us to solve hypergraph MIS optimally. While our lower bounds hold for different variants of hypergraph coloring, they do not hold for this variant.

\paragraph{Technical Remarks.}
In order to prove our lower bounds, we make use of the round elimination framework, which provides a general outline for proving lower bounds in the distributed setting.
Round elimination has first been applied in \cite{Brandt2016} to obtain lower bounds on computing a sinkless orientation and a $\Delta$-coloring of a graph.
In 2019, \cite{Brandt2019} showed that the round elimination framework can be applied to almost any locally checkable problem\footnote{Roughly speaking, a locally checkable problem is a problem that can be defined via local constraints. A standard example of locally checkable problems are proper coloring problems, which can be described by the local constraint that the two endpoints of any edge must receive different colors.}, at least \emph{in principle}; however to obtain a good lower bound for some concrete locally checkable problem, it is necessary to overcome a certain set of challenges that is individual for each problem or problem family.
While for some problems---such as $(\Delta + 1)$-vertex coloring or $(2 \Delta - 1)$-edge coloring---these challenges seem far beyond the reach of current techniques, they have been overcome for a number of other important problems~\cite{Balliu2019,trulytight,balliurules,BBKOmis,hideandseek,binary}.
Each of these results has been accompanied by a new key idea that succeeded in making the round elimination technique applicable also \emph{in practice} for the considered problem(s); often the key idea unlocked the applicability of round elimination not only for a single problem, but for several problems at once, a whole problem family, or a new setting.
In a similar vein, we prove the first substantial lower bounds for important hypergraph problems (that do not simply follow from lower bounds on graphs) by showing how to make these problems susceptible to a round elimination type of approach.
We hope that our techniques can serve as a starting point for unlocking the round elimination technique more generally for hypergraph problems.

\subsection{Road Map}

We start, in \Cref{sec:approach}, by providing a high-level overview of our lower bound proof, highlighting the challenges and the new ingredients that we bring to the table.

%Round elimination was afterwards introduced as a generic method for general locally checkable problems by Brandt~\cite{Brandt2019}. Since then, round elimination has proved to be extremely successful for proving lower bounds in the \LOCAL model \cite{Balliu2019,trulytight,balliurules,BBKOmis,hideandseek,binary}. There are many challenges in using this technique, and we start, in \Cref{sec:approach}, by providing a high-level overview of our lower bound proof, highlighting the new ingredients that we bring to the table.

We then proceed, in \Cref{sec:defs}, by first defining the model of computing, and then by formally describing the round elimination technique. In this section, we also provide an example of the application of this technique.

In \Cref{sec:fixedpoint}, we prove lower bounds for hypergraph coloring and some of its variants. Hypergraph coloring seems to be, at first sight, unrelated to hypergraph maximal matching. We show the connections between these two problems in \Cref{sec:approach}, where we also explain why the proof that we present for the hardness of hypergraph coloring can be seen as a simplified version of the proof for hypergraph maximal matching.  

The hypergraph maximal matching lower bound is presented in \Cref{sec:lbmm}. The complexity of this proof is higher than the one for coloring, and requires some heavy notation. We will start by describing how it is connected to the proof for hypergraph coloring, by informally describing how it can be seen as an extension of the one for hypergraph coloring.

In \Cref{sec:ub}, we present upper bounds for the hypergraph MIS problem. We show that, while the trivial algorithm is able to solve the problem in $O(\Delta r + \log^* n)$ rounds, if we allow a larger dependency on $\Delta$ we can almost entirely remove the dependency on $r$. In fact, in hypergraphs of maximum degree $\Delta = O(1)$ and very large rank $r$, we show that it is possible to solve the hypergraph MIS problem in just $O(\log^* n)$ rounds. 

We conclude, in \Cref{sec:open}, by stating some open questions. We show that hypergraph MIS could be solved optimally if we are provided with an $O(\Delta)$-unique-maximum coloring, and we leave as an open question determining the complexity of this variant of coloring.

\section{High-Level Overview of the Lower Bound}\label{sec:approach}
Due to the fact that round elimination lower bound proofs are typically technically complex and notation-heavy, which hides the underlying conceptual ideas, we will give a detailed explanation of the round elimination framework, its fundamental issue, approaches of previous work to take care of this issue and our conceptual contributions in the following.
While we have to resort to painting a very rough picture in places (as giving all details of the previous approaches would be far beyond the scope of this work), we hope that our high-level overview will make the recent developments regarding distributed lower bounds accessible to a broader audience while highlighting how our approach extends the lower bound tool box.

\paragraph{The Round Elimination Technique.}
On a high level, round elimination is a technique that can be used to prove lower bounds on the time required to solve (on trees or hypertrees) locally checkable problems, i.e., problems that can be defined by specifying some constraints that each node and hyperedge must satisfy. In the round elimination framework, a problem $\Pi$ is defined by listing what labelings, or \emph{configurations}, are allowed around the nodes and on the hyperedges.
More precisely, each node $v$ has to output some output label for each pair $(v,e)$ such that $e$ is a hyperedge containing $v$; the output is considered to be correct if for each node, the configuration of output labels on the pairs with first entry $v$ is contained in an explicitly given list of allowed ``node configurations'', and for each hyperedge $e$, the configuration of output labels on the pairs with second entry $e$ is contained in an explicitly given list of allowed ``hyperedge configurations''.

The round elimination technique is a mechanical procedure $\re$ that can be applied to $\Pi$ to obtain a new problem $\Pi' = \re(\Pi)$ that, under mild assumptions regarding the considered graph class, is exactly one round easier than $\Pi$. Suppose we prove that $\Pi' = \re(\Pi)$ cannot be solved in $0$ rounds of communication: we obtain that solving $\Pi'$ requires at least $1$ round, and hence that solving $\Pi$ requires at least $2$ rounds. The obtained problem $\Pi'$ is also always locally checkable, and hence we can apply this technique again to obtain some problem $\Pi''$ guaranteed to be exactly one round easier than $\Pi'$. If we prove that also this problem cannot be solved in $0$ rounds, we obtain a lower bound of $3$ rounds for $\Pi$.

\paragraph{The Issue with Round Elimination.}
While the high-level outline given above suggests that obtaining lower bounds via round elimination is an entirely mechanical task, this is unfortunately not the case due to a fundamental issue in the framework.
A simple example to explain this issue is the following. Let $\Pi$ be the problem of $3$-coloring a cycle, which is clearly locally checkable.
The constraints of this problem can be specified by a list of allowed configurations using $3$ output labels (the colors). If we apply round elimination to this problem, we automatically obtain a new problem $\Pi' = \re(\Pi)$ that can be described by using $16$ labels. If we apply round elimination again, we get a problem with thousands of labels. While in theory we could continue forever, in practice we cannot, because it is not feasible to compute the next problem (and even less to determine whether it can be solved in $0$ rounds or not).
In general, the downside of this techique is that the description of $\Pi'$ can be doubly exponentially larger than the one of $\Pi$.

\paragraph{Decomposing a Problem.}
In \cite{hideandseek}, a technique has been presented that sometimes allows us to take care of this issue.
In the following, we describe this technique for the case of the MIS problem (on graphs). Let $\re^0(\Pi) = \Pi$ and $\re^{i+1}(\Pi) = \re(\re^i(\Pi))$.

It has been observed already in \cite{balliurules} that if we start from $\Pi$ equal to the MIS problem and apply the round elimination technique iteratively $k$ times, the result $\re^k(\Pi)$ can be essentially decomposed into three parts:
\begin{enumerate}
	\item the \emph{original} problem, i.e., MIS,
	\item a \emph{natural} part of small size, corresponding to the $k$-coloring problem, and
	\item an \emph{unnatural} part whose description is of size roughly equal to a power tower of height $k$ and that cannot be easily understood.
\end{enumerate}
MIS is not the only problem that behaves in this way under round elimination; in fact, many problems admit, after applying round elimination for $k$ steps, a decomposition into the original problem, some natural part that is easy to understand (and often unrelated to the original problem), and some unnatural part of much larger size (where the latter two depend on $k$).
%MIS is not the only problem that behaves in this way under round elimination, and in fact many problems behave similarly under this technique. In fact, many problems satisfy that, if we apply round elimination for $k$ steps, we obtain a problem that can be decomposed into three parts: the original problem, some natural part that is easy to understand (and often unrelated to the original problem), and some unnatural part of much larger size.

The main idea of the previous works that use the round elimination technique is to relax the problems obtained at each step with the goal of replacing the unnatural part with something much smaller. In this context, relaxing a problem means adding further configurations to the lists collecting the allowed configurations for the nodes or the hyperedges, at the cost of making the problem potentially easier to solve. Observe that, if a problem is relaxed too much, then it is not possible to obtain strong lower bounds, because we would soon reach a problem that is $0$-round solvable. Hence, the typical goal is to find relaxations that make the description of the problem much smaller while still not making the problem much easier. Note that, effectively, the description of a problem can indeed become smaller by \emph{adding} allowed configurations: by allowing more (in a suitable way), a substantial number of allowed configurations can be shown to become useless and hence can be ignored.

\paragraph{Previous Approaches.}
In order to prove tight lower bounds for ruling sets (and other interesting problems), in \cite{hideandseek}, the third part in the decomposition of the problems has been handled as follows.

Consider the $k$-coloring problem, for $k \le \Delta$. In some sense, this problem is known to be much harder than MIS, since, unlike MIS, it cannot be solved in $O(f(\Delta) + \log^* n)$ rounds for any function $f$. We call problems that cannot be solved in $O(f(\Delta) + \log^* n)$ rounds for any function $f$ \emph{hard}. If we apply round elimination to the $k$-coloring problem, we essentially get something very similar to the third part presented above.
But for some of the hard problems, we have techniques that allow us to deal with them, that is, finding relaxations of the problem that result in a so-called \emph{fixed point}. Fixed point problems are special: by applying the round elimination technique to a fixed point problem $\Pi$ we obtain the problem itself, i.e., $R(\Pi) = \Pi$. This may look like a contradictory statement: it is not possible for a problem to be exactly one round easier than itself. The reason why this is not an actual contradiction is that the condition that makes the round elimination statement hold stops working when the complexity is $\Omega(\log_\Delta n)$. In fact, it is possible to prove that a fixed point directly implies a lower bound of $\Omega(\log_\Delta n)$ for the problem. 

The point of relaxing the $k$-coloring problem to a fixed point before applying the round elimination technique is the following: if we apply round elimination to the $k$-coloring problem itself, we obtain a lot of configurations that make the obtained problem hard to understand, while if we apply round elimination to a fixed point relaxation of the $k$-coloring problem, we obtain the problem itself, and no additional allowed configurations. Observe that, for problems solvable in $O(f(\Delta) + \log^* n)$ rounds, like MIS, fixed point relaxations cannot exist, since they would give a lower bound higher than the upper bound for small values of $\Delta$.

The idea of \cite{hideandseek} is to \emph{embed} the fixed point of the $k$-coloring problem into the sequence obtained from MIS, in the following sense. We start from MIS, and we apply round elimination iteratively. At each step $k$, we relax the obtained problem by adding all the configurations allowed by the fixed point relaxation of the $k$-coloring problem. It has been shown in \cite{hideandseek} that thereby we obtain a problem sequence in which the number of allowed configurations is always bounded by $O(2^k)$, and hence much better than a power tower of height $k$. In other words, the third part of the problem (the unnatural part) is replaced by the configurations that we need to add to the $k$-coloring problem to make it a fixed point, of which there are just $O(2^k)$. The reason why this approach works for MIS seems to be related to the fact that all the unnatural configurations present at step $k+1$ are generated (by $\re$) by starting from the $k$-coloring configurations present at step $k$.

The take-home message of the approach of \cite{hideandseek} is the following: try to decompose the problem into the three parts explained above, and try to find a fixed point relaxation for the natural part; it may then happen that the unnatural part just disappears.

By following this approach, the authors of \cite{hideandseek} managed to prove tight lower bounds for MIS, ruling sets, and many other interesting natural problems. Also, they asked if there are other interesting problems for which this technique is applicable (see Open Problem 2 in \cite{hideandseek}), and we answer this question affirmatively.

\paragraph{A Colorful Matching.}
We extend the applicability of the technique to hypergraphs by solving a number of intricate challenges introduced by the more complex setting.
In particular, we show that while the explained technique cannot be applied directly, by combining it with a set of new ingredients we obtain substantial lower bounds for the hypergraph maximal matching problem.
The good news from the perspective of the general outline is that we can still decompose $\re^k(\Pi)$, where $\Pi$ is the hypergraph maximal matching problem, into three parts, by making use of $k$-hypergraph colorful coloring:
%We show that a similar technique can be used to prove lower bounds for the hypergraph maximal matching problem.
%The explained technique cannot be applied directly, and we have to bring new ingredients. In particular, we can apparently still decompose $\re^k(\Pi)$, where $\Pi$ is the hypergraph maximal matching, into three parts:
\begin{enumerate}
	\item The \emph{original} problem,
	\item a \emph{natural} part: the $k$-hypergraph colorful coloring problem, and
	\item an \emph{unnatural} part, of size roughly equal to a power tower of height $k$, without any easily discernible useful structure.
\end{enumerate}
By following the approach of \cite{hideandseek}, in order to prove a lower bound of $\Omega(\Delta r)$, it is necessary to first find a fixed point relaxation for the $k$-hypergraph colorful coloring problem, for $k$ that can be as large as $\Omega(\Delta r)$. 

\paragraph{Challenges and New Ingredients.}
Two labels $\ell_1$ and $\ell_2$ are called \emph{equivalent} when, for any allowed configuration containing $\ell_1$, if we replace $\ell_1$ with $\ell_2$, we obtain a configuration that is allowed, and for any configuration containing $\ell_2$, if we replace $\ell_2$ with $\ell_1$, we also obtain a configuration that is allowed. If two labels are equivalent, we can actually discard all configurations containing one of the two labels, still obtaining a problem that is not harder than the original one.  Relaxing a problem in order to make $\ell_1$ equivalent to $\ell_2$ is called \emph{label merging}, and it is a powerful technique that can be used to reduce the number of labels. 
In previous works, relaxations are typically done in two ways:
\begin{itemize}
	\item When proving lower bounds for problems solvable in $O(f(\Delta) + \log^* n)$ rounds, it is typically the case that, by applying round elimination, many new labels are obtained, and to keep the problem small, a lot of label merging is performed.
	\item When proving lower bounds for hard problems, configurations (possibly containing new labels) are added in order to transform the problem into a fixed point. Label mergings are not performed.
\end{itemize}
%In general, it usually does not make sense to merge labels of the given original problem: it only makes sense to merge labels obtained after applying round elimination. The reason is that otherwise there would be a more succinct problem that we could start from.

Unfortunately, the latter outline for finding a fixed point relaxation does not seem to work for $\Delta(r-1)$-hypergraph colorful coloring. Via a new approach, we nevertheless show how to find a relaxation of the $\Delta(r-1)$-hypergraph colorful coloring problem that is a fixed point, but it comes with two issues: the fixed point does not resemble the $\Delta(r-1)$-hypergraph colorful coloring problem at all, and the approach of \cite{hideandseek}, applied to it, does not work. 

We obtain the fixed point as follows. We first find a fixed point relaxation for the $\Delta$-hypergraph (non-colorful) coloring problem, and then show that there exists a $0$-round algorithm that converts a $\Delta(r-1)$-hypergraph colorful coloring into a proper $\Delta$-hypergraph coloring. This algorithm can be interpreted in an interesting way: what it does is mapping different colors into a single one, essentially showing that the $\Delta$-hypergraph (non-colorful) coloring problem can be obtained from $\Delta(r-1)$-hypergraph colorful coloring by performing label merging. 
This observation allows us to obtain a fixed point for the $\Delta(r-1)$-hypergraph colorful coloring problem: start from $\Delta(r-1)$-hypergraph colorful coloring, merge some labels to obtain $\Delta$-hypergraph coloring, then add configurations to this new problem in order to make it a fixed point.
We provide a more detailed overview of our approach for obtaining the fixed point in Section~\ref{sec:fixedpoint}.

We now turn to explaining the challenges and our solutions for obtaining the hypergraph maximal matching lower bound with the help of the aforementioned fixed point.

Recall the aforementioned decomposition of problems into the original problem, a natural colorful coloring part, and an unnatural part of large description size.
In the hypergraph maximal matching problem, by applying round elimination, we essentially get an additional new color in the coloring part per step while the size of the unnatural part grows like a power tower. We want to be able to perform $\Omega(\Delta r)$ steps of round elimination before reaching a $0$-round solvable problem, but we cannot let the number of labels grow like a power tower because it would then be infeasible to actually understand the problem sequence. In order to keep the number of labels reasonably small throughout the problem sequence without compromising the lower bound quality, we have to perform relaxations that do not make the problems too easy, as otherwise we would obtain a $0$-round solvable problem after $o(\Delta r)$ steps, which would give a worse lower bound.

%Recall the aforementioned decomposition of problems into the original problem, a natural coloring part, and an unnatural part of large description size.
%In the hypergraph maximal matching problem, by applying round elimination, at each step we essentially get an additional new color in the coloring part plus a lot of unnatural allowed configurations. This fact can be confirmed experimentally for small values of $\Delta$ and $r$, by applying round elimination for a few steps, but we need to prove that this is indeed the case for any $\Delta$, $r$, and number of steps. Hence, we want to be able to perform $\Omega(\Delta r)$ steps of round elimination before reaching a $0$-round solvable problem, but we cannot leave the number of labels grow like a power tower, because it would then be infeasible to actually understand the problem sequence. In order to achieve this goal, we have to perform simplifications that do not simplify the problem too much, as otherwise we would obtain a $0$-round solvable problem after $o(\Delta r)$ steps, which would give a worse lower bound.

We would like to use the same approach as \cite{hideandseek}, i.e., after round elimination step $k$, we want to add the allowed configurations of the fixed point of the $k$-hypergraph colorful coloring problem to the obtained problem. 
It is actually possible to extend the approach of \Cref{sec:fixedpoint} to show a fixed point relaxation for the $k$-hypergraph colorful coloring problem. This fixed point is also a relaxation of the $\lceil k / (r-1) \rceil$-hypergraph coloring problem.
The problem is that this fixed point is not exactly designed with colorful coloring in mind, and this makes it hard to adapt it to the case where colors increase by one at each step. 
In particular, consider the case where we try to use, as the $k$-th problem in our relaxed sequence, the problem $\Pi(k)$ containing two parts: the original problem, and the fixed point relaxation of the $k$-hypergraph colorful coloring problem (as would be the case with the approach from~\cite{hideandseek}). Note that the second part would contain $\lceil k / (r-1) \rceil$ colors, and these colors would satisfy the constraints of the hypergraph non-colorful coloring problem. By performing round elimination, we obtain a new color, and it does not seem possible to relax the obtained problem to $\Pi(k + 1)$. In particular, we would like to relax the coloring part of the obtained problem to the fixed point for $\lceil (k+1) / (r-1) \rceil$-hypergraph coloring, but it seems that the coloring part of the obtained problem is already easier than that, and that it can only be relaxed to the fixed point for $(\lceil k / (r-1) \rceil + 1)$-hypergraph coloring. If instead we use the latter approach (i.e., the approach of relaxing the coloring part of the obtained problem to the fixed point with the smallest number of colors that the coloring part can be relaxed to), the number of colors in the coloring part grows by $1$ at each step, meaning that  after only $\Delta+1$ steps of round elimination, we would have $\Delta+1$ colors.
However, the construction of the fixed point is only possible up to $\Delta$ colors, which renders also the modified approach infeasible (as already after $\Delta + 1$ steps we would obtain a problem that is $0$-round solvable).

In order to make things work, our idea is to not use directly a fixed point for hypergraph coloring, but to build a problem sequence in which, if we take the coloring part of each problem, we obtain a problem sequence that (except for the few first steps) has always $\Delta$ colors, and in which the hyperedge constraint gets more relaxed at each step. This coloring part at some point will become actually equal to the $\Delta$-coloring fixed point, but only after $\Omega(\Delta r)$ steps. All this is achieved by merging, at each step, the new color with an existing one, in such a way that we relax the hyperedge constraint as little as possible.

Unfortunately, the challenges mentioned above are just the conceptual issues that we have to take care of---there are a number of technical issues on top of the conceptual ones that make our lower bound proof fairly complicated.
One of these technical issues is that even when following our new approach, we obtain a problem sequence that is too relaxed, and hence we cannot obtain a good lower bound in this way. However, if during the relaxation procedure we select one color and treat it slightly different than all the other colors (i.e., amongst all the configurations that we add during the relaxation procedure the ones that contain the selected color have a slightly different structure than those containing other colors), then the relaxations turn out to be sufficiently tight.
The choice of the selected color is irrelevant, as long as we do not select the same color in two subsequent relaxation procedures.
In general, the hypergraph setting increases the complexity of various crucial ingredients in our proof (such as constructing a suitable problem family, proving that our relaxations indeed lead to the stated problem family by setting up a suitable instance on which we can apply Hall's marriage theorem, etc.) and therefore of the proof itself---we hope that our contributions (both the conceptual and the technical ones) will help navigating the difficulties of proving further lower bounds for important hypergraph problems.
We provide a more detailed overview of our lower bound approach for hypergraph maximal matching in Section~\ref{sec:lbmm}.

\paragraph{An Easier Proof.}
We would like to point out that in order to prove a lower bound for hypergraph maximal matching, we could have entirely skipped \Cref{sec:fixedpoint}, i.e., the proof of the hypergraph coloring lower bound: it is not necessary to prove that the relaxation of the $\Delta$-hypergraph coloring that we provide is indeed a fixed point, because it does not seem possible to make the approach of \cite{hideandseek} modular.
In other words, we cannot use the lower bound for $\Delta$-hypergraph coloring as a building block for the lower bound proof for hypergraph maximal matching, but we have to prove the latter lower bound from scratch; however, we can make use of some of the structures used in the proof for $\Delta$-hypergraph coloring.
In fact, essentially, the lower bound proof for coloring is a much simplified version of the one for hypergraph maximal matching, which has to handle many additional challenges. For example, the number of colors in the aforementioned ``natural part'' of the problems in the sequence obtained by starting from hypergraph maximal matching grows by $1$ at each step, while in the case of $\Delta$-hypergraph coloring it never changes.

%\paragraph{An Easier Proof.}
%In order to prove a lower bound for hypergraph maximal matching, we could have entirely skipped \Cref{sec:fixedpoint}, i.e., the proof of the hypergraph coloring lower bound: it is not necessary to prove that the relaxation of the $\Delta$-hypergraph coloring that we provide is indeed a fixed point, because it does not seem possible to make the approach of \cite{hideandseek} modular. In order to prove a lower bound for hypergraph maximal matching we have anyways to prove everything from scratch. But, the proof that we provide for $\Delta$-hypergraph coloring in \Cref{sec:fixedpoint} can be seen as a simplified version of the proof that we provide for hypergraph maximal matching in \Cref{sec:lbmm}. In fact, essentially, the lower bound proof for hypergraph maximal matching is very similar to the one for coloring, but it has to handle many additional challenges. For example, the number of colors in the natural part of the problems in the sequence obtained by starting from hypergraph maximal matching grows by $1$ at each step, while in the case of $\Delta$-hypergraph coloring they never change.

\section{Definitions and Notation}\label{sec:defs}

The upper and lower bounds presented in this work hold on hypergraphs, which have been defined in \Cref{sec:intro}. We now provide some additional notation.

We denote with $N_H(v)$ the neighbors of a node $v$, that is, all those nodes in $H$, different from $v$, contained in the hyperedges incident to $v$. 

Each node of a hypergraph is at distance $0$ from itself. Let $w_0,\cdots,w_d$ be the smallest sequence of nodes of $H$ such that:  $w_0=u$, $w_d=v$, and $w_{i+1}\in N_H(w_i)$ for all $0\le i< d$. Then we say that nodes $u$ and $v$ are at distance $d$ in $H$.
We call the \emph{$r$-hop neighborhood} of a node $u$ the subhypergraph induced by nodes in $H$ that are at distance at most $r$ from $u$, where we remove the hyperedges that contain only nodes at distance exactly $r$ from $u$. 

We will omit the subscript ``$H$'' from the notation if $H$ is clear from the context.

\subsection{The \LOCAL Model}
Our lower and upper bounds hold in the classic \LOCAL model of distributed computing, introduced by Linial \cite{Linial1992}. Given an $n$-node hypergraph $H$, each node is equipped with a unique identifier from $\{1,\ldots, n^c\}$, for some constant $c\ge1$. This is a message passing model, where the computation proceeds in synchronous rounds: at each round, nodes of a hypergraph exchange messages with the neighbors and perform some local computation. In this model, both the local computational power of each node, and the size of the messages, are not bounded. A distributed algorithm that solves a graph problem in the \LOCAL model runs at each node in parallel, and the goal is for nodes to produce their local outputs that together form a valid global solution of the desired problem. We say that an algorithm has round complexity $T$ if each node decides its local output and terminates within $T$ communication rounds. 

In the randomized version of the \LOCAL model, each node is equipped with a private random bit string, and in this context we consider Monte Carlo algorithms, that is, we require to produce a correct solution to the desired problem with high probability, that is, with probability at least $1-1/n$.

As typically done in this context, throughout this paper, we assume that each node $v$ of a hypergraph $H$ knows initially its own degree $\deg(v)$, the maximum degree $\Delta$, the maximum rank $r$, and the total number $n$ of nodes.

Notice that, since in the \LOCAL model we do not restrict the bandwidth, a $T$-round \LOCAL algorithm is equivalent to the following: first each node spends $T$ rounds to collect its $T$-hop neighborhood in $H$, and then maps each $T$-hop neighborhood into an output. Sometimes, it is convenient to work on the incidence graph (or bipartite representation) of a hypergraph $H=(V_H, E_H)$, which is a bipartite graph $B=(U_B\cup W_B, E_B)$ where $U_B=V_H$, $W_B=E_H$, and there is an edge between $u\in U_B$ and $v\in W_B$ if and only if the hyperedge of $H$ corresponding to $v$ contains the node of $H$ corresponding to $u$. Notice that, in the \LOCAL model, any $T$-rounds algorithm that solves a problem in $H$ can clearly be simulated in $B$ in at most $2T$ communication rounds.

\paragraph{The PN Model.}
While both our upper bound and lower bound results hold in the \LOCAL model, for technical reasons, the lower bounds are first shown on the \emph{port numbering} (PN) model, which is weaker than the \LOCAL one, and then we lift them for the randomized and deterministic \LOCAL model. 

The PN model is the same as the \LOCAL one (synchronous message passing model, unbounded computational power, unbounded size of messages), with the difference that nodes do not have unique identifiers, and instead each node has an internal ordering of its incident hyperedges. More precisely, incident hyperedges of a node $v$ have a pairwise distinct (port) number in $\{1,\ldots,\deg(v)\}$ (assigned arbitrarily). For technical reasons, we use a slightly modified version of the PN model, where in addition we require that incident nodes of a hyperedge have a pairwise distinct number in $\{1, \ldots, \rank(e)\}$ (assigned arbitrarily).

\subsection{The Automatic Round Elimination Framework}

Round elimination is a technique that can be used to prove lower bounds in the distributed setting. It has been first used to prove lower bounds for sinkless orientation and $\Delta$-coloring \cite{Brandt2016}. In its current form, called \emph{automatic} round elimination, it has been introduced by Brandt~\cite{Brandt2019}, and since then it has been proved to be extremely useful for proving lower bounds in the \LOCAL model.

In this section, we describe the automatic round elimination framework. We start by describing how to encode a problem in this framework (introducing some notation as well), and what should satisfy an output in order to be correct---we illustrate these concepts by showing the concrete example of the encoding of the MIS problem on hypergraphs with maximum rank $2$ (that is, on standard graphs).

\paragraph{Encoding of a Problem.}
For the purpose of showing our lower bounds, it is enough to consider only regular linear hypertrees, that is, hypergraphs such that their incidence graph is a bipartite $2$-colored tree, where each white non-leaf node has degree $\Delta$ and each black node has degree $r$. Hence, we show how to encode problems in such case. In this formalism, a problem $\Pi_{\Delta, r}$ is described by a triple $(\Sigma_\Pi,\nodeconst_\Pi, \edgeconst_\Pi)$, where:

\begin{itemize}
	\item the \emph{alphabet set} $\Sigma_{\Pi}$ is a set of allowed labels,
	\item the \emph{node constraint} $\nodeconst_\Pi$ is a set of multisets of size $\Delta$ over the alphabet $\Sigma_\Pi$,
	\item the \emph{hyperedge constraint} $\edgeconst_\Pi$ is a set of multisets of size $r$ over the alphabet $\Sigma_\Pi$.
\end{itemize}
Sometimes, instead of \emph{multisets}, we will use \emph{words}. It is just a matter of convenience: writing a word is shorter than writing a multiset. But, the word will still represent a multiset, in the sense that a label can appear many times in a word and the order of appearance of the labels does not matter. 

A word of length $\Delta$ (resp.\ $r$) over the alphabet $\Sigma_{\Pi}$ is called \emph{node configuration} (resp.\ \emph{hyperedge configuration}), and it is \emph{valid} or \emph{allowed} or \emph{it satisfies the constraint} if it is contained in $\nodeconst_\Pi$ (resp.\ $\edgeconst_\Pi$).

\paragraph{The Output.} Informally, in this formalism, in order to solve a problem $\Pi$ on a regular linear hypertree $H=(V,E)$, each node must output a label from $\Sigma_{\Pi}$ on each incident hyperedge. More precisely, each (node, incident hyperedge) pair of the set $\{(v,e)~|~ v\in V, e\in E, v\in e \}$ must be labeled with an element from the alphabet $\Sigma_{\Pi}$. We then say that an output is correct if it satisfies $\nodeconst_\Pi$ and $\edgeconst_\Pi$. More precisely, let $v\in V$ be a node and let $I_v$ be the set of all hyperedges $e\in E$ such that $v\in e$. Then the word of size $\Delta$ described by all the labels given to $(v,e)$ must be in $\nodeconst_\Pi$. Similarly, let $e\in E$ be a hyperedge and let $I_e$ be the set of all nodes $v$ such that $v\in e$. Then the word of size $r$ described by all the labels given to $(v,e)$ must be in $\edgeconst_\Pi$.

\paragraph{Example: Encoding of MIS.} As an example, let us see how to encode the MIS problem on a $\Delta$-regular graph $G=(V,E)$ in the round elimination framework. We show an example with MIS because, in the round elimination framework, MIS behaves similarly as hypergraph maximal matching. In other words, we want to define $\Sigma$, $\nodeconst$, and $\edgeconst$, such that any labeling that satisfies the node and edge constraint results in a maximal independent set, and any maximal independent set can be used to produce such a labeling. The alphabet set is $\Sigma=\{\M,\P,\O\}$, where intuitively $\M$ is used to say that a node is in the MIS, $\P$ is used to point to a neighbor that is in the MIS, while $\O$ stands for ``other'' and is used to express that somehow it does not matter what happens in that part of the graph (for example, from the perspective of a node $v$, the label $\O$ on an incident half-edge $(v,e)$ indicates that it does not matter whether the neighbor reached through $e$ is in the MIS or not). The node and edge constraints are defined as follows.

\begin{equation*}
	\begin{aligned}
		\begin{split}
			\nodeconst\text{:} \\ 
			& \quad\M^{\Delta} \\
			& \quad\P\s\O^{\Delta-1} 
		\end{split}
		\qquad
		\begin{split}
			\edgeconst\text{:} \\
			& \quad \P\s\M \\
			& \quad \O\s \M\\
			& \quad \O\s \O
		\end{split}
	\end{aligned}
\end{equation*}

Nodes in the MIS output the configuration $\M^\Delta$. Nodes not in the MIS output the configuration $\P\s\O^{\Delta-1}$, where $\P$ is used to guarantee the maximality of the solution by pointing to one neighbor that is in the MIS, while $\O$ is given on the other incident half-edges that are connect to nodes that may or may not be in the MIS. Since we do not allow two neighboring nodes to be in the MIS, then $\M\s\M\notin\edgeconst$. Since $\P$ must only point to a neighbor in the MIS then we have $\P\s\M\in\edgeconst$, while $\P\s\P,\P\s\O\notin\edgeconst$. Also, since a node not in the MIS may have more than one neighbor in the MIS, or other neighbors not in the MIS, then we get $\O\s\O\in\edgeconst$ and $\O\s\M\in\edgeconst$.

Please observe that the node constraint only specifies what needs to be satisfied by nodes of degree exectly $\Delta$, and this means that leaves are unconstrained. Hence, these constraint do not describe the exact MIS problem, but a similar problem, where only non-leaf nodes need to actually solve MIS. Since we prove lower bounds, this is not an issue: a lower bound for this simplified variant of MIS would imply a lower bound for the standard MIS problem.

\subsection{Automatic Round Elimination Technique}
We now present the automatic round elimination technique, and we dive into some details that are necessary and sufficient for understanding the technical parts of our lower bound proofs. The automatic version of the round elimination technique was introduced by Brandt \cite{Brandt2019}, and later actually implemented by Olivetti \cite{Olivetti2019}. On a high level, if we apply the round elimination technique on a problem $\Pi$ with complexity $T$ we obtain a (possibly) new problem $\Pi'$ having complexity $\max\{0, T-1\}$. Let $\Pi_1$ be our problem of interest for which we want to show a lower bound. By applying automatic round elimination iteratively we get a sequence of problems $\Pi_2, \Pi_3,\dotsc$ such that each problem $\Pi_i$ is at least one round easier than $\Pi_{i -1}$, assuming that $\Pi_{i -1}$ is not already $0$-rounds solvable. Hence, if we can apply round elimination for $T$ times before reaching a $0$-rounds solvable problem, we get a lower bound of $T$ rounds for our problem of interest $\Pi_1$. 

If we take a more fine-grained look at this technique, what actually happens is that, when we apply round elimination on a problem $\Pi$, this procedure first gives an \emph{intermediate} problem $\Pi'$, and only after applying round elimination on $\Pi'$ we then get the problem $\Pi''$ that is at least one round easier than $\Pi$. Hence, let $\re(\cdot)$ be the function that takes in input a problem $\Pi$ and outputs the intermediate problem $\Pi'$, and let $\rere(\cdot)$ be the function that takes in input an intermediate problem $\Pi'$ and outputs the problem $\Pi''$. We get that $\Pi'=\re(\Pi)$ and $\Pi''=\rere(\Pi')=\rere(\re(\Pi))$. 

Applying round elimination means to first compute $\Pi' = \re(\Pi)$ and then compute $\Pi'' = \rere(\Pi')$, and $\Pi''$ is the problem that is guaranteed to have complexity $\max\{0, T-1\}$ if $\Pi$ has complexity $T$. On a high level, the labels that are used to define $\Pi'$ are \emph{sets} of labels of $\Pi$, and the labels used to define $\Pi''$ are \emph{sets} of labels of $\Pi'$. Observe that the configurations allowed by the constraints of $\Pi'$ are still just multisets of labels, but the elements of these multisets are hence now \emph{sets} of labels of $\Pi$. 

We now present the definition of the function $\re$. 
Starting from a problem $\Pi=(\Sigma_{\Pi},\nodeconst_\Pi,\edgeconst_\Pi)$, the problem $\re(\Pi=(\Sigma_{\Pi},\nodeconst_\Pi,\edgeconst_\Pi))=\Pi'=(\Sigma_{\Pi'},\nodeconst_{\Pi'},\edgeconst_{\Pi'})$  is defined as follows.

\begin{itemize}
	\item The edge constraint $\edgeconst_{\Pi'}$ is defined as follows.
	Consider a configuration $\L_1\s\dotsc\s\L_r$ where for all $1\le i\le r$ it holds that $\L_i\in 2^{\Sigma_{\Pi}}\setminus\{\emptyset\}$, such that \emph{for any} $(\ell_1,\dotsc,\ell_r)\in \L_1\times\dotsc\times\L_r$ it holds that a permutation of $\ell_1\s\dotsc\s\ell_r$ is in $\edgeconst_\Pi$ (note that, by definition of the above configuration, $\ell_i\in\Sigma_{\Pi}$, for all $1\le i\le r$). Let $\mathcal{S}$ be the set that contains all and only such configurations. We say that any configuration in $\mathcal{S}$ \emph{satisfies the universal quantifier}. We call a configuration $\L_1\s\dotsc\s\L_r\in\mathcal{S}$ \emph{non-maximal} if there exists a permutation $\L'_1\s\dotsc\s\L'_r$ of another configuration in $\mathcal{S}$ such that, $L_i\subseteq L'_i$ for all $i$, and there exists at least an index $i$ such that $L'_i\subsetneq L_i$. The hyperedge constraint $\edgeconst_{\Pi'}$ is defined as the set $\mathcal{S}$ where we remove all non-maximal configurations (in other words, only \emph{maximal} configurations are kept). We refer to computing $\mathcal{S}$ as \emph{applying the universal quantifier}.
	
	\item $\Sigma_{\Pi'}$ is defined as all subsets of $\Sigma_{\Pi}$ that appear at least once in $\edgeconst_{\Pi'}$.
	
	\item Consider a configuration $\L_1\s\dotsc\s\L_\Delta$ of labels in $\Sigma_{\Pi'}$, such that \emph{there exists} $(\ell_1,\dotsc,\ell_\Delta)\in \L_1\times\dotsc\times\L_\Delta$ for which it holds that a permutation of $\ell_1\s\dotsc\s\ell_\Delta$ is in $\nodeconst_\Pi$. Let $\mathcal{S}$ be the set of all and only such configurations. We say that any configuration in $\mathcal{S}$ \emph{satisfies the existential quantifier}. The node constraint $\nodeconst_{\Pi'}$ is defined as the set $\mathcal{S}$. We refer to computing $\mathcal{S}$ as \emph{applying the existential quantifier}.
\end{itemize}

The problem $\rere(\re(\Pi(\Sigma_{\Pi},\nodeconst_\Pi,\edgeconst_\Pi)))=\rere(\Pi'(\Sigma_{\Pi'},\nodeconst_{\Pi'},\edgeconst_{\Pi'}))=\Pi''(\Sigma_{\Pi''},\nodeconst_{\Pi''},\edgeconst_{\Pi''})$ is defined similarly. The constraint $\nodeconst_{\Pi''}$ is obtained by applying the universal quantifier on $\nodeconst_{\Pi'}$, the set of labels $\Sigma_{\Pi''}$ is defined as all subsets of $\Sigma_{\Pi'}$ that appear at least once in $\nodeconst_{\Pi''}$, and  $\edgeconst_{\Pi''}$ is obtained by applying the existential quantifier on $\edgeconst_{\Pi'}$.

\paragraph{Example: MIS.}
Let us see now a concrete example that illustrates the rules described above. 
Our starting problem is $\Pi=\mbox{MIS}$ on $\Delta$-regular trees, which we already showed how to encode in this formalism, and we now show how $\Pi'=\re(\Pi)$ looks like.

\begin{itemize}
	\item We start by computing the edge constraint $\edgeconst_{\Pi'}$. In order to do so, it is useful to first compute the set $2^{\Sigma_{\Pi}}\setminus\{\emptyset\}=\{\{\M\}, \{\P\}, \{\O\}, \{\M,\P\}, \{\M,\O\}, \{\P,\O\}, \{\M,\P,\O\} \}$. All configurations with labels in the above-defined set that satisfy the universal quantifier are the following: $\{\P\}\s \{\M\}$, $\{\O\}\s \{\M\}$, $\{\O\}\s \{\O\}$, $\{\O\}\s \{\M,\O\}$, $\{\M\}\s \{\P,\O\}$. Let $\mathcal{S}$ be the set that contains all the above configurations. The set $\edgeconst_{\Pi'}$ is then defined as the set $\mathcal{S}$ where we remove all non-maximal configurations. The configuration $\{\P\}\s \{\M\}$ is non-maximal because there exists $\{\M\}\s \{\P,\O\}$; the configurations $\{\O\}\s \{\M\}$ and $\{\O\}\s \{\O\}$ are non-maximal because there exists $\{\O\}\s \{\M,\O\}$. Hence, $\edgeconst_{\Pi'}=\{ \{\M\}\s \{\P,\O\}, \{\O\}\s \{\M,\O\}\}$.
	\begin{itemize}
		\item For the sake of readability we rename the sets of labels as follows: $\{\M\} \mapsto \bM$, $\{\O\} \mapsto \bO$, $\{\P,\O\} \mapsto \bP$, $\{\M,\O\} \mapsto \bX$. Hence $\edgeconst_{\Pi'}=\{\bP\s\bM, \bO\s\bX\}$.
	\end{itemize}

	\item $\Sigma_{\Pi'}=\{\bM,\bO,\bP,\bX\}$.
	
	\item It is possible to check that $\nodeconst_{\Pi'}$ is given by all the words generated by the following regular expression: $[\bM\bX]^\Delta  ~|~ \bP[\bO\bP\bX]^{\Delta-1}$.
\end{itemize}
In order to fully apply round elimination, we would then apply $\rere$ on $\Pi'$, and hence perform the following operations:
\begin{itemize}
	\item Compute the node constraint $\nodeconst_{\Pi''}$ by applying the universal quantifier on $\nodeconst_{\Pi'}$.
	\item Define $\Sigma_{\Pi''}$ as the set of labels appearing in $\nodeconst_{\Pi''}$.
	\item Compute the edge constraint  $\edgeconst_{\Pi''}$ by applying the existential quantifier on $\edgeconst_{\Pi'}$. 
	\item As before, we would then perform a renaming to replace sets with plain labels. This is not necessary, but it improves the readability of the result.
\end{itemize}
By a suitable renaming of the obtained sets of labels, we would obtain the following problem, which, by round elimination, is exactly one round easier than MIS. The node constraint is given by all the words generated by the regular expression $\M^{\Delta} ~|~ \P\s\O^{\Delta-1} ~|~ \X\s\A^{\Delta-1} ~|~ \C^{\Delta}$, and the edge constraint is given by all the words generated by the regular expression $[\M\A] \s [\P\C\A\U] ~|~  [\X\M\C\A\U] \s \U$.

\paragraph{The Automatic Round Elimination Theorem.}
In order to prove our lower bounds we will use the following theorem proved in \cite{Brandt2019}.

\begin{theorem}[\cite{Brandt2019}, Theorem 4.3 (rephrased)]\label{thm:rethm}
	Let $T > 0$. Consider a class $\mathcal{G}$ of hypergraphs with girth at least $2T+2$, and some locally checkable problem $\Pi$. Then, there exists an algorithm that solves problem $\Pi$ in $\mathcal{G}$ in $T$ rounds if and only if there exists an algorithm that solves $\rere(\re(\Pi))$ in $T-1$ rounds.
\end{theorem}

We notice that the more challenging part of applying the round elimination technique is computing the set of configurations that must satisfy the for all quantifier. For example, in computing the intermediate problem $\Pi'$, the challenging part is computing $\edgeconst_{\Pi'}$. In fact, an alternative and easy way to compute $\nodeconst_{\Pi'}$ is to simply define it as the set that contains all configurations given by regular expressions obtained by taking a configuration in $\nodeconst_\Pi$ and replacing each label $\ell$ with the disjunction of all label sets in $\Sigma_{\Pi'}$ that contain $\ell$. One thing that may help in proving lower bounds is that, it is not necessary to show that $\edgeconst_{\Pi'}$ contains  all \emph{and only} those configurations that satisfy the universal quantifier, but it is enough to show that there are no maximal configurations that satisfy the universal quantifier that are not in $\edgeconst_{\Pi'}$. In other words, we can add more allowed node and/or edge configurations to our lower bound sequence of problems and make these problems potentially easier, and if we are able to show a lower bound on these problems, then this lower bound applies to harder problems as well. Hence, we now define the notion of \emph{relaxation of a configuration}, which was already introduced in \cite{balliurules}.

\begin{definition}[Relaxation]\label{def:relaxing}
	Let $\fB = \B_1 \s \dots \s \B_j$ and $\fB' = \B'_1 \s \dots \s \B'_k$ be configurations consisting of subsets of some label space $\Sigma$.
	We say that $\fB$ \emph{can be relaxed to} $\fB'$ if $j = k$ and there exists a permutation $\rho \colon \{ 1, \dots, j \} \rightarrow \{ 1, \dots, k \}$ such that for any $1 \leq i \leq j$ we have $\B_i \subseteq \B'_{\rho(i)}$.
	
	Let $C$ and $C'$ be two constraints, i.e., two collections of configurations.
	Then we call $C'$ a \emph{relaxation} of $C$ if every configuration in $C$ can be relaxed to some configuration in $C'$.
	Moreover, for two problems $\Pi_1 = (\Sigma_1, \nodeconst_1, \edgeconst_1)$ and $\Pi_2 = (\Sigma_2, \nodeconst_2, \edgeconst_2)$, we call $\Pi_2$ a \emph{relaxation} of $\Pi_1$ if $\nodeconst_2$ is a relaxation of $\nodeconst_1$ and $\edgeconst_2$ is a relaxation of $\edgeconst_1$.		
\end{definition}

\paragraph{Relation Between Labels.} 
In order to argue about desired properties of our lower bound family of problems, it will be useful to determine a relation between labels. Given a problem $\Pi$, let $C$ be either the node or the hyperedge constraint, and let $\A$ and $\B$ be two labels of $\Pi$. We say that $\A$ is \emph{at least as strong as} $\B$ \emph{according to} $C$ if, for any configuration $\ell_1\s\dotsc\s\ell_\delta\in C$ ($\delta=\Delta$ or $\delta=r$) that contains $\B$, there is also a configuration in $C$ that is a permutation of $\ell_1\s\dotsc\s\ell_\delta$ where $\B$ is replaced with $\A$. In other words, consider any configuration $\C\in C$ that contains $\B$; Let $\C'$ be the configuration obtained by $\C$ where $\B$ is replaced with $\A$; a permutation of $\C'$ must be in $C$. If a label $\A$ is at least as strong as another label $\B$, then we say that $\B$ is \emph{at least as weak} as $\A$, denoted with $\B\le\A$. If $\A$ is at least as strong as $\B$ and $\B$ is not at least as strong as $\A$, then we say that $\A$ is \emph{stronger} than $\B$ and $\B$ is \emph{weaker} than $\A$, denoted with $\B<\A$. If $\A \le \B$ and $\B \le \A$, then $\A$ and $\B$ are \emph{equally strong}.

In our lower bound proof we make use of a \emph{diagram} in order to show the relation between labels according to some constraint $C$: we call it \emph{node diagram} if it shows the relation between labels according to the node constraint $\nodeconst$, and \emph{hyperedge diagram} (or \emph{edge diagram}, if we are on graphs) if instead the relations are according to $\edgeconst$. In this diagram, there is a directed edge from a label $\B$ to a label $\A$ if $\A$ is at least as strong as $\B$ and there is no label $\D\not\in \{\A,\B\}$ such that $\D$ is stronger than $\B$ and weaker than $\A$. See \Cref{fig:mis} for an example of the diagram of the MIS problem according to its edge constraint. All labels in the diagram that are reachable from  a label $\ell$ are called the \emph{successors} of $\ell$. In other words, the successors of $\ell$, according to some constraint $C$, are all those labels that are at least as strong as $\ell$ according to $C$. 

\begin{figure}[h]
	\centering
	\includegraphics[width=0.17\textwidth]{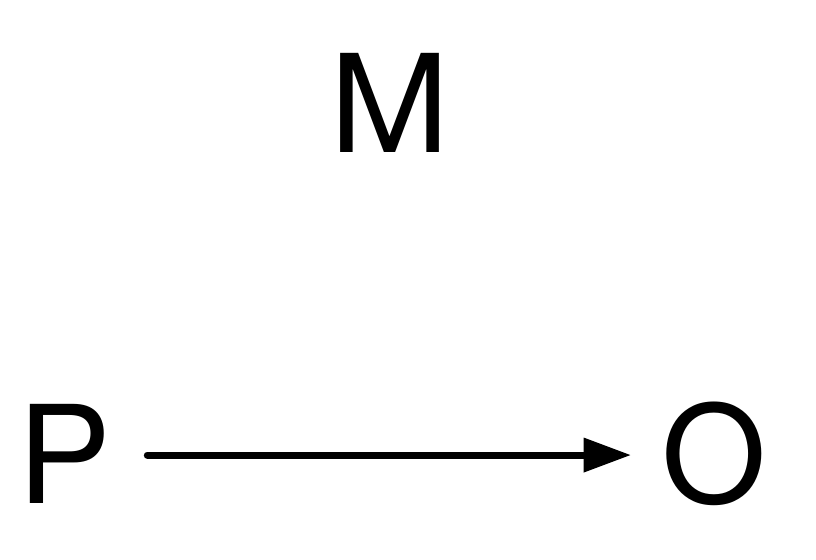}
	\caption{The edge diagram of the MIS problem. There is an edge from $\P$ to $\O$ because, in all edge configurations where $\P$ appears, that is $\P\s\M$, by replacing $\P$ with $\O$ we get the constraint $\O\M$ which is also in $\edgeconst$. There is no strength relation between $\M$ and labels $\P$ or $\O$.}
	\label{fig:mis}
\end{figure}

\paragraph{Additional Notation.}
Lastly, we introduce some additional notation that is essential and heavily used for describing in a concise way the problem family of our lower bounds. Consider a problem $\Pi=(\Sigma_{\Pi},\nodeconst_{\Pi},\edgeconst_\Pi)$, where $\Sigma_{\Pi} = \{\ell_1,\dotsc,\ell_\delta\}$. We denote by $\gen{\ell_1,\dotsc,\ell_\delta}$ the set of all labels that are at least as strong as at least one label in $\{\ell_1,\dotsc,\ell_\delta\}$, according to either $\nodeconst_\Pi$ or $\edgeconst_\Pi$. In other words, $\gen{\ell_1,\dotsc,\ell_\delta}$ is the set containing the labels $\ell_1,\dotsc,\ell_\delta$ along with all the successors of the labels $\ell_i$, $1\le i \le \delta$, according to either $\nodeconst_\Pi$ or $\edgeconst_\Pi$. Let us be a bit more precise regarding when it is that, by using $\gen{}$, we refer to successors of labels according to $\nodeconst_\Pi$ and when instead according to $\edgeconst_\Pi$. Consider a problem $\Pi=(\Sigma_{\Pi},\nodeconst_{\Pi},\edgeconst_\Pi)$ and let the argument of $\gen{}$ contains labels from $\Sigma_{\Pi}$. Then, when we compute $\re(\Pi)=\Pi'=(\Sigma_{\Pi'},\nodeconst_{\Pi'},\edgeconst_\Pi')$, the set $\gen{}$ contains successors according to $\edgeconst_\Pi$. Suppose now that the argument of $\gen{}$ contains labels from $\Sigma_{\Pi'}$. Then, when we compute $\Pi''=\rere(\Pi')=\rere(\re(\Pi))$, the set $\gen{}$ contains successors according to $\nodeconst_\Pi'$. Hence, the relation between labels is always considered according to the constraint on which we apply the universal quantifier. Actually, when we compute $\rere(\Pi)$ in our lower bound proofs, we will often use the expression of the form $\gen{\gen{\ell}}$, where $\ell\in\Sigma_\Pi$, that represents a set of sets of labels in $\Sigma_{\Pi}$, where $\gen{\ell}$ is generated according to $\edgeconst_\Pi$ and $\gen{\gen{\ell}}$ according to $\nodeconst_{\Pi'}$. 

We say that a set $S=\{\ell_1,\dotsc,\ell_\delta\}\subseteq\Sigma_\Pi$ is \emph{right-closed} if it contains also all successors of the labels $\ell_1,\dotsc,\ell_\delta$ according to $\edgeconst_{\Pi}$, that is, if $S=\gen{\ell_1,\dotsc,\ell_\delta}$. For a label $\ell$, we denote with \emph{right-closed subset generated by $\ell$} the set $\gen{\ell}$. Similarly, for a set of labels $\{\ell_1,\dotsc,\ell_\delta\}$, the right-closed subset generated by that set is $\gen{\ell_1,\dotsc,\ell_\delta}$. Observe that it is equal to $\bigcup_i \gen{\ell_i}$.
 In our proofs, we will make use of the following simple observation from~\cite{balliurules}.

	\begin{observation}[\cite{balliurules}]\label{obs:rcs}
		Consider an arbitrary collection of labels $\A_1, \dots, \A_p \in \Sigma_{\Pi}$.
		If $\{ \A_1, \dots, \A_p \} \in \Sigma_{\re(\Pi)}$, then the set $\{ \A_1, \dots, \A_p \}$ is right-closed (according to $\edgeconst_\Pi$).
		If $\{ \A_1, \dots, \A_p \} \in \Sigma_{\rere(\Pi)}$, then the set $\{ \A_1, \dots, \A_p \}$ is right-closed (according to $\nodeconst_\Pi$).
	\end{observation}

In order to get our lower bound, we will make use of the following theorem, which has been already proved in \cite{hideandseek}, and allows us to transform a lower bound obtained through round elimination into a lower bound for the \LOCAL model.
\begin{theorem}[Theorem 7.1 of \cite{hideandseek}, rephrased]\label{thm:lifting}
	Let $\Pi_0 \rightarrow \Pi_1 \rightarrow \ldots \rightarrow \Pi_t$ be a sequence of problems.
	Assume that, for all $0 \le i < t$, and for some function $f$, the following holds:
	\begin{itemize}
		\item There exists a problem $\Pi'_i$ that is a relaxation of $\re(\Pi_i)$;
		\item $\Pi_{i+1}$ is a relaxation of $\rere(\Pi'_i)$;
		\item The number of labels of $\Pi_i$, and the ones of $\Pi'_i$, are upper bounded by $f(\Delta \cdot r)$.
	\end{itemize}
	Also, assume that $\Pi_t$ has at most $f(\Delta \cdot r)$ labels and is not $0$-round solvable in the deterministic port numbering model. Then, $\Pi_0$ requires $\Omega(\min\{t, \log_{\Delta r} n - \log_{\Delta r} \log f(\Delta r)\})$ rounds in the deterministic \LOCAL model and $\Omega(\min\{t, \log_{\Delta r} \log n - \log_{\Delta r} \log f(\Delta r)\})$ rounds in the randomized \LOCAL model.
\end{theorem}

For any subset $W$ of nodes in some graph, let $N(W)$ denote the set of nodes that are adjacent to at least one member of $W$. In our lower bound proofs, we will make use of Hall's marriage theorem, which can be phrased as follows.
\begin{theorem}[Hall]\label{thm:hall}
	Let $G = (V \cup W, E)$ be a bipartite graph.
	If $|N(W')| \geq |W'|$ for any $W' \subseteq W$, then there exists a matching saturating $W$, i.e., a function $f \colon W \rightarrow V$ such that $f(w) \neq f(w')$ for any $w \neq w'$, and $\{ w, f(w) \} \in E$ for any $w \in W$.
\end{theorem}

\section{Lower Bounds for Hypergraph Colorings}\label{sec:fixedpoint}
In this section, we prove lower bounds for different variants of hypergraph coloring. We first consider the standard definition of coloring, where on each hyperedge it must hold that at least two incident nodes have different colors. We show how to encode the $\Delta$-hypergraph coloring problem in the round elimination framework, and then we show which additional allowed configurations we have to allow in order to transform it into a fixed point. A problem is a fixed point if, modulo some renaming, $\rere(\re(\Pi)) = \Pi$, and by prior work we know that if a fixed point problem cannot be solved in $0$ rounds in the PN model, then it means that it requires $\Omega(\log_{\Delta r} n)$ rounds for deterministic algorithms in the \LOCAL model, and of $\Omega(\log_{\Delta r} \log n)$ rounds for randomized ones.
We then show that there is a $0$ round algorithm that is able to convert a $\Delta(r-1)$-colorful coloring into a $\Delta$-coloring, implying that the lower bound that we show for $\Delta$-hypergraph coloring holds for the colorful variant as well.

\subsection{The $\Delta$-Hypergraph Coloring Problem.}
We now show how to define the hypergraph coloring problem in the round elimination framework. The label set $\Sigma_{\Pi}$ contains one label for each color, that is, $\Sigma_{\Pi} = \{1,\ldots,\Delta\}$. The node constraint $\nodeconst_{\Pi}$ contains one configuration for each possible color, that is,  $\nodeconst_{\Pi} = \{ c^\Delta ~|~ 1 \le c \le \Delta\}$. The hyperedge constraint $\edgeconst_{\Pi}$ contains all words of size $r$ where at least two different colors appear, that is, $\edgeconst_{\Pi} = \{ c_1 \ldots c_r ~|~ c_i \in \{1,\ldots,\Delta\} \land c_1 \neq c_2\}$. 

We now provide the high level idea on how this problem can be relaxed in order to make it a fixed point. The idea is to allow nodes to output not just colors, but sets of colors. For example, a node could use color $1$ and $2$ at the same time, but then, on each hyperedge incident to this node, there should be another incident node that is neither of color $1$ \emph{nor} of color $2$. Additionally, nodes are \emph{rewarded} for using more than one color: if a node uses $x$ colors, then it is allowed to mark up to $x-1$ incident hyperedges, and on those marked hyperedges the node counts as being uncolored. Observe that if a hyperedge $e$ is marked by at least one node, then the hyperedge constraint is always satisfied on $e$, because for each color it holds that there is at least one incident node that is \emph{not} colored with that color.

There is no clear intuition why this problem relaxation would be a fixed point, and in fact our proof is essentially a tedious case analysis. But, on a very high level, when applying the universal quantifier, the functions $\re$ and $\rere$ essentially try to combine the existing configurations in all possible ways that make the forall quantifier satisfied. If we look at $\Pi' = \rere(\re(\Pi))$, where $\Pi$ is the original $\Delta$-hypergraph coloring problem, we actually see something very similar to the problem that we just described. The main difference is that the allowed configurations of $\Pi'$ are a bit more restrictive: in the fixed point relaxations, if a hyperedge is marked by some incident node, then it is always happy, while in $\Pi'$ this is not always the case. Relaxing $\Pi'$ to make hyperedges always happy when marked is essentially what gives our fixed point relaxation.

\subsection{The $\Delta$-Hypergraph Coloring Fixed Point.}
We now define a problem $\Pi$ that is a relaxation of the $\Delta$-hypergraph coloring problem and that we will later show to be a fixed point problem that cannot be solved in $0$ rounds.

Let $\ccs := \{ 1, \dots, \Delta \}$.
We call the elements of $\ccs$ \emph{colors} from the \emph{color space} $\ccs$.
In the following, we define $\Pi = (\spi, \npi, \epi)$ by specifying $\spi$, $\npi$, and $\epi$.

\paragraph{The label set \spi.}
The label set $\spi$ of $\Pi$ contains one label $\ell(\ccc)$ for each (possibly empty) color set $\ccc \subseteq \ccs$, i.e., $\spi := \{ \ell(\ccc) \mid \ccc \subseteq \ccs \}$.  In other words, the labels are now sets of colors. 

\paragraph{The node constraint \npi.}
The node constraint $\npi$ of $\Pi$ consists of all configurations of the form $\ell(\ccc)^{\Delta - |\ccc| + 1} \s \ell(\emptyset)^{|\ccc| - 1}$ for some $\emptyset \neq \ccc \subseteq \ccs$, i.e., $\npi := \{ \ell(\ccc)^{\Delta - |\ccc| + 1} \s \ell(\emptyset)^{|\ccc| - 1} \mid \emptyset \neq \ccc \subseteq \ccs \}$.
In other words, each node outputs a set of colors, and on some of the incident hyperedges the empty set can be used. The intuition is that the empty set corresponds to the marking previously described.

\paragraph{The hyperedge constraint \epi.}
The hyperedge constraint $\epi$ of $\Pi$ consists of all configurations $\ell(\ccc_1) \s \dots \s \ell(\ccc_r)$ of labels from $\spi$ such that for each color $i \in \ccs$, there is at least one index $1 \leq j \leq r$ satisfying $i \notin \ccc_j$.

\paragraph{A Fixed Point.}
We will prove that $\Pi$ is a fixed point, that is, under some renaming, $\rere(\re(\Pi)) = \Pi$.
\begin{restatable}{lemma}{fponestep}\label{lem:fponestep}
	The problem $\rere(\re(\Pi))$ is equivalent to $\Pi$.
\end{restatable}
Since the proof of this statement is a tedious case analysis, we defer its proof to the end of the section. We now use this statement to derive our lower bounds.

\subsection{Proving the Lower Bounds}
We now prove that $\Pi$ cannot be solved in $0$ rounds. We will then show that this fact, combined with \Cref{lem:fponestep}, implies strong lower bounds for different hypergraph coloring variants.

\begin{lemma}\label{lem:fpzerorounds}
	The problem $\Pi$ cannot be solved in $0$ rounds in the deterministic port numbering model.
\end{lemma}
\begin{proof}
	All nodes of degree $\Delta$ that run a $0$-round algorithm have the same view (they just know their degree, the size of the graph, and the parameters $\Delta$ and $r$). Hence, any deterministic $0$-round algorithm uses the same configuration in $\npi$ on all nodes. Thus, in order to prove the statement, we consider all possible configurations and we show that they cannot be properly used in a $0$-round algorithm.
	The configurations allowed on the nodes are all those that satisfy  $\ell(\ccc)^{\Delta - |\ccc| + 1} \s \ell(\emptyset)^{|\ccc| - 1}$ for each $\emptyset \neq \ccc \subseteq \ccs$. Hence, all nodes must output one of those configuration for the same set $\ccc$ satisfying $\emptyset \neq \ccc \subseteq \ccs$. Since $|\ccc| \le \Delta$, then this configuration contains $\ell(\ccc)$ at least once. Since the algorithm is deterministic, then all nodes must output $\ell(\ccc)$ on the same port. W.l.o.g., let this port be the port number $1$. If we connect the port $1$ of $r$ nodes to the same hyperedge, then the hyperedge would have the configuration $\ell(\ccc)^{r}$, which is not in $\epi$.
\end{proof}

\begin{theorem}
	The hypergraph $\Delta$-coloring problem requires $\Omega(\log_{\Delta r} n)$ in the \LOCAL model for deterministic algorithms, and $\Omega(\log_{\Delta r} \log n)$ for randomized ones.
\end{theorem}
\begin{proof}
		By \Cref{lem:fponestep}, $\rerepi = \Pi$, and by \Cref{lem:fpzerorounds} $\Pi$ is not $0$-round solvable. Hence, we can build an arbitrarily long sequence of problems satisfying \Cref{thm:lifting} (or, in other words, we can take $t = \infty$). Also, note that the number of labels of each problem is bounded by $2^{O(\Delta)}$. Thus, since $\log_{\Delta r} \log f(\Delta) = O(1)$, we obtain a lower bound for $\Pi$ of  $\Omega(\log_{\Delta r} n)$ rounds  in the \LOCAL model for deterministic algorithms, and $\Omega(\log_{\Delta r} \log n)$ rounds for randomized ones.
\end{proof}

\begin{theorem}
	The hypergraph $\Delta(r-1)$-colorful coloring problem requires $\Omega(\log_{\Delta r} n)$ in the \LOCAL model for deterministic algorithms, and $\Omega(\log_{\Delta r} \log n)$ for randomized ones.
\end{theorem}
\begin{proof}
	We prove that, given a solution for the $\Delta(r-1)$-hypergraph coloring problem, we can solve hypergraph $\Delta$-coloring in $0$ rounds, implying the claim. We group the $\Delta(r-1)$ colors of the hypergraph colorful coloring problem into $\Delta$ groups $G_1,\ldots, G_{\Delta}$ of $r-1$ colors each, in an arbitrary way. Then, each node of color $c$ outputs $\ell(\{i\})^{\Delta}$, for $i$ satisfying $c \in G_i$. Observe that the solution clearly satisfies $\npi$. Also, since in any solution for the  $\Delta(r-1)$-hypergraph coloring problem on each hyperedge there can be at most $1$ incident node for each color, then we obtain that each hypergraph has at most $r-1$ incident nodes with colors from the same group. Hence, the constraint $\epi$ is satisfied.
\end{proof}

\subsection{Relations Between Labels}
In the rest of the section we prove \Cref{lem:fponestep}. We start by proving a relation between the labels.
From the definition of $\epi$, we can infer the following lemma characterizing the strength relations according to $\epi$.

\begin{lemma}\label{lem:colorstrong}
	Consider any two subsets $\ccc, \ccc'$ of $\ccs$.
	Then $\ell(\ccc) \leq \ell(\ccc')$ if and only if $\ccc' \subseteq \ccc$.
\end{lemma}
\begin{proof}
	Consider first the case that $\ccc' \subseteq \ccc$, and let $\ell(\ccc) \s \ell(\ccc_2) \s \dots \s \ell(\ccc_r)$ be an arbitrary configuration from $\epi$.
	Then it follows directly from the definition of $\epi$ that also $\ell(\ccc') \s \ell(\ccc_2) \s \dots \s \ell(\ccc_r)$ is contained in $\epi$.
	Hence, $\ell(\ccc) \leq \ell(\ccc')$, by the definition of strength.

	Now consider the other case, namely that $\ccc' \nsubseteq \ccc$, and let $i \in \ccc'$ be a color that is not contained in $\ccc$.
	Consider the configuration $\ell(\ccc) \s \ell(\ccs \setminus \ccc) \s \ell(\ccs)^{r - 2}$.
	Since, for any $i' \in \ccs$, we have $i' \notin \ccc$ or $i \notin \ccs \setminus \ccc$, we have $\ell(\ccc) \s \ell(\ccs \setminus \ccc) \s \ell(\ccs)^{r - 2} \in \epi$.
	Now consider the configuration $\ell(\ccc') \s \ell(\ccs \setminus \ccc) \s \ell(\ccs)^{r - 2}$ obtained from the above configuration by replacing $\ell(\ccc)$ by $\ell(\ccc')$.
	Since $i \in \ccc'$, $i \in \ccs \setminus \ccc$, and $i \in \ccs$, we have $\ell(\ccc') \s \ell(\ccs \setminus \ccc) \s \ell(\ccs)^{r - 2} \notin \epi$.
	Hence, $\ell(\ccc) \nleq \ell(\ccc')$, by the definition of strength.	
\end{proof}
	
	From Lemma~\ref{lem:colorstrong}, we immediately obtain the following corollary, by collecting, for each label $\ell(\ccc) \in \spi$, all labels $\ell(\overline{\ccc}) \in \spi$ satisfying $\ell(\ccc) \leq \ell(\overline{\ccc})$.

\begin{corollary}\label{cor:colorgen}
	For each $\ccc \subseteq \ccs$, we have $\gen{\ell(\ccc)} = \{ \ell(\ccc') \mid \ccc' \subseteq \ccc \}$.
\end{corollary}

\subsection{Computing $\re(\Pi)$}
We now prove that problem $\re(\Pi)$ is defined as in the following lemma.
\begin{lemma}\label{lem:colorrepidef}
	The set $\sre$ of output labels of $\repi$ is given by $\{ \gen{\ell(\ccc)} \mid \ccc \subseteq \ccs \}$.
	
	The hyperedge constraint $\ere$ of $\repi$ consists of all configurations $\gen{\ell(\ccc_1)} \s \dots \s \gen{\ell(\ccc_r)}$ of labels from $\{ \gen{\ell(\ccc)} \mid \ccc \subseteq \ccs \}$ such that for each color $i \in \ccs$, there is exactly one index $1 \leq j \leq r$ satisfying $i \notin \ccc_j$.

	The node constraint $\nre$ of $\repi$ consists of all configurations $\B_1 \s \dots \s \B_{\Delta}$ of labels from $\{ \gen{\ell(\ccc)} \mid \ccc \subseteq \ccs \}$ such that there exists a choice $(\A_1, \dots, \A_{\Delta}) \in \B_1 \times \dots \times \B_{\Delta}$ satisfying $\A_1 \s \dots \s \A_{\Delta} \in \npi$.
\end{lemma}
\begin{proof}
	We start by showing that the hyperedge constraint $\ere$ is as given in the lemma.
	Let $\edda$ denote the set of all configurations $\gen{\ell(\ccc_1)} \s \dots \s \gen{\ell(\ccc_r)}$ of labels from $\sre$ such that for each color $i \in \ccs$, there is exactly one index $1 \leq j \leq r$ satisfying $i \notin \ccc_j$.
	We need to show that $\ere = \edda$.

	We first show that $\ere \subseteq \edda$.
	Let $\B_1 \s \dots \s \B_r$ be an arbitrary configuration from $\ere$.
	Our first step is to show that each $\B_j$ is of the form $\gen{\ell(\ccc)}$ for some $\ccc \subseteq \ccs$.
	Assume for a contradiction that there exists some index $1 \leq j \leq r$ such that $\B_j \neq \gen{\ell(\ccc)}$ for each $\ccc \in \ccs$.
	Let $\ccc$ be the union of all color sets $\ccc'$ such that $\ell(\ccc') \in \B_j$.
	By Corollary~\ref{cor:colorgen}, it follows that $\B_j \subseteq \gen{\ell(\ccc)}$, and since $\B_j \neq \gen{\ell(\ccc)}$, we know that $\gen{\ell(\ccc)} \setminus \B_j$ is nonempty.
	Observe that, by Observation~\ref{obs:rcs}, $\B_j$ is right-closed (according to $\epi$) as $\B_j \in \sre$ (by the definition of $\B_j$).
	Hence, if $\ell(\ccc) \in \B_j$, then $\gen{\ell(\ccc)} \subseteq \B_j$, yielding a contradiction to the nonemptiness of $\gen{\ell(\ccc)} \setminus \B_j$.
	Thus, $\ell(\ccc) \notin \B_j$.

	Now consider the configuration $\B_1 \s \dots \s \B_{j-1} \s \gen{\ell(\ccc)} \s \B_{j+1} \s \dots \s \B_r$.
	Since $\B_j \subsetneq \gen{\ell(\ccc)}$, we know that $\B_1 \s \dots \s \B_{j-1} \s \gen{\ell(\ccc)} \s \B_{j+1} \s \dots \s \B_r$ and all configurations that $\B_1 \s \dots \s \B_{j-1} \s \gen{\ell(\ccc)} \s \B_{j+1} \s \dots \s \B_r$ can be relaxed to are not contained in $\ere$, by (the maximality condition in) the definition of $\ere$ (and the fact that $\B_1 \s \dots \s \B_r \in \ere$).
	It follows, by the definitions of $\ere$ and $\epi$, that there exist some choice $(\ell(\ccc'_1), \dots, \ell(\ccc'_r)) \in \B_1 \times \dots \times \B_{j-1} \times \gen{\ell(\ccc)} \times \B_{j+1} \times \dots \times \B_r$ and some color $i \in \ccs$ such that $i \in \ccc'_{j'}$ for each $1 \leq j' \leq r$.
	By Corollary~\ref{cor:colorgen}, the fact that $\ell(\ccc'_j) \in \gen{\ell(\ccc)}$ implies that $\ccc'_j \subseteq \ccc$, which, combined with $i \in \ccc'_j$, yields $i \in \ccc$.
	Let $\ccc^* \subseteq \ccs$ be a color set such that $i \in \ccc^*$ and $\ell(\ccc^*) \in \B_j$.
	Such a color set $\ccc^*$ exists due to the definition of $\ccc$ and the fact that $i \in \ccc$.
	Since $i \in \ccc^*$ and $i \in \ccc'_{j'}$ for each $1 \leq j' \leq r$ satisfying $j' \neq j$, we obtain $\ell(\ccc'_1) \s \dots \s \ell(\ccc'_{j - 1}) \s \ell(\ccc^*) \s \ell(\ccc'_{j + 1}) \s \dots \s \ell(\ccc'_r) \notin \epi$, by the definition of $\epi$.
	But, since $(\ell(\ccc'_1), \dots, \ell(\ccc'_{j - 1}), \ell(\ccc^*), \ell(\ccc'_{j + 1}), \dots, \ell(\ccc'_r)) \in \B_1 \times \dots \times \B_r$, this implies that $\B_1 \s \dots \s \B_r \notin \ere$, by the definition of $\ere$, yielding a contradiction.
	Hence, each $\B_j$ is of the form $\gen{\ell(\ccc)}$ for some $\ccc \subseteq \ccs$.
	
	Now, let $\gen{\ell(\ccc_1)} \s \dots \s \gen{\ell(\ccc_r)}$ be an arbitrary configuration from $\ere$ (justified by the above discussion), and consider an arbitrary color $i \in \ccs$.
	If for each index $1 \leq j \leq r$ we have $i \in \ccc_j$, then we know that $\ell(\ccc_1) \s \dots \s \ell(\ccc_r) \notin \epi$, by the definition of $\epi$; however, since, by Corollary~\ref{cor:colorgen}, we have $\ell(\ccc_j) \in \gen{\ell(\ccc_j)}$ for each $1 \leq j \leq r$, it follows that $\gen{\ell(\ccc_1)} \s \dots \s \gen{\ell(\ccc_r)} \notin \ere$ (by the definition of $\ere$), yielding a contradiction.
	Hence, we can conclude that there exists at least one index $1 \leq j \leq r$ such that $i \notin \ccc_j$.
	
	It remains to show that there do not exist two distinct indices $j, j' \in \{ 1, \dots, r \}$ such that $i \notin \ccc_j$ and $i \notin \ccc_{j'}$.
	For a contradiction, assume that such two indices $j, j'$ exist.
	Consider the configuration $Y = \gen{\ell(\ccc_1)} \s \dots \s \gen{\ell(\ccc_{j - 1})} \s \gen{\ell(\ccc_j \cup \{ i \})} \s \gen{\ell(\ccc_{j + 1})} \s \dots \s \gen{\ell(\ccc_r)}$.
	Since $i \notin \ccc_j$, we have $\gen{\ell(\ccc_j \cup \{ i \})} \supsetneq \gen{\ell(\ccc_j)}$, by Corollary~\ref{cor:colorgen}, and it follows by (the maximality condition in) the definition of $\ere$ (and the fact that $\gen{\ell(\ccc_1)} \s \dots \s \gen{\ell(\ccc_r)} \in \ere$) that $Y \notin \ere$.

	Now consider an arbitrary choice $(\ell(\ccc'_1), \dots, \ell(\ccc'_r)) \in \gen{\ell(\ccc_1)} \times \dots \gen{\ell(\ccc_{j - 1})} \times \gen{\ell(\ccc_j \cup \{ i \})} \times \gen{\ell(\ccc_{j + 1})} \times \dots \times \gen{\ell(\ccc_r)}$.
	%If $i \notin \ccc'_j$, then we also have $(\ell(\ccc'_1), \dots, \ell(\ccc'_r)) \in \gen{\ell(\ccc_1)} \times \dots \times \gen{\ell(\ccc_r)}$, by Corollary~\ref{cor:colorgen}, and it follows that $\ell(\ccc'_1) \s \dots \s \ell(\ccc'_r) \in \epi$, by the definition of $\ere$ (and the fact that $\gen{\ell(\ccc_1)} \s \dots \s \gen{\ell(\ccc_r)} \in \ere$).
	%Hence, assume that $i \in \ccc'_j$.
	Observe that $\ell(\ccc'_j) \in \gen{\ell(\ccc_j \cup \{ i \})}$ implies that $\ell(\ccc'_j \setminus \{ i \}) \in \gen{\ell(\ccc_j)}$, by Corollary~\ref{cor:colorgen}.
	Thus, $\ell(\ccc'_1) \s \dots \s \ell(\ccc'_{j - 1}) \s \ell(\ccc'_j \setminus \{ i \}) \s \ell(\ccc'_{j + 1}) \s \dots \s \ell(\ccc'_r) \in \epi$, by the definition of $\ere$ (and the fact that $\gen{\ell(\ccc_1)} \s \dots \s \gen{\ell(\ccc_r)} \in \ere$).
	It follows, by the definition of $\epi$, that for each $i' \in \ccs$ satisfying $i' \neq i$, there is at least one index $1 \leq j'' \leq r$ satisfying $i' \notin \C'_{j''}$.
	Observe that there also is at least one index $1 \leq j'' \leq r$ satisfying $j'' \neq j$ and $i \notin \ccc'_{j''}$, due to the assumptions that $j' \neq j$ and $i \notin \ccc_{j'}$.
	Hence, we can conclude that also $\ell(\ccc'_1) \s \dots \s \ell(\ccc'_r) \in \epi$, by the definition of $\epi$.
	Since $(\ell(\ccc'_1), \dots, \ell(\ccc'_r))$ was chosen arbitrarily from $\gen{\ell(\ccc_1)} \times \dots \gen{\ell(\ccc_{j - 1})} \times \gen{\ell(\ccc_j \cup \{ i \})} \times \gen{\ell(\ccc_{j + 1})} \times \dots \times \gen{\ell(\ccc_r)}$, it follows that $Y \in \ere$, by the definition of $\ere$.
	This yields a contradiction to the already established fact that $Y \notin \ere$.
	Thus, there is exactly one index $1 \leq j \leq r$ such that $i \notin \ccc_j$.
	Since $i$ was chosen arbitrarily from $\ccs$, it follows that $\gen{\ell(\ccc_1)} \s \dots \s \gen{\ell(\ccc_r)} \in \edda$, by the definition of $\edda$.
	Since $\gen{\ell(\ccc_1)} \s \dots \s \gen{\ell(\ccc_r)}$ was chosen arbitrarily from $\ere$, this implies in turn that $\ere \subseteq \edda$, as desired.

	Now, we show that $\edda \subseteq \ere$.
	Let $\gen{\ell(\ccc_1)} \s \dots \s \gen{\ell(\ccc_r)}$ be an arbitrary configuration from $\edda$, and consider some arbitrary choice $(\ell(\ccc'_1), \dots, \ell(\ccc'_r)) \in \gen{\ell(\ccc_1)} \times \dots \times \gen{\ell(\ccc_r)}$.
	From the definition of $\edda$ and Corollary~\ref{cor:colorgen}, it follows that for each color $i \in \ccs$, there is at least one index $1 \leq j \leq r$ satisfying $i \notin \ccc'_j$.
	Hence $\ell(\ccc'_1) \s \dots \s \ell(\ccc'_r) \in \epi$, by the definition of $\epi$.
	Since $(\ell(\ccc'_1), \dots, \ell(\ccc'_r))$ was arbitrarily chosen from $\gen{\ell(\ccc_1)} \times \dots \times \gen{\ell(\ccc_r)}$, this implies, by the definition of $\ere$, that $\gen{\ell(\ccc_1)} \s \dots \s \gen{\ell(\ccc_r)} \in \ere$ or there is some configuration $\B_1 \s \dots \s \B_r \in \ere$ such that $\gen{\ell(\ccc_j)} \subseteq \B_j$ for all $j \in \{ 1, \dots, r \}$ and $\gen{\ell(\ccc_j)} \subsetneq \B_j$ for at least one $j \in \{ 1, \dots, r \}$.
	Assume that such a configuration $\B_1 \s \dots \s \B_r \in \ere$ exists.
	By the definition of $\B_1 \s \dots \s \B_r$, the definition of $\edda$, and the fact that $\gen{\ell(\ccc_1)} \s \dots \s \gen{\ell(\ccc_r)} \in \edda$, there must be some color $i \in \ccs$ such that $i \in \B_j$ for each $1 \leq j \leq r$.
	Hence, $\B_1 \s \dots \s \B_r \notin \edda$ (by the definition of $\edda$), which contradicts $\B_1 \s \dots \s \B_r \in \ere$ since $\ere \subseteq \edda$ (as established above).
	Thus, no such configuration $\B_1 \s \dots \s \B_r \in \ere$ exists, which implies that $\gen{\ell(\ccc_1)} \s \dots \s \gen{\ell(\ccc_r)} \in \ere$.
	It follows that $\edda \subseteq \ere$, as desired.
	Hence, we can conclude that $\ere = \edda$, and it follows that the hyperedge constraint $\ere$ is as given in the lemma.
	
	Next, we show that $\sre$ is as given in the lemma.
	By definition, $\sre$ consists of those labels that appear in some configuration from $\ere$.
	Observe that, for any $\ccc \subseteq \ccs$, the configuration $\gen{\ell(\ccc)} \s \gen{\ell(\ccs \setminus \ccc)} \s \gen{\ell(\ccs)}^{r - 2}$ is contained in $\edda$, by the definition of $\edda$.
	Since $\ere = \edda$ (and all labels appearing in some configuration from $\ere$ are contained in $\{ \gen{\ell(\ccc)} \mid \ccc \subseteq \ccs \}$), it follows that $\sre$ is as given in the lemma.
	
	From the definition of $\re(\cdot)$, it immediately follows that $\nre$ is as given in the lemma.
\end{proof}

	Similar to before, it will be useful to collect information about the strength of the labels in $\sre$ (according to $\nre$) and compute the right-closed subsets generated by each label.
	We will do so in the following with Lemma~\ref{lem:colorrestrong} and Corollary~\ref{cor:colorregen}, which are analogues of Lemma~\ref{lem:colorstrong} and Corollary~\ref{cor:colorgen} for $\repi$ instead of $\Pi$.

\begin{lemma}\label{lem:colorrestrong}
	Consider any two subsets $\ccc, \ccc'$ of $\ccs$.
	Then $\gen{\ell(\ccc)} \leq \gen{\ell(\ccc')}$ if and only if $\ccc \subseteq \ccc'$.
\end{lemma}
\begin{proof}
		Consider first the case that $\ccc \subseteq \ccc'$.
		From the definition of $\nre$, it directly follows that replacing some set in a configuration from $\nre$ by a superset will result in a configuration that is contained in $\nre$ as well.
		Since, by Corollary~\ref{cor:colorgen}, $\ccc \subseteq \ccc'$ implies $\gen{\ell(\ccc)} \subseteq \gen{\ell(\ccc')}$, it follows that $\gen{\ell(\ccc)} \leq \gen{\ell(\ccc')}$ (by the definition of strength), as desired.

	Now consider the other case, namely that $\ccc \nsubseteq \ccc'$, and let $i \in \ccc$ be a color that is not contained in $\ccc'$.
	In particular, we have $\ccc \neq \emptyset$.
	Consider the configuration $\gen{\ell(\ccc)}^{\Delta - |\ccc| + 1} \s \gen{\ell(\emptyset)}^{|\ccc| - 1}$.
	Since $\ell(\ccc)^{\Delta - |\ccc| + 1} \s \ell(\emptyset)^{|\ccc| - 1} \in \npi$, we have $\gen{\ell(\ccc)}^{\Delta - |\ccc| + 1} \s \gen{\ell(\emptyset)}^{|\ccc| - 1} \in \nre$, by the definition of $\nre$.
	Now consider the configuration $\gen{\ell(\ccc')} \s \gen{\ell(\ccc)}^{\Delta - |\ccc|} \s \gen{\ell(\emptyset)}^{|\ccc| - 1}$ obtained from the above configuration by replacing (one) $\gen{\ell(\ccc)}$ by $\gen{\ell(\ccc')}$.
	Consider an arbitrary choice $(\ell(\ccc_1), \dots, \ell(\ccc_{\Delta})) \in \gen{\ell(\ccc')} \times \gen{\ell(\ccc)}^{\Delta - |\ccc|} \times \gen{\ell(\emptyset)}^{|\ccc| - 1}$.
	Set $Y := \ell(\ccc_1) \s \dots \s \ell(\ccc_{\Delta})$.
	We want to show that $Y \notin \npi$.

	For a contradiction, assume that $Y \in \npi$, which implies that there is some color set $\emptyset \neq \ccc^* \subseteq \ccs$ such that $Y = \ell(\ccc^*)^{\Delta - |\ccc^*| + 1} \s \ell(\emptyset)^{|\ccc^*| - 1}$.
	If $|\ccc^*| < |\ccc|$, then $Y$ contains strictly more than $\Delta - |\ccc| + 1$ labels that are not contained in $\gen{\ell(\emptyset)}$ (by Corollary~\ref{cor:colorgen}), which yields a contradiction to the definition of $Y$.
	Hence, $|\ccc^*| \geq |\ccc|$.
	Consider the case that $\ccc^* \neq \ccc$.
	Then it follows that $\ccc^* \nsubseteq \ccc$, which implies $\ell(\ccc^*) \notin \gen{\ell(\ccc)}$, by Corollary~\ref{cor:colorgen}; it follows that $\ell(\ccc^*)$ is contained in at most one set from $\gen{\ell(\ccc')} \s \gen{\ell(\ccc)}^{\Delta - |\ccc|} \s \gen{\ell(\emptyset)}^{|\ccc| - 1}$, which in turn implies $\Delta - |\ccc^*| + 1 \leq 1$, by the definition of $Y$, and therefore $|\ccc^*| = \Delta$.
	As $|\ccc^*| = \Delta$ implies $\ccc^* = \ccs$, but $\ell(\ccs) \notin \gen{\ell(\ccc')} \cup \gen{\ell(\ccc)} \cup \gen{\ell(\emptyset)}$ (due to the fact that $i \notin \ccc'$, the fact that $\ccs = \ccc^* \nsubseteq \ccc$, and Corollary~\ref{cor:colorgen}), we obtain a contradiction to the definition of $Y$ (given the fact that $Y = \ell(\ccc^*)^{\Delta - |\ccc^*| + 1} \s \ell(\emptyset)^{|\ccc^*| - 1}$).
	Hence, we know that $\ccc^* = \ccc$.
	Since $\ccc \nsubseteq \ccc'$, it follows, by Corollary~\ref{cor:colorgen}, that $\ell(\ccc^*) \notin \gen{\ell(\ccc')}$.
	Thus, $\ell(\ccc^*)$ is contained in at most $\Delta - |\ccc| = \Delta - |\ccc^*|$ sets from $\gen{\ell(\ccc')} \s \gen{\ell(\ccc)}^{\Delta - |\ccc|} \s \gen{\ell(\emptyset)}^{|\ccc| - 1}$, which yields a contradiction to the definition of $Y$ (given the fact that $Y = \ell(\ccc^*)^{\Delta - |\ccc^*| + 1} \s \ell(\emptyset)^{|\ccc^*| - 1}$).
	Hence, our assumption was false, and we have $Y \notin \npi$.
	It follows that $\gen{\ell(\ccc')} \s \gen{\ell(\ccc)}^{\Delta - |\ccc|} \s \gen{\ell(\emptyset)}^{|\ccc| - 1} \notin \nre$, which, by the definition of strength, implies that $\gen{\ell(\ccc)} \nleq \gen{\ell(\ccc')}$, as desired.
\end{proof}

From Lemma~\ref{lem:colorrestrong} we obtain Corollary~\ref{cor:colorregen} by simply collecting for any label $\L \in \sre$ the set of all labels $\L'$ satisfying $\L \leq \L'$.

\begin{corollary}\label{cor:colorregen}
	For each $\ccc \subseteq \ccs$, we have $\gen{\gen{\ell(\ccc)}} = \{ \gen{\ell(\ccc')} \mid \ccc \subseteq \ccc' \}$.
\end{corollary}

\subsection{Computing $\rere(\re(\Pi))$}
We now proceed by computing $\rere(\re(\Pi))$. We start with its node constraint.
\begin{lemma}\label{lem:colornode}
	The node constraint $\nrere$ of $\rerepi$ consists of all configurations of the form $\gen{\gen{\ell(\ccc)}}^{\Delta - |\ccc| + 1} \s \gen{\gen{\ell(\emptyset)}}^{|\ccc| - 1}$ for some $\emptyset \neq \ccc \subseteq \ccs$.
\end{lemma}
\begin{proof}
	Let $\neda$ denote the set of all configurations of the form $\gen{\gen{\ell(\ccc)}}^{\Delta - |\ccc| + 1} \s \gen{\gen{\ell(\emptyset)}}^{|\ccc| - 1}$ for some $\emptyset \neq \ccc \subseteq \ccs$.
	We need to show that $\nrere = \neda$.

	We start by showing that each configuration from $\neda$ can be relaxed to some configuration from $\nrere$.
	Let $\gen{\gen{\ell(\ccc)}}^{\Delta - |\ccc| + 1} \s \gen{\gen{\ell(\emptyset)}}^{|\ccc| - 1}$ be an arbitrary configuration from $\neda$, and consider some arbitrary choice $(\gen{\ell(\ccc'_1)}, \dots, \gen{\ell(\ccc'_{\Delta})}) \in \gen{\gen{\ell(\ccc)}}^{\Delta - |\ccc| + 1} \times \gen{\gen{\ell(\emptyset)}}^{|\ccc| - 1}$.
	By Corollary~\ref{cor:colorregen}, we know that, for each $1 \leq k \leq \Delta - |\ccc| + 1$, we have $\ccc \subseteq \ccc'_k$.
	Hence, by Corollary~\ref{cor:colorgen}, we can pick one label from each of the sets in $(\gen{\ell(\ccc'_1)}, \dots, \gen{\ell(\ccc'_{\Delta})})$ such that the resulting configuration is $\ell(\ccc)^{\Delta - |\ccc| + 1} \s \ell(\emptyset)^{|\ccc| - 1}$.
	Since $\ell(\ccc)^{\Delta - |\ccc| + 1} \s \ell(\emptyset)^{|\ccc| - 1} \in \npi$, it follows, by the definition of $\nrere$ (and $\nre$), that $\gen{\gen{\ell(\ccc)}}^{\Delta - |\ccc| + 1} \s \gen{\gen{\ell(\emptyset)}}^{|\ccc| - 1}$ can be relaxed to some configuration from $\nrere$.
	We conclude that each configuration from $\neda$ can be relaxed to some configuration from $\nrere$.
	Note that this does not show yet that $\neda \subseteq \nrere$; we will come back to this goal later.

	First, we show that $\nrere \subseteq \neda$.
	Let $\fB = \B_1 \s \dots \s \B_{\Delta}$ be an arbitrary configuration from $\nrere$.
	We claim that $\fB$ can be relaxed to some configuration from $\neda$.
	For a contradiction, assume that the claim is false, i.e., $\fB$ cannot be relaxed to any configuration from $\neda$.
	Observe that $\B_k$ is right-closed for each $1 \leq k \leq \Delta$, by Observation~\ref{obs:rcs}.
	
	For each $i \in \ccs$, define $R_i := \ccs \setminus \{ i \}$.
	Moreover, we will make use of the bipartite graph $\overline{G} = (\overline{V} \cup \overline{W}, \overline{E})$ obtained by defining $\overline{V} := \{ \B_1, \dots, \B_{\Delta} \}$ and $\overline{W} := \{ R_i \mid i \in \ccs \}$, and setting $\overline{E}$ to be the set of all edges $\{ \B_k, R_i \}$ satisfying $\gen{\ell(R_i)} \in \B_k$.
		Note that, for any distinct $k, k'$, we consider $\B_k$ and $\B_{k'}$ to be different vertices in $\overline{V}$ even if $\B_k = \B_{k'}$.
		
	Consider any arbitrary subset $\emptyset \neq \ccc \subseteq \ccs$.
	By our assumption, configuration $\fB$ cannot be relaxed to the configuration $\gen{\gen{\ell(\ccc)}}^{\Delta - |\ccc| + 1} \s \gen{\gen{\ell(\emptyset)}}^{|\ccc| - 1}$, which implies that there are at least $|\ccc|$ indices $k \in \{ 1, \dots, \Delta \}$ satisfying $\B_k \nsubseteq \gen{\gen{\ell(\ccc)}}$, since $\gen{\gen{\ell(\emptyset)}} = \sre$ (by Corollary~\ref{cor:colorregen}).
	By Corollary~\ref{cor:colorregen}, it follows that there are at least $|\ccc|$ indices $k \in \{ 1, \dots, \Delta \}$ such that there exists some $\ccc' \nsupseteq \ccc$ satisfying $\gen{\ell(\ccc')} \in \B_k$.
	Since, for each $\ccc' \nsupseteq \ccc$, there exists some $i \in \ccc$ satisfying $i \notin \ccc'$, it follows, by the right-closedness of the $\B_k$ (and Lemma~\ref{lem:colorrestrong}), that there are at least $|\ccc|$ indices $k \in \{ 1, \dots, \Delta \}$ such that there exists some $i \in \ccc$ satisfying $\gen{\ell(R_i)} \in \B_k$.
	
	Recall the definition of $\overline{G}$.
	By the above discussion, we conclude that for any arbitrary subset $\emptyset \neq \ccc \subseteq \ccs$, the vertex set $\{ R_i \mid i \in \ccc \} \subseteq \overline{W}$ has at least $|\ccc|$ neighbors in $\overline{V}$.
	Hence, we can apply Theorem~\ref{thm:hall} (i.e., Hall's marriage theorem) to $\overline{G}$ and obtain a function $f \colon \overline{W} \rightarrow \overline{V}$ such that $f(R_i) \neq f(R_{i'})$ for any $i \neq i'$ and $\gen{\ell(R_i)} \in f(R_i)$ for any $i \in \ccs$.
	This implies that there is some choice $(\A_1, \dots, \A_{\Delta}) \in \B_1 \times \dots \times \B_{\Delta}$ such that the configuration $\A_1 \s \dots \s \A_{\Delta}$ is a permutation of the configuration $\gen{\ell(R_1)} \s \dots \s \gen{\ell(R_{\Delta})}$.
	Therefore, by the definition of $\nrere$, there exists some choice $(\L_1, \dots, \L_{\Delta}) \in \gen{\ell(R_1)} \times \dots \times \gen{\ell(R_{\Delta})}$ such that $\L_1 \s \dots \s \L_{\Delta} \in \npi$.
	Observe that, for each $\emptyset \neq \ccc \subseteq \ccs$, there are at least $|\ccc|$ indices $i \in \ccs$ (namely all colors $i$ contained in $\ccc$) such that $\ell(\ccc) \notin \gen{\ell(R_i)}$, by Corollary~\ref{cor:colorgen}.
	It follows that, for each $\emptyset \neq \ccc \subseteq \ccs$, the configuration $\ell(\ccc)^{\Delta - |\ccc| + 1} \s \ell(\emptyset)^{|\ccc| - 1}$ is not a permutation of $\L_1 \s \dots \s \L_{\Delta}$.
	By the definition of $\npi$, this yields a contradiction to the fact that $\L_1 \s \dots \s \L_{\Delta} \in \npi$, and proves the claim.

	Hence, each configuration from $\nrere$ can be relaxed to some configuration from $\neda$.
	Recall that, as shown before, each configuration from $\neda$ can be relaxed to some configuration from $\nrere$.
	Combining these two insights, we obtain that any arbitrary configuration $\fB_1 \in \nrere$ can be relaxed to some configuration $\fB_2 \in \neda$, which in turn can be relaxed to some configuration $\fB_3 \in \nrere$.
	From (the maximality condition in) the definition of $\nrere$ it follows that $\fB_1 = \fB_3$, which implies that also $\fB_1 = \fB_2$. 
	Hence, $\fB_1 \in \neda$, which implies that $\nrere \subseteq \neda$, as desired.

	We conclude the proof by showing that also $\neda \subseteq \nrere$.
	We start by observing that for any two nonempty subsets $\ccc, \ccc' \subseteq \ccs$ satisfying $\ccc \neq \ccc'$, the configuration $\gen{\gen{\ell(\ccc)}}^{\Delta - |\ccc| + 1} \s \gen{\gen{\ell(\emptyset)}}^{|\ccc| - 1}$ cannot be relaxed to $\gen{\gen{\ell(\ccc')}}^{\Delta - |\ccc'| + 1} \s \gen{\gen{\ell(\emptyset)}}^{|\ccc'| - 1}$:
	If $\ccc' \subseteq \ccc$, then $|\ccc'| < |\ccc|$, which implies that there are strictly fewer than $|\ccc| - 1$ sets in $\gen{\gen{\ell(\ccc')}}^{\Delta - |\ccc'| + 1} \s \gen{\gen{\ell(\emptyset)}}^{|\ccc'| - 1}$ that are a superset of $\gen{\gen{\ell(\emptyset)}}$, since $\gen{\gen{\ell(\ccc')}} \nsupseteq \gen{\gen{\ell(\emptyset)}}$, by Corollary~\ref{cor:colorregen}.
	If $\ccc' \nsubseteq \ccc$, then $\gen{\gen{\ell(\ccc)}} \nsubseteq \gen{\gen{\ell(\ccc')}}$ (by Corollary~\ref{cor:colorregen}), which implies that no set in $\gen{\gen{\ell(\ccc)}}^{\Delta - |\ccc| + 1} \s \gen{\gen{\ell(\emptyset)}}^{|\ccc| - 1}$ is a subset of $\gen{\gen{\ell(\ccc')}}$, as also $\gen{\gen{\ell(\emptyset)}} \nsubseteq \gen{\gen{\ell(\ccc')}}$ (as seen above).
	
	Hence, no configuration from $\neda$ can be relaxed to a different configuration from $\neda$.
	Recall (again) that each configuration from $\neda$ can be relaxed to some configuration from $\nrere$ and each configuration from $\nrere$ can be relaxed to some configuration from $\neda$.
	It follows that if there is some configuration from $\neda$ that is not contained in $\nrere$, then it can be relaxed to a different configuration from $\neda$, by first relaxing it to a (necessarily different) configuration from $\nrere$ and then relaxing the obtained configuration to some configuration from $\neda$.
	As no configuration from $\neda$ can be relaxed to a different configuration from $\neda$ (as shown above), we conclude that each configuration from $\neda$ is also contained in $\nrere$.
	Thus, $\neda \subseteq \nrere$, as desired.
	It follows that $\nrere = \neda$.
\end{proof}
We now describe what is the label set $\srere$ of the problem $\rerepi$, and we then compute $\erere$.
\begin{lemma}\label{lem:coloredge}
	The set $\srere$ of output labels of $\rerepi$ is given by $\{ \gen{\gen{\ell(\ccc)}} \mid \ccc \in \ccs \}$
	
	Moreover, the hyperedge constraint $\erere$ of $\rerepi$ consists of all configurations $\gen{\gen{\ell(\ccc_1)}} \s \dots \s \gen{\gen{\ell(\ccc_r)}}$ of labels from $\srere$ such that for each color $i \in \ccs$, there is at least one index $1 \leq j \leq r$ satisfying $i \notin \ccc_j$.
\end{lemma}
\begin{proof}
	From Lemma~\ref{lem:colornode} (and the fact that $|\ccs| = \Delta \geq 2$), it follows directly that $\srere$ is as given in the lemma, by the definition of $\srere$.
	Hence, what remains is to show that the hyperedge constraint $\erere$ is as given in the lemma.
	Let $\edda$ denote the set of all configurations $\gen{\gen{\ell(\ccc_1)}} \s \dots \s \gen{\gen{\ell(\ccc_r)}}$ of labels from $\srere$ such that for each color $i \in \ccs$, there is at least one index $1 \leq j \leq r$ satisfying $i \notin \ccc_j$.
	We need to show that $\erere = \edda$.

	We first show that $\edda \subseteq \erere$.
	Let $\gen{\gen{\ell(\ccc_1)}} \s \dots \s \gen{\gen{\ell(\ccc_r)}}$ be an arbitrary configuration from $\edda$.
	Since for each color $i \in \ccs$, there is at least one index $1 \leq j \leq r$ satisfying $i \notin \ccc_j$ (by the definition of $\edda$), there also exists a collection $\ccc'_1, \dots, \ccc'_r$ of subsets of $\ccs$ such that, for each $1 \leq j \leq r$, we have $\ccc_j \subseteq \ccc'_j$ and, for each $i \in \ccs$, there is exactly one index $1 \leq j \leq r$ satisfying $i \notin \ccc'_j$.
	By Corollary~\ref{cor:colorregen}, we know that $(\gen{\ell(\ccc'_1)}, \dots, \gen{\ell(\ccc'_r)}) \in \gen{\gen{\ell(\ccc_1)}} \times \dots \times \gen{\gen{\ell(\ccc_r)}}$; since $\gen{\ell(\ccc'_1)} \s \dots \s \gen{\ell(\ccc'_r)} \in \ere$ (by the definition of the $\ccc'_j$), it follows that $\gen{\gen{\ell(\ccc_1)}} \s \dots \s \gen{\gen{\ell(\ccc_r)}} \in \erere$, by the definition of $\erere$.
	Hence, $\edda \subseteq \erere$, as desired.
	
	Now, we show that $\erere \subseteq \edda$.
	Let $\gen{\gen{\ell(\ccc_1)}} \s \dots \s \gen{\gen{\ell(\ccc_r)}}$ be an arbitrary configuration from $\erere$.
	By the definitions of $\erere$ and $\ere$, we know that there is some choice $(\gen{\ell(\ccc'_1)}, \dots, \gen{\ell(\ccc'_r)}) \in \gen{\gen{\ell(\ccc_1)}} \times \dots \times \gen{\gen{\ell(\ccc_r)}}$ such that, for each color $i \in \ccs$, there is exactly one index $1 \leq j \leq r$ satisfying $i \notin \ccc'_j$.
	Since, for each $1 \leq j \leq r$, $\gen{\ell(\ccc'_j)} \in \gen{\gen{\ell(\ccc_j)}}$ implies $\ccc_j \subseteq \ccc'_j$ (by Corollary~\ref{cor:colorregen}), it follows that, for each color $i \in \ccs$, there is at least one index $1 \leq j \leq r$ satisfying $i \notin \ccc_j$.
	This implies that $\gen{\gen{\ell(\ccc_1)}} \s \dots \s \gen{\gen{\ell(\ccc_r)}} \in \edda$, by the definition of $\edda$.
	Hence, $\erere \subseteq \edda$, and we obtain $\erere = \edda$, as desired.
\end{proof}

\subsection{Renaming}
We now show that, if we rename the labels of $\rerepi$ correctly, we obtain that $\rerepi = \Pi$.
\fponestep*
\begin{proof}
	Consider the following renaming:
	\begin{align*}
	\gen{\gen{\ell(\ccc)}} \text{ for each } \emptyset \neq \ccc \subseteq \ccs &\hspace{0.5cm}\rightarrow\hspace{0.5cm} \ell(\ccc)
	\end{align*}
	Observe that, under this renaming, $\nrere$ becomes equal to $\npi$, and $\erere$ becomes equal to $\epi$.
\end{proof}

\section{Lower Bound for Hypergraph MM}\label{sec:lbmm}
In this section, we prove an $\Omega(\min\{ \Delta r, \log_{\Delta r} n \})$-round deterministic and an $\Omega(\min\{ \Delta r, \log_{\Delta r} \log n \})$-round randomized lower bound for the problem of computing a maximal matching on hypergraphs.

\paragraph{A Sequence of Problems.}
In order to prove a lower bound for hypergraph coloring, it was enough to provide a single problem $\Pi$ that is a relaxation of hypergraph coloring and prove that $\rere(\re(\Pi)) = \Pi$, that is, that $\Pi$ is a fixed point. This approach can only work for ``hard'' problems, that is, problems that cannot be solved in $O(f(\Delta,r) + \log^* n)$ for any function $f$. Instead, for problems solvable with this runtime, like hypergraph maximal matching, we have to follow a different approach. 

The idea is to design a sequence of problems $\Pi_0, \Pi_1, \dots$ such that $\Pi_0$ is essentially a relaxed version of hypergraph MM, and any other problem $\Pi_i$ in the sequence is a relaxation of $\rere(\re(\Pi_{i -1}))$.
Roughly speaking, this implies that each problem in the sequence can be solved at least one round faster than the previous problem---now, all we have to do is to prove that it takes $\Omega(\Delta r)$ steps in the sequence to reach a $0$-round solvable problem.
While this is a simplified outline that disregards certain technicalities (we, e.g., have to ensure that the sizes of the label sets of the problems in the sequence do not grow too fast), there are known techniques in the round elimination framework that take care of these technicalities, i.e., the real challenge lies in designing the aforementioned sequence and proving that it satisfies the desired properties.

\paragraph{A Parametrized Family of Problems.}
In order to be able to prove these properties, it is useful if the problems in the sequence come from some parameterized problem family.
In Section~\ref{sec:probfam}, we describe such a problem family that we will use to prove the lower bound for hypergraph MM---all problems in the desired problem sequence will come from this family.
For technical reasons, we will assume throughout Section~\ref{sec:lbmm} that $\Delta \geq 3$ and $r \geq 3$.
Note that for $\Delta = 2$, hypergraph MM is essentially\footnote{Note that $\Delta = 2$ allows for nodes of degree $1$. Hence by interpreting each node in the hypergraph MM problem as a (hyper)edge and each hyperedge as a node, in the case of $\Delta = 2$ we have a problem that is at least as hard as MM on graphs. It is also easy to see that these additional hyperedges (= nodes in the hypergraph MM problem) of rank $1$ do not really make the MM problem harder as they can simply be taken care of in one round of computation after executing a proper MM algorithm on all rank-$2$ hyperedges. Analogously, in the case of $r = 2$, hypergraph MM formally allows hyperedges of rank $1$, but these do not make the problem (asymptotically) harder as they can be taken care of in one round of computation after executing an MIS algorithm.} the same as MM (on graphs), while for $r = 2$, hypergraph MM is essentially the same as MIS (on graphs).
In particular, the lower bounds from~\cite{Balliu2019} directly imply an $\Omega(\min\{ r, \log_{r} n \})$-round deterministic and an $\Omega(\min\{ r, \log_{r} \log n \})$-round randomized lower bound for hypergraph MM with $\Delta = 2$, and an $\Omega(\min\{ \Delta, \log_{\Delta} n \})$-round deterministic and an $\Omega(\min\{ \Delta, \log_{\Delta} \log n \})$-round randomized lower bound for hypergraph MM with $r = 2$.\footnote{To obtain these lower bounds already on trees, we additionally require~\cite{hideandseek} for the case $r = 2$.}

\paragraph{A Summary of our Goals.}
We proceed as follows. The proof presented for hypergraph coloring can be seen as a simplified version of the proof presented in this section. Hence, we start in \Cref{sec:extension} by highlighting the analogies, and the differences, between the hypergraph maximal matching lower bound proof and the hypergraph coloring lower bound proof.

Then, in \Cref{sec:probfam}, we describe the problem family, characterized by two parameters. We will prove that, for each problem $\Pi$ in the family satisfying some conditions on these parameters, a relaxation of $\rere(\re(\Pi))$ is also in the family. The proof of this fact is going to be a tedious case analysis, and we will defer it to the end of the section.

In \Cref{sec:sublb}, we will then show that, for certain values of the parameters, $\Pi$ cannot be solved in $0$ rounds in the port numbering model. We will then put things together, by showing that we can construct a sequence of problems $\Pi_0,\ldots,\Pi_t$, satisfying that:
\begin{itemize}
	\item All problems are in the family and not $0$-round solvable;
	\item Each problem $\Pi_{i+1}$ is a relaxation of $\rere(\re(\Pi_i))$;
	\item The length of the sequence $t$ is in $\Omega(\Delta r)$;
	\item The first problem is a relaxation of the hypergraph MM problem.
\end{itemize}
We will then get our claimed lower bound by applying \Cref{thm:lifting}.

\subsection{Analogies and Differences with the Hypergraph Coloring Lower Bound}\label{sec:extension}
\paragraph{Recap of the Coloring Relaxations.}
As already mentioned while discussing the hypergraph coloring lower bound, there is some intuition on why the relaxations that we applied on hypergraph coloring are able to give a fixed point. The idea, there, was that, by computing $\rere(\re(\Pi))$, for $\Pi$ equal to the hypergraph coloring problem, we obtain a problem where nodes are allowed to use multiple colors at once, and they are rewarded for using more colors. The reward is that they can mark some amount of incident hyperedges, where the amount of marked hyperedges depends on the amount of colors, and marked hyperedges need to satisfy more relaxed constraints. In $\rere(\re(\Pi))$ (which we did not even present, as it would result in a very unnatural and hard to describe problem) there are many different ways to mark the edges, that give different guarantees. In the relaxation that we presented, we essentially relaxed the constraints so that there is a single possible way to mark the hyperedges.

\paragraph{The Behavior of Hypergraph Maximal Matchng.}
Informally, by applying round elimination on hypergraph maximal matching for $k$ times, what we obtain is essentially a problem that can be decomposed into three parts: the original problem, a hypergraph $k$-colorful coloring, and some unnatural part of much larger size. This can be confirmed experimentally, but since the unnatural part seems to have a size equal to a power tower of height $k$, it makes unfeasible to actually give an explicit form of it.

\paragraph{Our Problem Sequence.}
What we would like to do, is to relax the coloring part of the problem, in order to make the unnatural part disappear. This is essentially the approach used in \cite{hideandseek}. Unfortunately, as explained in \Cref{sec:approach}, if we try to do this, we cannot prevent the colors from growing by $1$ at each step, and hence, while we would like to obtain a sequence of problems of length $\Omega(\Delta r)$, after only $\Delta+1$ steps of round elimination we would have $\Delta+1$ colors, but the fixed point only tolerates up to $\Delta$ colors, and in fact we would obtain a problem that is $0$ round solvable.

Hence, this suggests that, in our problem sequence, we cannot let the colors just grow to $\Omega(\Delta r)$, and that we have to keep them bounded to at most $\Delta$ (for technical reasons, we will actually restrict them to $\Delta-1$). How can we perform $\Omega(\Delta r)$ steps of round elimination and keep the colors bounded to at most $\Delta$, if at each step the obtained problem allows the node to use one additional color? The idea is to perform simplifications that \emph{reduce} the amount of colors: this seems contradictory, how can we make the problem easier if nodes are allowed to use less colors? In order to achieve this, we remove the colors at the cost of relaxing the requirements of the remaining colors from the point of view of the hyperedge constraint. In particular, a problem in our family is described as a vector of length at most $\Delta -1$, where each position represents a color, and for each color we specify how many nodes incident on the same hyperedge are allowed to use that color. The idea is that we can get rid of a color at the cost of increasing some values in this vector.

\paragraph{Comparison with Hypergraph Coloring.}
Summarizing, in order to obtain a fixed point for hypergraph coloring, we had to relax the problem to allow nodes to use sets of colors, and reward nodes for using more colors by allowing nodes to mark hyperedges, such that marked hyperedges are always happy.

In the case of hypergraph maximal matching, we do something very similar. A problem in our family is essentially defined very similarly as hypergraph coloring, and the differences are the following:
\begin{itemize}
	\item We have at most $\Delta-1$ colors;
	\item We also allow the original configurations allowed by hypergraph maximal matching;
	\item There is a vector describing how hard it is to use a color on a hyperedge.
\end{itemize}
Unfortunately, since now the labels of the coloring part and the matching part can mix in the allowed configurations in nontrivial ways, it is not possible to give proper intuition behind some configurations allowed by the hyperedge constraint.

What we are going to prove is that if we take a problem $\Pi$ in this family, we get that $\rere(\re(\Pi))$ can be relaxed to a different problem of the family, where some values in the vector increase. For technical reasons, this is not the full description of the problems in our family: in the problem sequence that we define, at each step, we have to take one of the $\Delta-1$ colors and treat it differently (and the unfamiliar configuration $\D^{\Delta-1} \s \X$ present in the node constraint presented later is related to it).

\subsection{The Problem Family}\label{sec:probfam}
	Each problem $\Pi(z,s)$ in the family is characterized by two parameters $z, s$, where $z = (z_1, \dots, z_k)$ is a vector of $2 \leq k \leq \Delta - 1$ nonnegative integers $z_i \leq r - 1$ and $s \in \{1, \dots, k\}$. 
	We call $k$ the \emph{length} of $z$, denoted by $\len(z)$.
	Intuitively, some labels of $\Pi(z,s)$ can be seen as sets of colors, and the length of $z$ tells us the number of colors, and for each color $i$, parameter $z_i$ describes how often $i$ (or a color set containing $i$) can appear in the same hyperedge configuration.
	Parameter $s$ singles out one of the $\len(z)$ colors that behaves a bit differently than the other colors.
	In the following, we describe $\Pi(z,s)$ formally.

\paragraph{The Label Set.}
	To describe the space of colors discussed above, define $\ccs(z) := \{1, \dots, \len(z)\}$.
	The set $\Sigma(z,s)$ of output labels of $\Pi(z,s)$ is given by $\Sigma(z,s) = \{ \D, \M, \P, \U, \X \} \cup \{ \ell(\ccc) \mid \emptyset \neq \ccc \subseteq \ccs(z)\}$ (where we use the expression $\ell(\ccc)$ instead of $\ccc$ to clearly distinguish between color sets and labels).

\paragraph{The Node Constraint.}
	We denote the node constraint of $\Pi(z,s)$ by $\nodeconst(z,s)$. It is given by the following configurations.
	\begin{itemize}
		\item $\M^{\Delta}$
		\item $\P \s \U^{\Delta - 1}$
		\item $\D^{\Delta - 1} \s \X$
		\item $\ell(\ccc)^{\Delta - |\ccc| + 1} \s \U^{|\ccc| - 1}$ for each $\emptyset \neq \ccc \subseteq \ccs(z)$
	\end{itemize}

\paragraph{The Hyperedge Constraint.}
	We denote the hyperedge constraint of $\Pi(z,s)$ by $\edgeconst(z,s)$. It is given by all configurations $\L_1 \s \dots \s \L_r$ satisfying at least one of the following three conditions.
	\begin{enumerate}
		\item\label{cond:1} There is some index $1 \leq j \leq r$ such that $\L_j = \M$ and $\L_{j'} \notin \{ \D, \M\}$ for each $j' \neq j$.
		\item\label{cond:2} There are two distinct indices $j, j' \in \{ 1, \dots, r \}$ such that $\L_j = \X$, $\L_{j'}$ is arbitrary, and $\L_{j''} \notin \{ \D, \M\}$ for each $j'' \in \{ 1, \dots, r \} \setminus \{ j, j' \}$.
		\item\label{cond:3} All of the following properties hold.
		\begin{enumerate}
			\item\label{prop:a} $\L_j \neq \P$ for each $1 \leq j \leq r$.
			\item\label{prop:b} There is at most one index $j \in \{1, \dots, r\}$ such that $\L_j \in \{ \D, \M\}$.
			\item\label{prop:c} There are at most $z_s$ indices $j \in \{1, \dots, r\}$ such that $\L_j = \D$ or $\L_j = \ell(\ccc)$ for some color set $\ccc$ containing color $s$.
			\item\label{prop:d} For each $i \in \ccs(z)$ satisfying $i \neq s$, there are at most $z_i$ indices $j \in \{1, \dots, r\}$ such that $\L_j = \ell(\ccc)$ for some color set $\ccc$ containing color $i$.
		\end{enumerate}
	\end{enumerate}

	\paragraph{The Relation between the Problems.}
	We will prove that the problems in the defined family are related in the following way.
	\begin{restatable}{lemma}{onestep}\label{lem:onestep}
		Let $z = (z_1,\ldots,z_k)$ be a vector of $2 \le k \le \Delta-1$ nonnegative integers $z_i \le r-1$ and let $s \in \{1,\ldots,k\}$. Let $q$ be an integer satisfying $q \neq s$, $q\in \{1,\ldots,k\}$, and $z_q \le r-2$. Then, the problem  $\rere(\re(\Pi(z,s)))$ can be relaxed to $\Pi(z',q)$, for $z' = (z'_1,\ldots,z'_k)$, where $k = \len(z)$, $z'_i = z_i + 1$ if $i = q$, and $z'_i = z_i$ otherwise.
	\end{restatable}
Since the proof of this statement is a tedious case analysis, we defer its proof to the end of the section. We now use this statement to derive our lower bounds.

	\subsection{Proving the Lower Bounds}\label{sec:sublb}
In this section, we prove our lower bounds.
We start by showing that all problems in our problem family cannot be solved in $0$ rounds in the port numbering model.

\begin{lemma}\label{lem:zerorounds}
	Let $z = (z_1,\ldots,z_k)$ be a vector of $2 \le k \le \Delta-1$ nonnegative integers $z_i \le r-1$, and let $s \in \{1,\ldots,k\}$. The problem $\Pi(z,s)$ cannot be solved in $0$ rounds in the deterministic port numbering model.
\end{lemma}
\begin{proof}
	All nodes that run a $0$-round algorithm have the same view (they just know their degree, the size of the graph, and the parameters $\Delta$ and $r$). Hence, any deterministic $0$-round algorithm uses the same configuration in $\nodeconst(z,s)$ on all nodes. Thus, in order to prove the statement, we consider all possible configurations and we show that they cannot be properly used in a $0$-round algorithm.
	\begin{itemize}
		\item $\M^{\Delta}$: consider $r$ nodes connected to the same hyperedge. They would all output $\M$ on such hyperedge, that would hence have the configuration $\M^{r} \not\in \edgeconst(z,s)$.
		\item $\P \s \U^{\Delta - 1}$: since the algorithm is deterministic, then all nodes must output $\P$ on the same port. W.l.o.g., let this port be the port number $1$. If we connect the port $1$ of $r$ nodes to the same hyperedge, then the hyperedge would have the configuration $\P^{r} \not\in \edgeconst(z,s)$.
		\item $\D^{\Delta - 1} \s \X$: as before, we can obtain a hyperedge labeled with the configuration  $\D^{r} \not\in \edgeconst(z,s)$.
		\item $\ell(\ccc)^{\Delta - |\ccc| + 1} \s \U^{|\ccc| - 1}$ for each $\emptyset \neq \ccc \subseteq \ccs(z)$: since $|\ccc| \le \Delta-1$, then this configuration contains $\ell(\ccc)$ at least twice. Hence, we can obtain a hyperedge labeled with the configuration  $\ell(\ccc)^{r}$, which is not in $\edgeconst(z,s)$ by the assumption that $z_i \le r-1$ for all $i \in \ccc$.
	\end{itemize}
\end{proof}

We are now ready to prove our main result of this section.
\matchinglower*
\begin{proof}
	We show that we can build a problem sequence satisfying \Cref{thm:lifting}, such that the asymptotic complexity of $\Pi_0$ is not higher than the one of the hypergraph maximal matching problem.
	
	Consider a problem sequence $\Pi_0, \Pi_1, \dots, \Pi_{(\Delta - 1)(r - 2)}$ such that $\Pi_0 = \Pi((1, 1, \dots, 1), \Delta - 1)$, $\Pi_{(\Delta - 1)(r - 2)} = \Pi((r - 1, r - 1, \dots, r - 1), \Delta - 1)$, and for any two subsequent problems $\Pi_j, \Pi_{j + 1}$, there exist vectors $z, z'$ of length $\Delta -  1$ and colors $s, s' \in \{ 1, \dots, \Delta - 1 \}$ such that
	\begin{itemize}
		\item $\Pi_j = \Pi(z, s)$ and $\Pi_{j + 1} = \Pi(z', s')$,
		\item $s \neq s'$,
		\item $z'_{s'} = z_{s'} + 1$, and
		\item $z'_i = z_i$ for all $1 \leq i \leq \Delta - 1$ satisfying $i \neq s'$. 
	\end{itemize}
	Clearly, such a problem sequence exists as we can go from the length-($\Delta - 1$) vector $(1, 1, \dots, 1)$ to the length-($\Delta - 1$) vector $(r - 1, r - 1, \dots, r - 1)$ in $(\Delta - 1)(r - 2)$ steps of increasing a single entry (that is a different entry than in the previous step) by $1$ (such that we choose to increase the $(\Delta - 1)$th entry in the last but not the first step).
	
	Observe that, each problem of the sequence is not $0$-round solvable in the port numbering model by \Cref{lem:zerorounds}. Also, the number of labels of each problem is bounded by $2^{O(\Delta r)}$. Also, by \Cref{lem:onestep}, each problem $\Pi_{i+1}$ is a relaxation of $\rere(\re(\Pi_i))$, and also the intermediate problem, $\re(\Pi_i)$, has a number of labels bounded by $2^{O(\Delta r)}$. Hence, \Cref{thm:lifting} applies, and since $\log_{\Delta r} \log f(\Delta r) = O(1)$, we obtain a lower bound for $\Pi_0$ of  $\Omega(\min\{\Delta \cdot r,\log_{\Delta r} n\})$ in the \LOCAL model for deterministic algorithms, and $\Omega(\min\{\Delta \cdot r,\log_{\Delta r} \log n\})$ for randomized ones.
	
	We now prove that, given an algorithm $A$ with complexity $T$ for solving hypergraph maximal matching, we can turn it into an algorithm with complexity $O(T)$ for solving $\Pi_0$, implying the claim.
	In order to solve $\Pi_0$ on some hypergraph $H$, we start by simulating $A$ in the hypergraph $H'$ obtained by reversing the role of nodes and hyperedges of $H$. This can be performed in $O(T)$ rounds. Then, given a solution for the hypergraph maximal matching problem on $H'$, we solve $\Pi_0$ as follows: 
	\begin{itemize}
		\item Nodes of $H$ (that is, hyperedges of $H'$) that are in the matching, output $\M^\Delta$.
		\item Each node $v$ not in the matching must have at least one neighboring hyperedge $h$ that is incident to some other node $u$ in the matching. Node $v$ outputs $\P$ on $h$ and $\U$ on all the other incident hyperedges.
	\end{itemize}
	Observe that, on the nodes, we only use configurations $\M^\Delta$ and $\P \U^{\Delta-1}$, which are allowed by $\Pi_0$. Then, on each hyperedge, we obtain that there is at most one $\M$, and that a $\P$ is present only if there is also an $\M$, while all the other labels are $\U$. Hence, also on the hyperedges we obtain only configurations allowed by $\Pi_0$.
\end{proof}

\subsection{Relations Between Labels}
In the rest of the section we prove \Cref{lem:onestep}. We start by proving a relation between the labels.
	From the description of the hyperedge constraint, we can extract the strength relations w.r.t.\ $\edgeconst(z,s)$ by simply checking for any two labels $\L, \L' \in \Sigma(z,s)$ whether every hyperedge configuration containing $\L$ remains a configuration in $\edgeconst(z,s)$ if we replace $\L$ by $\L'$. An example is shown in \Cref{fig:diag1}.
	
\begin{figure}[h]
	\centering
	\includegraphics[width=0.9\textwidth]{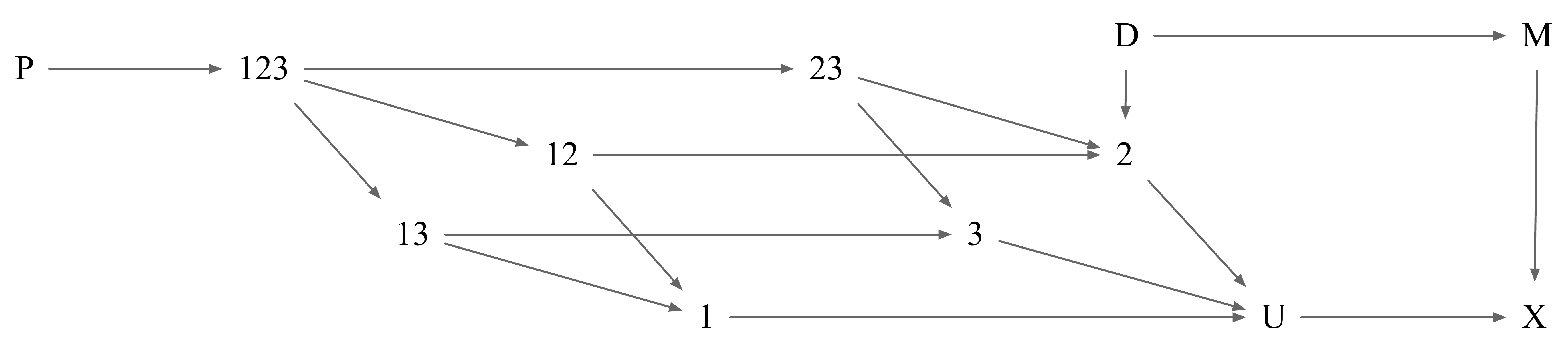}
	\caption{The edge diagram of the problems in the family satisfying $s = 2$ and $\len(z) = 3$. Numbers correspond to color sets, and $\ell(\cdot)$ is omitted for clarity.}
	\label{fig:diag1}
\end{figure} 

	\begin{lemma}\label{lem:pistrong}
		The following collection lists all strength relations (according to $\edgeconst(z,s)$) between distinct labels in $\Sigma(z,s)$.
		\begin{itemize}
			\item $\P < \ell(\ccc)$ for each $\emptyset \neq \ccc \subseteq \ccs(z)$
			\item $\P < \U$
			\item $\ell(\ccc) < \ell(\ccc')$ for any two $\ccc, \ccc' \in 2^{\ccs(z)} \setminus \{ \emptyset \}$ satisfying $\ccc' \subsetneq \ccc$
			\item $\ell(\ccc) < \U$ for each $\emptyset \neq \ccc \subseteq \ccs(z)$
			\item $\D < \M$
			\item $\D < \ell(\{ s \})$
			\item $\D < \U$
			\item $\L < \X$ for each $\L \in \Sigma(z,s) \setminus \{ \X \}$
		\end{itemize}
	\end{lemma}
	
	\begin{proof}
		We start by comparing the strength of $\X$ with the strength of all other labels.
		The definition of $\edgeconst(z,s)$ implies that in any configuration $\L_1 \s \dots \s \L_r \in \edgeconst(z,s)$ there is at most one $\L_j$ with $\L_j \in \{ \D, \M \}$.
		Hence, replacing an arbitrary $\L_{j'}$ with $\X$ will yield a configuration from $\edgeconst(z,s)$, by Condition~\ref{cond:2}, and we obtain $\L \leq \X$ for each $\L \in \Sigma(z,s) \setminus \{ \X \}$.
		To show that $\X \nleq \L$ for each $\L \in \Sigma(z,s) \setminus \{ \X \}$, it suffices to observe that $\X \s \M \s \P^{r - 2} \in \ezs$, which implies $\X \nleq \M$ (as replacing $\X$ with $\M$ does not yield a configuration from $\ezs$), and $\X \s \P^{r - 1} \in \ezs$, which implies $\X \nleq \L$ for each $\L \in \Sigma(z,s) \setminus \{ \M, \X \}$.

		Next, we compare the strength of $\P$ with the strength of the labels from $\szs \setminus \{ \P, \X \}$.
		From the definition of $\edgeconst(z,s)$, we can see that the only configurations containing $\P$ are of the form $\P \s \M \s \L_1 \s \dots \s \L_{r - 2}$ or $\P \s \X \s \L_1 \s \dots \s \L_{r - 2}$ where, in the former case, $\L_j \notin \{ \D, \M \}$ for each $1 \leq j \leq r - 2$, and, in the latter case, $\L_j \notin \{ \D, \M \}$ for each $2 \leq j \leq r - 2$.
		In either case, replacing $\P$ with some arbitrary label from $\szs \setminus \{ \D, \M \}$ yields a configuration contained in $\ezs$, by Conditions~\ref{cond:1} and \ref{cond:2}.
		It follows that $\P \leq \L$ for each $\L \in \szs \setminus \{ \D, \M \}$.
		Moreover, the configuration $\M \s \P^{r - 1} \in \ezs$ certifies that $\P \nleq \D$ and $\P \nleq \M$ since neither of the two configurations $\M \s \D \s \P^{r - 2}$ and $\M^2 \s \P^{r - 2}$ is contained in $\ezs$, by Conditions~\ref{cond:1}, \ref{cond:2} and \ref{prop:b}.
		In order to show that $\L \nleq \P$ for each $\L \in \szs \setminus \{ \P, \X \}$, it suffices to observe that $\L \s \U^{r - 1} \in \ezs$ for each $\L \in \szs \setminus \{ \P, \X \}$ (by Condition~\ref{cond:3}), while $\P \s \U^{r - 1} \notin \ezs$ (by Conditions~\ref{cond:1}, \ref{cond:2} and \ref{prop:a}).

		We continue by comparing the strength of $\D$ and the strength of $\M$ with each other and with the strength of the labels from $\szs \setminus \{ \D, \M, \P, \X \}$.
		As already observed above, in any configuration $\L_1 \s \dots \s \L_r \in \ezs$ there is at most one $\L_j$ with $\L_j \in \{ \D, \M \}$.
		Hence, we obtain $\D \s \M \s \U^{r - 2} \notin \ezs$ and $\M^2 \s \U^{r - 2} \notin \ezs$.
		Since, by Condition~\ref{cond:1}, $\L \s \M \s \U^{r - 2} \in \ezs$ for each $\L \in \szs \setminus \{ \D, \M\}$, it follows that $\L \nleq \L'$ for each $\L \in \szs \setminus \{ \D, \M\}$ and $\L' \in \{ \D, \M \}$.
		Moreover, since $\M \s \P^{r - 1} \in \ezs$ by Condition~\ref{cond:1}, but $\L \s \P^{r - 1} \notin \ezs$ for each $\L \in \szs \setminus \{ \M, \P, \X \}$ by Conditions~\ref{cond:1}, \ref{cond:2} and \ref{prop:a}, we obtain $\M \nleq \L$ for each $\L \in \szs \setminus \{ \M, \P, \X \}$.

		To show that $\D \nleq \L$ for each $\L \in \szs \setminus \{ \D, \M, \P, \U, \X, \ell(\{ s \}) \}$, assume that $\szs \setminus \{ \D, \M, \P, \U, \X, \ell(\{ s \}) \} \neq \emptyset$ (otherwise we are done) and consider some arbitrary $\L \in \szs \setminus \{ \D, \M, \P, \U, \X, \ell(\{ s \}) \}$.
		By the definition of $\szs$, it follows that $\L = \ell(\ccc)$ for some $\ccc \in 2^{\ccs(z)} \setminus \{ \emptyset, \{ s \} \}$.
		Hence, there exists some $i \in \ccs(z)$ satisfying $i \neq s$ such that $i \in \ccc$; consider such an $i$.
		Consider the configuration $\D \s \ell(\{ i \})^{z_i} \s \U^{r - z_i - 1}$ (which is well-defined as the definition of problem $\Pi(z,s)$ specifies that $z_i \leq r - 1$).
		By Condition~\ref{cond:3}, it is contained in $\ezs$ whereas the configuration $\ell(\ccc) \s \ell(\{ i \})^{z_i} \s \U^{r - z_i - 1}$ is not contained in $\ezs$, by Conditions~\ref{cond:1}, \ref{cond:2} and \ref{prop:d}.
		Therefore, $\D \nleq \ell(\ccc)$, and we obtain that $\D \nleq \L$ for each $\L \in \szs \setminus \{ \D, \M, \P, \U, \X, \ell(\{ s \}) \}$.

		To show that $\D \leq \L$ for each $\L \in \{ \M, \U, \ell(\{ s \}) \}$, it suffices to observe that if some configuration $\L_1 \s \dots \s \L_r$ that contains $\D$ satisfies one of the six conditions (\ref{cond:1}, \ref{cond:2}, \ref{prop:a}, \ref{prop:b}, \ref{prop:c}, \ref{prop:d}), then the same configuration where $\D$ is replaced by $\L$ satisfies the same condition, for each $\L \in \{ \M, \U, \ell(\{ s \}) \}$.

		It remains to compare the strength of the labels in $\szs \setminus \{ \D, \M, \P, \U, \X \}$ with each other and with the strength of $\U$.
		Similarly to before, it is straightforward to check that if some configuration $\L_1 \s \dots \s \L_r$ that contains some label $\ell(\ccc)$ satisfies one of the six conditions (\ref{cond:1}, \ref{cond:2}, \ref{prop:a}, \ref{prop:b}, \ref{prop:c}, \ref{prop:d}), then the same configuration where $\ccc$ is replaced by $\U$ satisfies the same condition (irrespective of the choice of $\ccc$).
		Hence, $\L \leq \U$ for each $\L \in \szs \setminus \{ \D, \M, \P, \U, \X \}$.

		In order to show that $\U \nleq \L$ for each $\L \in \szs \setminus \{ \D, \M, \P, \U, \X \}$, consider an arbitrary label $\L \in \szs \setminus \{ \D, \M, \P, \U, \X \}$.
		By the definition of $\szs$, we have $\L = \ell(\ccc)$ for some $\emptyset \neq \ccc \subseteq \ccs(z)$.
		Let $i$ be some arbitrary color in $\ccc$, and consider the configuration $\ell(\{ i \})^{z_i} \s \U^{r - z_i}$ (which contains at least one $\U$ since $z_i \leq r - 1$).
		By Condition~\ref{cond:3}, this configuration is contained in $\ezs$ whereas the configuration $\ell(\ccc) \s \ell(\{ i \})^{z_i} \s \U^{r - z_i - 1}$ (obtained by replacing one $\U$ by $\ell(\ccc)$) is not contained in $\ezs$, by Conditions~\ref{cond:1}, \ref{cond:2} and one of \ref{prop:c}, \ref{prop:d} (depending on whether $i = s$ or $i \neq s$).
		Hence, $\U \nleq \L$ for each $\L \in \szs \setminus \{ \D, \M, \P, \U, \X \}$.

		Finally, consider two distinct labels $\L, \L' \in \szs \setminus \{ \D, \M, \P, \U, \X \}$.
		As above, we know that $\L = \ell(\ccc)$ and $\L' = \ell(\ccc')$ for some $\ccc, \ccc' \in 2^{\ccs(z)} \setminus \{ \emptyset \}$ satisfying $\ccc \neq \ccc'$.
		Consider first the case that $\ccc \nsubseteq \ccc'$, which implies that there exists some color $i \in \ccc \setminus \ccc'$.
		Similarly to before, we can observe that $\ell(\ccc') \s \ell(\{ i \})^{z_i} \s \U^{r - z_i - 1} \in \ezs$ by Condition~\ref{cond:3} and the fact that $i \notin \ccc'$, whereas $\ell(\ccc) \s \ell(\{ i \})^{z_i} \s \U^{r - z_i - 1} \notin \ezs$ by Condition~\ref{prop:c} or \ref{prop:d} (depending on whether $i = s$ or $i \neq s$).
		Hence, $\ell(\ccc') \nleq \ell(\ccc)$.

		Now, consider the other case, namely that $\ccc \subsetneq \ccc'$.
		Again, we observe that if some configuration $\L_1 \s \dots \s \L_r$ that contains $\ccc'$ satisfies one of the six conditions (\ref{cond:1}, \ref{cond:2}, \ref{prop:a}, \ref{prop:b}, \ref{prop:c}, \ref{prop:d}), then the same configuration where the $\ccc'$ is replaced by $\ccc$ satisfies the same condition.
		Thus, $\ell(\ccc') \leq \ell(\ccc)$.
		By symmetry, we obtain $\ell(\ccc) \nleq \ell(\ccc')$ if $\ccc' \nsubseteq \ccc$, and $\ell(\ccc) \leq \ell(\ccc')$ if $\ccc' \subsetneq \ccc$.

		The above discussion shows that the strength relations according to $\ezs$ are exactly those listed in the lemma.
	\end{proof}

	Using Lemma~\ref{lem:pistrong} we can compute the set $\gen{\L}$ for each label $\L \in \szs$.

	\begin{corollary}\label{cor:pigen}
		We have
		\begin{itemize}
			\item $\gen{\X} = \{ \X \}$,
			\item $\gen{\M} = \{ \M, \X \}$,
			\item $\gen{\U} = \{ \U, \X \}$,
			\item $\gen{\D} = \{ \D, \M, \U, \X, \ell(\{ s \}) \}$,
			\item $\gen{\P} = \{ \P, \U, \X \} \cup \{ \ell(\ccc) \mid \emptyset \neq \ccc \subseteq \ccs(z) \}$, and
			\item $\gen{\ell(\ccc)} = \{ \U, \X \} \cup \{ \ell(\ccc') \mid \emptyset \neq \ccc' \subseteq \ccc \}$ for each $\emptyset \neq \ccc \subseteq \ccs(z)$.
		\end{itemize}
	\end{corollary}

\subsection{Computing $\re(\Pi(z,s))$}
	After defining our problem family and collecting some basic facts about it, the next step is to examine how the problems from this family behave under applying $\re(\cdot)$.
	Specifically, in Lemma~\ref{lem:repidef}, we compute $\re(\Pi(z,s))$.
	To this end, for any $\L \in \szs \setminus \{ \D, \M \}$, define
	\[
		\gen{\L}' :=
		\begin{cases}
			\gen{\L, \D} &\text{ if } \L \leq \ell(\{ s \}), \\
			\gen{\L, \M} &\text{ if } \L \nleq \ell(\{ s \}).
		\end{cases}
	\]

	\begin{lemma}\label{lem:repidef}
		The set $\sre$ of output labels of $\repizs$ is given by
		\[
			\{ \gen{\L}, \gen{\L}' \mid \L \in \szs \setminus \{ \D, \M \} \}.
		\]
		Moreover, $\gen{\L} \neq \gen{\L'}'$ for any $\L, \L' \in \szs \setminus \{ \D, \M \}$.

		The hyperedge constraint $\ere$ of $\repizs$ consists of all configurations $\gen{\L_1}' \s \gen{\L_2} \s \gen{\L_3} \dots \gen{\L_r}$ satisfying $\L_j \in \szs \setminus \{ \D, \M \}$ for all $1 \leq j \leq r$ and at least one of the following two conditions.
		\begin{enumerate}
			\item\label{cond:easy} There exists some index $1 \leq j \leq r$ such that $\L_j = \X$ and $\L_{j'} = \P$ for each $j' \neq j$.
			\item\label{cond:hard} Both of the following properties hold.
			\begin{enumerate}
				\item\label{prop:p} $\L_j \notin \{ \P, \X \}$ for each $1 \leq j \leq r$.
				\item\label{prop:i} For each $i \in \ccs(z)$, there are exactly $z_i$ indices $j \in \{1, \dots, r\}$ such that $\L_j = \ell(\ccc)$ for some color set $\ccc$ containing color $i$.
			\end{enumerate}
		\end{enumerate}
		The node constraint $\nre$ of $\repizs$ consists of all configurations $\B_1 \s \dots \s \B_{\Delta}$ of labels from $\{ \gen{\L}, \gen{\L}' \mid \L \in \szs \setminus \{ \D, \M \} \}$ such that there exists a choice $(\A_1, \dots, \A_{\Delta}) \in \B_1 \times \dots \times \B_{\Delta}$ satisfying $\A_1 \s \dots \s \A_{\Delta} \in \nzs$.
	\end{lemma}

	\begin{proof}
		Let $\edda$ denote the set of all configurations $\gen{\L_1}' \s \gen{\L_2} \s \gen{\L_3} \dots \gen{\L_r}$ with labels from $\szs \setminus \{ \D, \M \}$ that satisfy the two conditions given in the lemma.
		Recall Definition~\ref{def:relaxing}.
		We start by showing that the hyperedge constraint $\ere$ of $\repizs$ is equal to $\edda$.
		To this end, by the definition of $\ere$, it suffices to show that the following three conditions hold.
		\begin{description}
			\item[(Z1)] For each configuration $\B_1 \s \dots \s \B_r \in \edda$ and each choice $(\A_1, \dots, \A_r) \in \B_1 \times \dots \times \B_r$, we have $\A_1 \s \dots \s \A_r \in \ezs$.
			\item[(Z2)] For each configuration $\B_1 \s \dots \s \B_r$ with labels from $2^{\szs} \setminus \{ \emptyset \}$ that cannot be relaxed to any configuration from $\edda$, there exists a choice $(\A_1, \dots, \A_r) \in \B_1 \times \dots \times \B_r$ satisfying $\A_1 \s \dots \s \A_r \notin \ezs$.
			\item[(Z3)] If some configuration $\B_1 \s \dots \s \B_r \in \edda$ can be relaxed to some configuration $\overline{\B_1} \s \dots \s \overline{\B_r} \in \edda$, then $\B_1 \s \dots \s \B_r$ and $\overline{\B_1} \s \dots \s \overline{\B_r}$ are identical (up to permutation).
		\end{description}
		
		We first show that (Z1) holds.
		Consider an arbitrary configuration $Y = \gen{\L_1}' \s \gen{\L_2} \s \gen{\L_3} \dots \gen{\L_r}$ from $\edda$.
		If $Y$ satisfies Condition~\ref{cond:easy} in Lemma~\ref{lem:repidef}, then $Y = \gen{\X, \M} \s \gen{\P}^{r - 1} = \gen{\M} \s \gen{\P}^{r - 1}$ or $Y = \gen{\P, \D} \s \gen{\X} \s \gen{\P}^{r - 2}$ (where we use\footnote{Recall also that all permutations of a configuration are considered to be the same configuration since, formally, they are multisets.} Corollary~\ref{cor:pigen} and the fact that $\X \nleq \ell(\{ s \})$ and $\P \leq \ell(\{ s \})$, by Lemma~\ref{lem:pistrong}).
		By the definition of $\gen{\cdot}$, it follows that $Y$ can be obtained from at least one member of $\{ \M \s \P^{r - 1}, \X \s \P^{r - 1}, \D \s \X \s \P^{r - 2} \}$ by repeatedly choosing some label in the configuration and replacing it by a stronger one.
		Since, by Conditions~\ref{cond:1} and \ref{cond:2} in the definition of $\ezs$, we know that $\M \s \P^{r - 1}$, $\X \s \P^{r - 1}$ and $\D \s \X \s \P^{r - 2}$ are all contained in $\ezs$, we can conclude, by the definition of strength, that (Z1) is satisfied for each configuration $Y$ satisfying Condition~\ref{cond:easy} in Lemma~\ref{lem:repidef}.
		
		If, on the other hand, $Y$ satisfies Condition~\ref{cond:hard} in Lemma~\ref{lem:repidef}, then, in order to show that (Z1) holds, it suffices to show that $\D \s \L_2 \s \dots \s \L_r \in \ezs$ if $\L_1 \leq \ell(\{ s \})$, $\M \s \L_2 \s \dots \s \L_r \in \ezs$ if $\L_1 \nleq \ell(\{ s \})$, and $\L_1 \s \dots \L_r \in \ezs$, by the same argumentation as in the previous case.
		Each of these three configurations (under the respective condition) satisfies Conditions~\ref{prop:a}, \ref{prop:b}, and \ref{prop:d} in the definition of $\ezs$, by Condition~\ref{prop:p} in Lemma~\ref{lem:repidef}, the fact that $\L_j \in \szs \setminus \{ \D, \M \}$ for all $2 \leq j \leq r$, and Condition~\ref{prop:i} in Lemma~\ref{lem:repidef}, respectively.
		Moreover, we claim that it follows from Condition~\ref{prop:i} (and Condition~\ref{prop:p}) in Lemma~\ref{lem:repidef} that each of these three configurations (under the respective condition) satisfies also Condition~\ref{prop:c} in the definition of $\ezs$.
		For $\M \s \L_2 \s \dots \s \L_r$ and $\L_1 \s \dots \L_r$ this is immediate; for $\D \s \L_2 \s \dots \s \L_r$ under the condition $\L_1 \leq \ell(\{ s \})$, the argumentation is a bit more involved:
		Observe that $\L_1 \leq \ell(\{ s \})$, together with Condition~\ref{prop:p} and Lemma~\ref{lem:pistrong}, implies that $\L_1 = \ell(\ccc)$ for some $\ccc$ containing color $s$.
		By Condition~\ref{prop:i} in Lemma~\ref{lem:repidef}, it follows that there are exactly $z_s - 1$ indices $j \in \{ 2, \dots, r \}$ such that $\L_j = \ell(\ccc')$ for some $\ccc'$ containing color $s$, which in turn implies that $\D \s \L_2 \s \dots \s \L_r$ satisfies Condition~\ref{prop:c} in the definition of $\ezs$, proving the claim.
		Hence, all of $\D \s \L_2 \s \dots \s \L_r$, $\M \s \L_2 \s \dots \s \L_r$, and $\L_1 \s \dots \L_r$ (under the mentioned respective conditions) satisfy Condition~\ref{cond:3} in the definition of $\ezs$, which implies that they are contained in $\ezs$.
		As shown, (Z1) follows.

		Next, we prove (Z2).
		Let $Y = \B_1 \s \dots \s \B_r$ be a configuration with labels from $2^{\szs} \setminus \{ \emptyset \}$ that cannot be relaxed to any configuration from $\edda$.
		For a contradiction, assume that (Z2) does not hold, which implies that for each choice $(\A_1, \dots, \A_r) \in \B_1 \times \dots \times \B_r$ we have $\A_1 \s \dots \s \A_r \in \ezs$.
		Observe that adding to some set $\B_j$ all labels $\L \in \szs \setminus \B_j$ that satisfy $\L' \leq \L$ for some $\L' \in \B_j$ does not change whether there exists a choice $(\A_1, \dots, \A_r) \in \B_1 \times \dots \times \B_r$ satisfying $\A_1 \s \dots \s \A_r \notin \ezs$, by the definition of strength.
		As adding labels to the sets $\B_j$ also cannot make $Y$ relaxable to some configuration it could not be relaxed to before the addition, we can (and will) therefore assume that $\B_j$ is right-closed for each $1 \leq j \leq r$.
		We will now collect some properties of $Y$ that we derive from the above knowledge and assumptions regarding $Y$.
		
		From the definition of $\ezs$ and the assumption that each choice $(\A_1, \dots, \A_r) \in \B_1 \times \dots \times \B_r$ satisfies $\A_1 \s \dots \s \A_r \in \ezs$, it follows directly that there is at most one index $1 \leq j \leq r$ such that $\M \in \B_j$.
		W.l.o.g., let $1$ be this index (if it exists), i.e., we have
		\begin{equation}\label{eq:nom}
			\M \notin \B_j \text{ for all } 2 \leq j \leq r,
		\end{equation}
		which, by Corollary~\ref{cor:pigen} and the right-closedness of the $\B_j$, implies
		\begin{equation}\label{eq:pgen}
			\B_j \subseteq \gen{\P} \text{ for all } 2 \leq j \leq r.
		\end{equation} 
		Now, consider the configuration $\gen{\M} \s \gen{\P}^{r - 1} = \gen{\X, \M} \s \gen{\P}^{r - 1} = \gen{\X}' \s \gen{\P}^{r - 1}$, which is contained in $\edda$, by Condition~\ref{cond:easy} in Lemma~\ref{lem:repidef}.
		Since $Y$ cannot be relaxed to any configuration from $\edda$, we obtain $\B_1 \nsubseteq \gen{\M}$, by (\ref{eq:pgen}).

		Moreover, if $\B_j = \gen{\X}$ for some $2 \leq j \leq r$, then $\B_j \subseteq \gen{\X}$, which, combined with $\B_1 \subseteq \szs = \gen{\P}'$ (which follows from Lemma~\ref{lem:pistrong} and Corollary~\ref{cor:pigen}) and $\B_{j'} \subseteq \gen{\P}$ for all $j' \in \{ 2, \dots, r \} \setminus \{ j \}$ (which follows from (\ref{eq:pgen})), would imply that $Y$ can be relaxed to $\gen{\P}' \s \gen{\X} \s \gen{\P}^{r - 2} \in \edda$, yielding a contradiction.
		Hence, $\B_j \neq \gen{\X}$ for all $2 \leq j \leq r$, which, together with (\ref{eq:nom}), Corollary~\ref{cor:pigen}, and the already established fact that $\B_1 \nsubseteq \gen{\M}$, implies
		\begin{equation}\label{eq:outm}
			\B_j \nsubseteq \{ \M, \X \} \text{ for all } 1 \leq j \leq r.
		\end{equation} 
		It follows that if $\P \in \B_j$ for some $1 \leq j \leq r$, then there exists a choice $(\A_1, \dots, \A_r) \in \B_1 \times \dots \times \B_r$ such that $\A_j = \P$ and $\A_{j'} \notin \{ \M, \X \}$ for all $j' \neq j$, which would imply that $\A_1 \s \dots \s \A_r \notin \ezs$ (by Conditions~\ref{cond:1}, \ref{cond:2} and \ref{prop:a} in the definition of $\ezs$), yielding a contradiction.
		Hence,
		\begin{equation}\label{eq:nop}
			\P \notin \B_j \text{ for all } 1 \leq j \leq r.
		\end{equation}
		While, so far, we only made use of Conditions~\ref{cond:1}, \ref{cond:2} and \ref{prop:a} in the definition of $\ezs$ and the fact that $Y$ cannot be relaxed to any configuration satisfying Condition~\ref{cond:easy} in Lemma~\ref{lem:repidef}, we will now also take advantage of the other conditions in the definition of $\ezs$ and Lemma~\ref{lem:repidef}.
		
		For each $1 \leq j \leq r$ let $S_j$ be the set of all colors $i \in \ccs(z)$ satisfying $\ell(\{ i \}) \in \B_j$.
		We claim that
		\begin{equation}\label{eq:atmost}
			\text{ for each $i \in \ccs(z)$ there are at most $z_i$ indices $j \in \{ 1, \dots, r \}$ such that $i \in S_j$.}
		\end{equation}
		If this was not true for some $i \in \ccs(z)$, then, by (\ref{eq:outm}) and the definition of the $S_j$, it would be possible to pick one label from each $\B_j$ such that the resulting configuration is $\ell(\{ i \})^{z_i + 1} \s \L_{z_i + 2} \s \dots \s \L_r$ where $\L_{j'} \notin \{ \M, \X \}$ for each $z_i + 2 \leq j' \leq r$.
		By Conditions ~\ref{cond:1}, \ref{cond:2}, \ref{prop:c} and \ref{prop:d} in the definition of $\ezs$, this configuration is not contained in $\ezs$, yielding a contradiction and proving the claim.

		Consider $\B_j$ for some arbitrary index $2 \leq j \leq r$.
		By the right-closedness of $\B_j$, Corollary~\ref{cor:pigen}, and the definition of $S_j$, we know that $\ell(\ccc) \notin \B_j$ for each $\ccc \nsubseteq S_j$.
		By (\ref{eq:pgen}), (\ref{eq:nop}) and Corollary~\ref{cor:pigen}, it follows that
		\begin{equation}\label{eq:insj2}
			\B_j \subseteq \gen{\ell(S_j)} \text{ for all } 2 \leq j \leq r,
		\end{equation}
		where we set $\ell(\emptyset) := \U$.
		Now consider $\B_1$.
		Using an analogous argumentation we obtain that $\ell(\ccc) \notin \B_1$ for each $\ccc \nsubseteq S_1$ and, if $s \notin S_1$, additionally that $\D \notin \B_1$.
		By (\ref{eq:nop}), Corollary~\ref{cor:pigen}, and the observation that $s \notin S_1$ if and only if $\ell(S_1) \nleq \ell(\{ s \})$, it follows that
		\begin{equation}\label{eq:insj1}
			\B_1 \subseteq \gen{\ell(S_1)}'.
		\end{equation}
		
		Let $S'_1, \dots, S'_r$ be a collection of subsets of $\ccs(z)$ such that $S_j \subseteq S'_j$ for all $1 \leq j \leq r$ and for each $i \in \ccs(z)$ there are exactly $z_i$ indices $j \in \{ 1, \dots, r \}$ such that $i \in S'_j$.
		Such a collection exists by (\ref{eq:atmost}) and the fact that $z_i \leq r - 1$.
		By Corollary~\ref{cor:pigen} and the fact that $S_j \subseteq S'_j$ for each $1 \leq j \leq r$, we know that $\gen{\ell(S_j)} \subseteq \gen{\ell(S'_j)}$ and $\gen{\ell(S_j)}' \subseteq \gen{\ell(S'_j)}'$, for each $1 \leq j \leq r$.
		By (\ref{eq:insj2}) and (\ref{eq:insj1}), it follows that $\B_1 \subseteq \gen{\ell(S'_1)}'$ and $\B_j \subseteq \gen{\ell(S'_j)}$ for each $2 \leq j \leq r$, which in turn implies that $Y$ can be relaxed to the configuration $\gen{\ell(S'_1)}' \s \gen{\ell(S'_2)} \s \gen{\ell(S'_3)} \s \dots \s \gen{\ell(S'_r)}$.
		As $i \in S'_j$ if and only if $\ell(S'_j) = \ell(\ccc)$ for some color set $\ccc$ containing color $i$, the definition of the $S'_j$ implies that $\gen{\ell(S'_1)}' \s \gen{\ell(S'_2)} \s \gen{\ell(S'_3)} \s \dots \s \gen{\ell(S'_r)}$ satisfies Condition~\ref{cond:hard} in Lemma~\ref{lem:repidef}.
		Since also $\ell(S'_j) \in \szs \setminus \{ \D, \M \}$, we obtain $\gen{\ell(S'_1)}' \s \gen{\ell(S'_2)} \s \gen{\ell(S'_3)} \s \dots \s \gen{\ell(S'_r)} \in \ere$, yielding a contradiction to the fact that $Y$ cannot be relaxed to any configuration from $\edda$.
		Hence, (Z2) follows.

		Now, we prove (Z3).
		For a contradiction, assume that (Z3) does not hold.
		Then, there exist two configurations $Y = \B_1 \s \dots \s \B_r \in \edda$ and $\overline{Y} = \overline{\B_1} \s \dots \s \overline{\B_r} \in \edda$ such that $\B_j \subseteq \overline{\B_j}$ for each $1 \leq j \leq r$ and there exists some index $j$ satisfying $\B_j \subsetneq \overline{\B_j}$.
		W.l.o.g., assume that $1$ is such an index, i.e., we have $\B_1 \subsetneq \overline{\B_1}$.
		Observe that $\overline{\B_1}$ is right-closed, by the definition of $\edda$.
		
		Consider first the case that $Y$ satisfies Condition~\ref{cond:easy} in Lemma~\ref{lem:repidef}.
		Since $\gen{\P}$ and $\gen{\P}'$ both contain $\P$, and $\P \notin \gen{\L}$ as well as $\P \notin \gen{\L}'$ for each $\L \in \szs \setminus \{ \P \}$, it follows that also $\overline{Y}$ satisfies Condition~\ref{cond:easy}, by Condition~\ref{prop:p}.
		Observe further that $\gen{\X} \subsetneq \gen{\X}' \subsetneq \gen{\P}'$, $\gen{\X} \subsetneq \gen{\P} \subsetneq \gen{\P}'$, $\gen{\P} \nsubseteq \gen{\X}'$, and $\gen{\X}' \nsubseteq \gen{\P}$, by Corollary~\ref{cor:pigen}.
		Hence, it holds for all $1 \leq j \leq r$ that if $\B_j \in \{ \gen{\P}, \gen{\P}' \}$, then $\overline{\B_j} \in \{ \gen{\P}, \gen{\P}' \}$, and if $\B_j \in \{ \gen{\X}', \gen{\P}' \}$, then $\overline{\B_j} \in \{ \gen{\X}', \gen{\P}' \}$.
		As, by Condition~\ref{cond:easy}, both $Y$ and $\overline{Y}$ contain some element from $\{ \gen{\X}', \gen{\P}' \}$ in exactly one position and some element from $\{ \gen{\P}, \gen{\P}' \}$ in exactly $r - 1$ positions, it follows that $\B_j = \overline{\B_j}$ for each $1 \leq j \leq r$, yielding a contradiction to $\B_1 \subsetneq \overline{\B_1}$.

		Consider now the other case, namely that $Y$ satisfies Condition~\ref{cond:hard} in Lemma~\ref{lem:repidef}.
		We start by observing that sets of the form $\gen{\L}$ for some $\L \in \szs \setminus \{ \D, \M \}$ do not contain the label $\M$ whereas sets of the form $\gen{\L}'$ for some $\L \in \szs \setminus \{ \D, \M \}$ contain the label $\M$.
		In particular, $Y$ and $\overline{Y}$ contain a set of the form $\gen{\L}'$ for some $\L \in \szs \setminus \{ \D, \M \}$ in exactly one position, and the index specifying the position is the same for $Y$ and $\overline{Y}$.
		
		For each $1 \leq j \leq r$, let $\L_j \in\szs$ be the (unique) label satisfying $\B_j = \gen{\L_j}$ or $\B_j = \gen{\L_j}'$, and $\overline{\L_j} \in \szs$ the (unique) label satisfying $\overline{\B_j} = \gen{\overline{\L_j}}$ or $\overline{\B_j} = \gen{\overline{\L_j}}'$.
		Since $\B_1 \subsetneq \overline{\B_1}$, we have $\overline{\L_1} < \L_1$, by Lemma~\ref{lem:pistrong}, Corollary~\ref{cor:pigen} and the definition of $\gen{\cdot}'$.
		Analogously, from $\B_j \subseteq \overline{\B_j}$ we obtain $\overline{\L_j} \leq \L_j$ for each $2 \leq j \leq r$.
		Observe that $\overline{Y}$ cannot satisfy Condition~\ref{cond:easy} in Lemma~\ref{lem:repidef} as otherwise $\gen{\X}$ or $\gen{\X}'$ is contained in $\overline{Y}$ whereas no subset of $\gen{\X}$ or $\gen{\X}'$ is contained in $Y$ (by Condition~\ref{prop:p} and Corollary~\ref{cor:pigen}).
		Hence, also $\overline{Y}$ satisfies Condition~\ref{cond:hard}.
		From Condition~\ref{prop:p} it follows that both $\L_1$ and $\overline{\L_1}$ are contained in $\szs \setminus \{ \D, \M, \P, \X \}$, which, together with $\overline{\L_1} < \L_1$ and Lemma~\ref{lem:pistrong}, implies that $\overline{\L_1} = \ell(\overline{\ccc})$ for some $\emptyset \neq \overline{\ccc} \subseteq \ccs(z)$ while $\L_1 = \U$ or $\L_1 = \ell(\ccc)$ for some $\emptyset \neq \ccc \subsetneq \overline{\ccc}$.
		In particular, there is some color $i \in \overline{\ccc}$ such that there is no $\ccc' \subseteq \ccs(z)$ containing $i$ and satisfying $\L_1 = \ell(\ccc')$.
		Moreover, for each index $2 \leq j \leq r$ such that $\L_j = \ell(\ccc_j)$ for some color set $\ccc_j$ containing color $i$, we also have that $\overline{\L_j} = \ell(\overline{\ccc_j})$ for some color set $\overline{\ccc_j}$ containing color $i$, by $\overline{\L_j} \leq \L_j$, Lemma~\ref{lem:pistrong}, and the fact that $\L_j$ and $\overline{\L_j}$ are contained in $\szs \setminus \{ \D, \M, \P, \X \}$.
		It follows that the number of indices $1 \leq j \leq r$ such that $\overline{\L_j} = \ell(\overline{\ccc_j})$ for some color set $\overline{\ccc_j}$ containing color $i$ is strictly larger than the number of indices $1 \leq j' \leq r$ such that $\L_{j'} = \ell(\ccc_{j'})$ for some color set $\ccc_{j'}$ containing color $i$.
		This yields a contradiction to Condition~\ref{prop:i} and proves that (Z3) holds.
		It follows that $\ere = \edda$, i.e., $\ere$ is indeed as specified in the lemma.

		The set $\sre$ of output labels of $\repizs$ is precisely the set of all labels appearing in at least one configuration in $\ere$.
		From Condition~\ref{cond:easy} in Lemma~\ref{lem:repidef}, it follows that $\{ \gen{\L}, \gen{\L}' \mid \L \in \{ \P, \X \} \} \subseteq \sre$.
		Now, consider some (not necessarily nonempty) arbitrary color set $\ccc \subseteq \ccs(z)$ and, for each color $i \in \ccs(z)$, define $z'_i := z_i - 1$ if $i \in \ccc$, and $z'_i := z_i$ if $i \notin \ccc$.
		Consider the configuration $Y = \gen{\ell(\ccc_1)}' \s \gen{\ell(\ccc_2)} \s \gen{\ell(\ccc_3)} \dots \gen{\ell(\ccc_{r - 1})} \s \gen{\ell(\ccc)}$ defined by $\ccc_j := \{ i \in \ccs(z) \mid j \leq z'_i \}$ for each $1 \leq j \leq r - 1$, where we set $\ell(\emptyset) := \U$.
		By construction (and the fact that $z_i \leq r - 1$ for each $i \in \ccs(z)$), configuration $Y$ satisfies Condition~\ref{cond:hard}, which implies that $\{ \gen{\L}, \gen{\L}' \mid \L \in \szs \setminus \{ \D, \M, \P, \X \} \}$.
		It follows that $\sre$ is as specified in the lemma.
		The claimed property that $\gen{\L} \neq \gen{\L'}'$ for any $\L, \L' \in \szs \setminus \{ \D, \M \}$ follows from the fact that $\gen{\L'}'$ contains $\M$ for any $\L' \in \szs \setminus \{ \D, \M \}$ while $\gen{\L}$ does not contain $\M$ for any $\L \in \szs \setminus \{ \D, \M \}$, by Lemma~\ref{lem:pistrong}.

	The provided characterization of $\nre$ directly follows from the definition of $\re(\cdot)$.
	\end{proof}

	Set $\sre^0 := \{ \gen{\L} \mid \L \in \szs \setminus \{ \D, \M \} \}$ and $\sre^1 := \{ \gen{\L}' \mid \L \in \szs \setminus \{ \D, \M \} \}$.
	By Lemma~\ref{lem:repidef}, it follows that $\sre^0$ and $\sre^1$ are disjoint and $\sre^0 \cup \sre^1 = \sre$. 

	Similar to before, it will be useful to collect information about the strength of the labels in $\sre$ and compute the right-closed subsets generated by each label.
	We will do so in the following with Lemma~\ref{lem:restrong} and Corollary~\ref{cor:regen}, which are analogues of Lemma~\ref{lem:pistrong} and Corollary~\ref{cor:pigen} for $\repizs$ instead of $\Pi(z,s)$.
	
	\begin{figure}[h]
		\centering
		\includegraphics[width=0.9\textwidth]{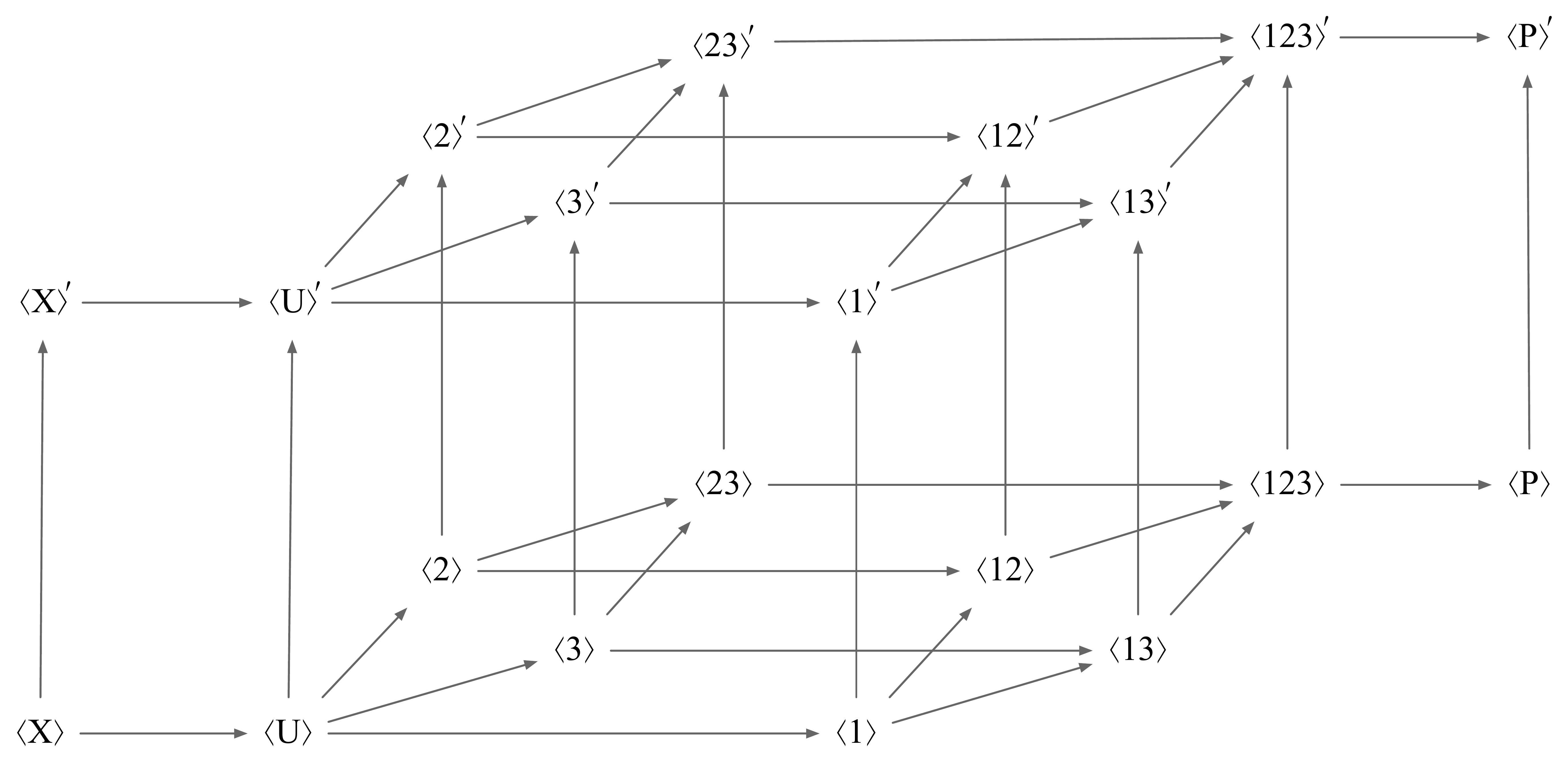}
		\caption{The node diagram of $\repizs$ for $\len(z) = 3$. Numbers correspond to color sets, and $\ell(\cdot)$ is omitted for clarity.}
		\label{fig:diag2}
	\end{figure} 
	
	\begin{lemma}\label{lem:restrong}
		For any $f(\cdot), \overline{f}(\cdot) \in \{ \gen{\cdot}, \gen{\cdot}' \}$ and any $\L, \overline{\L} \in \szs \setminus \{ \D, \M \}$, we have $f(\L) \leq \overline{f}(\overline{\L})$ if and only if the following two conditions are satisfied.
		\begin{enumerate}
			\item \label{cond:newone} $f(\cdot) = \gen{\cdot}$ or $\overline{f}(\cdot) = \gen{\cdot}'$ (or both).
			\item \label{cond:newtwo} At least one of the following conditions holds.
			\begin{enumerate}
				\item \label{prop:newa} $\L = \X$
				\item \label{prop:newb} $\overline{\L} = \P$
				\item \label{prop:newc} $\L = \U$ and $\overline{\L} \neq \X$
				\item \label{prop:newd} $\L = \ell(\ccc)$ and $\overline{\L} = \ell(\overline{\ccc})$ for some $\ccc, \overline{\ccc} \in 2^{\ccs(z)} \setminus \{ \emptyset \}$ satisfying $\ccc \subseteq \overline{\ccc}$
			\end{enumerate}
		\end{enumerate}
		Moreover, two labels from $\sre$ are equally strong if and only if they are identical.
	\end{lemma}
	\begin{proof}
		Observe that, by the definition of $\nre$, it holds for any $\B, \overline{\B} \in \sre$ with $\B \subseteq \overline{\B}$ that for any configuration $\B \s \B_2 \s \dots \s \B_{\Delta} \in \nre$, we also have $\overline{\B} \s \B_2 \s \dots \s \B_{\Delta} \in \nre$.
		Hence, for any $\B, \overline{\B} \in \sre$ with $\B \subseteq \overline{\B}$ we have $\B \leq \overline{\B}$, by the definition of strength.
		By Corollary~\ref{cor:pigen}, it follows that $\gen{\L} \leq \gen{\overline{\L}}$ for any $\L, \overline{\L}$ satisfying Condition~\ref{cond:newtwo} in the lemma.
		Moreover, since $\gen{\L}' \in \{ \gen{\L} \cup \{ \M \}, \gen{\L} \cup \{ \D, \M \} \}$ for any $\L \in \szs \setminus \{ \D, \M \}$, we also obtain that $\gen{\L} \leq \gen{\overline{\L}}'$ for any $\L, \overline{\L}$ satisfying Condition~\ref{cond:newtwo}.
		Furthermore, note that for any $\L, \overline{\L} \in \szs \setminus \{ \D, \M \}$, it holds that if $\gen{\L}' = \gen{\L} \cup \{ \D, \M \}$ and $\gen{\overline{\L}}' = \gen{\overline{\L}} \cup \{ \M \}$, then $\L \leq \ell(\{ s \})$ and $\overline{\L} \nleq \ell(\{ s \})$ (by the definition of $\gen{\cdot}'$), which implies $\gen{\L} \nsubseteq \gen{\overline{\L}}$ (by the definition of $\gen{\cdot}$).
		Since, by Corollary~\ref{cor:pigen}, any $\L, \overline{\L}$ satisfying Condition~\ref{cond:newtwo} satisfy $\gen{\L} \leq \gen{\overline{\L}}$, we also obtain that $\gen{\L}' \leq \gen{\overline{\L}}'$ for any $\L, \overline{\L}$ satisfying Condition~\ref{cond:newtwo}, analogously to above.
		Hence, we conclude that if Conditions~\ref{cond:newone} and~\ref{cond:newtwo} hold, then we have $f(\L) \leq \overline{f}(\overline{\L})$, as desired.
		What remains is to show that $f(\L) \leq \overline{f}(\overline{\L})$ does not hold if at least one of the two conditions is violated.

		Consider first the case that Condition~\ref{cond:newone} is violated, i.e., consider $\gen{\L}', \gen{\overline{\L}}$ for some arbitrary $\L, \overline{\L} \in \setminus \{ \D, \M \}$.
		We need to show that $\gen{\L}' \nleq \gen{\overline{\L}}$.
		Consider the configuration $\gen{\L}' \s \left(\gen{\X}'\right)^{\Delta - 1}$.
		Since $M^{\Delta} \in \nzs$ and, by the definition of $\gen{\cdot}'$, we have $\M \in \gen{\L}'$ and $\M \in \gen{\X}'$, it holds that $\gen{\L}' \s \left(\gen{\X}'\right)^{\Delta - 1} \in \nre$, by the definition of $\nre$.
		Now, consider the configuration $\gen{\overline{\L}} \s \left(\gen{\X}'\right)^{\Delta - 1}$ obtained from the above configuration by replacing $\gen{\L}'$ by $\gen{\overline{\L}}$.
		Observe that we have $\gen{\X} = \{ \M, \X \}$ and the only configurations in $\nzs$ containing $\M$ or $\X$ are $\M^{\Delta}$ and $\D^{\Delta - 1} \s \X$.
		Since, by Corollary~\ref{cor:pigen}, neither $\M$ nor $\D$ is contained in $\gen{\overline{\L}}$, it follows that $\gen{\overline{\L}} \s \left(\gen{\X}'\right)^{\Delta - 1} \notin \nre$, by the definition of $\nre$.
		By the definition of strength, it follows that $\gen{\L}' \nleq \gen{\overline{\L}}$, as desired.
		
		Now, consider the other case, i.e., consider some arbitrary $f(\cdot), \overline{f}(\cdot) \in \{ \gen{\cdot}, \gen{\cdot}' \}$ and $\L, \overline{\L} \in \szs \setminus \{ \D, \M \}$ violating Condition~\ref{cond:newtwo}.
		We need to show that $f(\L) \nleq \overline{f}(\overline{\L})$.
		Note that Condition~\ref{cond:newtwo} is violated if and only if all of Conditions~\ref{prop:newa}, \ref{prop:newb}, \ref{prop:newc}, and \ref{prop:newd} are violated.
		We now go through the different cases for $\L$.
		Note that $\L \neq \X$, due to the violation of Condition~\ref{prop:newa}.

		If $\L = \U$, consider the configuration $f(\L) \s \gen{\P} \s \gen{\U}^{\Delta - 2}$.
		Observe that, by Corollary~\ref{cor:pigen}, $\U$ is contained in $\gen{\U}$ and $f(\L)$, and $\P$ is contained in $\gen{\P}$.
		Since $\P \s \U^{\Delta - 1} \in \nzs$, it follows that $f(\L) \s \gen{\P} \s \gen{\U}^{\Delta - 2} \in \nre$.
		Further, observe that the violation of Condition~\ref{prop:newc} implies that $\overline{\L} = \X$, and consider the configuration $\overline{f}(\X) \s \gen{\P} \s \gen{\U}^{\Delta - 2}$ obtained from the above configuration by replacing $f(\L)$ by $\overline{f}(\overline{\L})$.
		We have $\overline{f}(\X) \in \{ \{ \X \}, \{ \M, \X \} \}$, the only configurations in $\nzs$ containing $\M$ or $\X$ are $\M^{\Delta}$ and $\D^{\Delta - 1} \s \X$, and both $\M$ and $\D$ are not contained in $\gen{\P}$, by Corollary~\ref{cor:pigen}.
		It follows that $\overline{f}(\X) \s \gen{\P} \s \gen{\U}^{\Delta - 2} \notin \nre$.
		Hence, $f(\L) \nleq \overline{f}(\overline{\L})$, as desired.

		If $\L = \ell(\ccc)$ for some $\emptyset \neq \ccc \subseteq \ccs(z)$, consider the configuration $f(\L) \s \gen{\L}^{\Delta - |\ccc|} \s \gen{\U}^{|\ccc| - 1}$.
		Observe that, by Corollary~\ref{cor:pigen}, $\ell(\ccc)$ is contained in $\gen{\L}$ and $f(\L)$, and $\U$ is contained in $\gen{\U}$.
		Since $\ell(\ccc)^{\Delta - |\ccc| + 1} \s \U^{\ccc - 1} \in \nzs$, it follows that $f(\L) \s \gen{\L}^{\Delta - |\ccc|} \s \gen{\U}^{|\ccc| - 1} \in \nre$.
		Further, observe that the violation of Conditions~\ref{prop:newb} and~\ref{prop:newd} implies that $\overline{\L} \in \{ \U, \X \}$ or $\overline{\L} = \ell(\overline{\ccc})$ for some $\emptyset \neq \overline{\ccc} \subseteq \ccs(z)$ satisfying $\ccc \nsubseteq \overline{\ccc}$.
		Consider the configuration $Y := \overline{f}(\overline{\L}) \s \gen{\ell(\ccc)}^{\Delta - |\ccc|} \s \gen{\U}^{|\ccc| - 1}$ obtained from the above configuration by replacing $f(\L)$ by $\overline{f}(\overline{\L})$, and assume for a contradiction that $Y \in \nre$.
		Since $\len(z) \leq \Delta - 1$, we have $|\ccc| \leq |\ccs(z)| \leq \Delta - 1$, which implies that $\Delta - |\ccc| \geq 1$, i.e., $Y$ contains $\gen{\ell(\ccc)}$.
		Observe also that our assumption for this section, $\Delta \geq 3$, implies that there are at least two positions in $Y$ with a label from $\{ \gen{\ell(\ccc)}, \gen{\U} \}$.
		Since, by Corollary~\ref{cor:pigen} and our insights about $\L$, we have $\M \notin \gen{\ell(\ccc)}$, $\P \notin \overline{f}(\overline{\L}) \cup \gen{\ell(\ccc)} \cup \gen{\U}$, and $\D \notin \gen{\ell(\ccc)} \cup \gen{\U}$, it follows that there is no choice $(\A_1, \dots, \A_{\Delta}) \in \overline{f}(\overline{\L}) \times \gen{\ell(\ccc)}^{\Delta - |\ccc|} \times \gen{\U}^{|\ccc| - 1}$ such that $\A_1 \s \dots \s \A_{\Delta} \in \{ \M^{\Delta}, \P \s \U^{\Delta - 1}, \D^{\Delta - 1} \s \X \}$.
		By the definition of $\nzs$ and the fact that $Y \in \nre$, this implies that there is some $\emptyset \neq \ccc^* \subseteq \ccs(z)$ such that $\ell(\ccc^*)$ is contained in at least $\Delta - |\ccc^*| + 1$ of the $\Delta$ sets in $Y$.
		Since $\ell(\ccc^*) \notin \gen{\U}$ (by Corollary~\ref{cor:pigen}), we have $\Delta - |\ccc^*| + 1 \leq \Delta - |\ccc| + 1$, which implies $|\ccc^*| \geq |\ccc|$.
		If $\ccc^* \neq \ccc$, then it follows that $\ccc^* \nsubseteq \ccc$, which implies $\ell(\ccc^*) \notin \gen{\ell(\ccc)}$, by Corollary~\ref{cor:pigen}; it follows that $\ell(\ccc^*)$ is contained in at most one of the $\Delta$ sets in $Y$, which implies $|\ccc^*| \geq \Delta$, yielding a contradiction to the fact that $|\ccc^*| \leq |\ccs(z)| \leq \Delta - 1$.
		Hence, we know that $\ccc^* = \ccc$.
		Recall that we have $\overline{\L} \in \{ \U, \X \}$ or $\overline{\L} = \ell(\overline{\ccc})$ for some $\emptyset \neq \overline{\ccc} \subseteq \ccs(z)$ satisfying $\ccc \nsubseteq \overline{\ccc}$, and observe that in the latter case we have $\ell(\ccc) \notin \gen{\ell(\ccc^*)}$ and $\ell(\ccc) \notin \gen{\ell(\ccc^*)}'$, by Corollary~\ref{cor:pigen} and the definition of $\gen{\cdot}'$.
		Since $\ell(\ccc)$ is also not contained in any of $\gen{\U}$, $\gen{\U}'$, $\gen{\X}$, and $\gen{\X}'$, it follows that $\ell(\ccc) \notin \overline{f}(\overline{\L})$.
		Moreover, since $\ccc^* = \ccc$, we have $\ell(\ccc^*) \notin \overline{f}(\overline{\L})$.
		By the definition of $Y$, it follows that $\ell(\ccc^*)$ is contained in at most $\Delta - |\ccc| = \Delta - |\ccc^*|$ of the $\Delta$ sets in $Y$.
		This contradicts the fact established above that $\ell(\ccc^*)$ is contained in at least $\Delta - |\ccc^*| + 1$ of the $\Delta$ sets in $Y$.
		Hence, we have $Y \notin \nre$, which implies that $f(\L) \nleq \overline{f}(\overline{\L})$, as desired.
		
		If $\L = \P$, consider the configuration $f(\L) \s \gen{\U}^{\Delta - 1}$.
		Observe that, by Corollary~\ref{cor:pigen}, $\P$ is contained in $f(\L)$ and $\U$ is contained in $\gen{\U}$.
		 Since $\P \s \U^{\Delta - 1} \in \nzs$, it follows that $f(\L) \s \gen{\U}^{\Delta - 1} \in \nre$.
		Further, observe that the violation of Condition~\ref{prop:newb} implies that $\overline{\L} \neq \P$, and consider the configuration $\overline{f}(\overline{\L}) \s \gen{\U}^{\Delta - 1}$ obtained from the above configuration by replacing $f(\L)$ by $\overline{f}(\overline{\L})$.
		We have $\gen{\U} = \{ \U, \X \}$, the only configurations in $\nzs$ containing some element from $\{ \U, \X \}$ in at least $\Delta - 1$ positions is $\P \s \U^{\Delta - 1}$ (due to $\Delta \geq 3$ and $|\ccs(z)| = \len(z) \leq \Delta - 1$), and $\P$ is contained neither in $\overline{f}(\overline{\L})$ nor in $\gen{\U}$, by Corollary~\ref{cor:pigen}.
		It follows that $\overline{f}(\overline{\L}) \s \gen{\U}^{\Delta - 1} \notin \nre$.
		Hence, $f(\L) \nleq \overline{f}(\overline{\L})$, as desired.

		This covers all possible cases for $\L \in \szs \setminus \{ \D, \M \}$ and concludes the second case.
		Hence, the lemma statement follows.
		Note that the fact that two labels from $\sre$ are equally strong if and only if they are identical follows from the fact that for any two distinct labels $\B, \B' \in \sre$ we have $\B \nleq \B'$ or $\B' \nleq \B$.

	\end{proof}

	From Lemma~\ref{lem:restrong} we obtain Corollary~\ref{cor:regen} by simply collecting for any label $\L \in \sre$ the set of all labels $\L'$ satisfying $\L \leq \L'$.

	\begin{corollary}\label{cor:regen}
		We have
		\begin{itemize}
			\item $\gen{\gen{\P}} = \{ \gen{\P}, \gen{\P}' \}$,
			\item $\gen{\gen{\U}} = \{ \gen{\L}, \gen{\L}' \mid \L \in \{ \P, \U \} \cup \{ \ell(\ccc) \mid \emptyset \neq \ccc \subseteq \ccs(z) \} \}$,
			\item $\gen{\gen{\X}} = \sre$,
			\item $\gen{\gen{\ell(\ccc)}} = \{ \gen{\L}, \gen{\L}' \mid \L \in \{ \P \} \cup \{ \ell(\ccc') \mid \ccc \subseteq \ccc' \subseteq \ccs(z) \} \}$ for each $\emptyset \neq \ccc \subseteq \ccs(z)$, and
			\item $\gen{\gen{\L}'} = \gen{\gen{\L}} \cap \sre^1$, for each $\L \in \szs \setminus \{ \D, \M \}$.
		\end{itemize}
	\end{corollary}

\subsection{Computing $\rere(\re(\Pi(z,s)))$}
	After computing and analyzing $\repizs$, we now turn our attention to the next step, computing $\rerepizs$.
	However, instead of computing $\rerepizs$ exactly, we instead show that $\rerepizs$ can be relaxed to some problem $\pistar$ (without computing $\rerepizs$ explicitly).
	We will define $\pistar$ so that by just renaming the labels in the set $\sstar := \Sigma_{\pistar}$ suitably we obtain some problem of the form $\Pi(z', s')$ with new parameters $z', s'$ (compared to $\Pi(z, s)$).
	As relaxing a problem, by definition, does not increase the complexity of the problem, this approach will yield a proof that $\Pi(z',s')$ can be solved at least one round faster than $\Pi(z,s)$; analyzing how the parameters change from $\Pi(z,s)$ to $\Pi(z',s')$ will then enable us to prove the desired lower bound.
	We start by defining $\pistar$, which depends on some parameter $q$ that can be chosen arbitrarily from $\ccs(z) \setminus \{ s \}$.
	Note that such a parameter $q$ exists since, by definition, $\len(z) \geq 2$.

	\begin{definition}\label{def:pistar}
		Fix some parameter $q \in \ccs(z) \setminus \{ s \}$ and define $\ell(\emptyset) := \U$.
		The label set $\sstar$ of $\pistar$ is defined by
		\[
			\sstar := \{ \gen{\gen{\P}}, \gen{\gen{\U}}, \gen{\gen{\U}'}, \gen{\gen{\X}}, \gen{\gen{\X}'} \} \cup \fminus \cup \fplus
		\]
		where
		\begin{align*}
			\fminus &:= \{ \gen{\gen{\ell(\ccc)}} \mid \emptyset \neq \ccc \subseteq \ccs(z), q \notin \ccc  \} \qquad \text{ and }\\
			\fplus &:= \{ \gen{\gen{\ell(\ccc)}, \gen{\ell(\ccc \setminus \{ q \})}'} \mid \emptyset \neq \ccc \subseteq \ccs(z), q \in \ccc \}.
		\end{align*}
		The node constraint $\nstar$ of $\pistar$ consists of the following configurations.
		\begin{itemize}
			\item $\gen{\gen{\X}'}^{\Delta}$
			\item $\gen{\gen{\P}} \s \gen{\gen{\U}}^{\Delta - 1}$
			\item $\gen{\gen{\U}'}^{\Delta - 1} \s \gen{\gen{\X}}$
			\item $\gen{\gen{\ell(\ccc)}}^{\Delta - |\ccc| + 1} \s \gen{\gen{\U}}^{|\ccc| - 1}$ for each $\emptyset \neq \ccc \subseteq \ccs(z)$ satisfying $q \notin \ccc$
			\item $\gen{\gen{\ell(\ccc)}, \gen{\ell(\ccc \setminus \{ q \})}'}^{\Delta - |\ccc| + 1} \s \gen{\gen{\U}}^{|\ccc| - 1}$ for each $\emptyset \neq \ccc \subseteq \ccs(z)$ satisfying $q \in \ccc$
		\end{itemize}
		The hyperedge constraint $\estar$ of $\pistar$ consists of all configurations $\Q_1 \s \dots \s \Q_r$ with $\Q_j \in \sstar$ for all $1 \leq j \leq r$ such that there exists a choice $(\B_1, \dots, \B_r) \in \Q_1 \times \dots \times \Q_r$ satisfying $\B_1 \s \dots \s \B_r \in \ere$.
	\end{definition}

	Now we are ready to prove the desired relation between $\rerepizs$ and $\pistar$.
	We start by relating the node configurations of the two problems.

	\begin{lemma}
		Each node configuration of $\rerepizs$ can be relaxed to some configuration from $\nstar$.
	\end{lemma}
	\begin{proof}
		For a contradiction, suppose that the lemma does not hold, i.e., there exists some configuration $\fQ = \Q_1 \s \dots \s \Q_{\Delta}$ in the node constraint of $\rerepizs$ that cannot be relaxed to any configuration from $\nstar$.
		By the definition of $\rerepizs$, this implies that for any choice $(\B_1, \dots, \B_{\Delta}) \in \Q_1 \times \dots \times \Q_{\Delta}$ there exists some choice $(\A_1, \dots, \A_{\Delta}) \in \B_1 \times \dots \times \B_{\Delta}$ such that $\A_1 \s \dots \s \A_{\Delta} \in \nzs$.
		Moreover, as $\fQ$ is contained in the node constraint of $\rerepizs$, we know that $\Q_k$ is right-closed for each $1 \leq k \leq \Delta$, by Observation~\ref{obs:rcs}.
		
		We distinguish two cases based on whether there exists some $\Q_k$ containing $\gen{\X'}$.
		Consider first the case that $\gen{\X}' \in \Q_k$ for some $1 \leq k \leq \Delta$.
		W.l.o.g., assume that $k = 1$, i.e., we have $\gen{\X}' \in \Q_1$.
		We claim that this implies $\gen{\P} \notin \Q_k$ for each $2 \leq k \leq \Delta$:
		If there exists some $2 \leq k \leq \Delta$ satisfying $\gen{\P} \in \Q_k$, then there exists some choice $(\B_1, \dots, \B_{\Delta}) \in \Q_1 \times \dots \times \Q_{\Delta}$ with $\B_1 = \gen{\X}'$ and $\B_k = \gen{\P}$.
		It follows by Lemma~\ref{lem:pistrong} and Corollary~\ref{cor:pigen} that for every choice $(\A_1, \dots, \A_{\Delta}) \in \B_1 \times \dots \times \B_{\Delta}$ we have $\A_1 \in \{ \M, \X \}$ and $\A_k \notin \{ \D, \M \}$; now the claim follows by observing that every configuration in $\nzs$ containing $\M$ requires all other labels in the configuration to be $\M$ as well and every configuration in $\nzs$ containing $\X$ requires all other labels in the configuration to be $\D$.

		Combining the claim with the right-closedness of the $\Q_k$, we obtain that $\Q_k \subseteq \gen{\gen{\X}'}$ for each $2 \leq k \leq \Delta$, by Lemma~\ref{lem:restrong} and Corollary~\ref{cor:regen}.
		As, by definition, $\fQ$ cannot be relaxed to $\gen{\gen{\X}'}^{\Delta}$, it follows that $\Q_1 \nsubseteq \gen{\gen{\X}'}$.
		Since, by Lemma~\ref{lem:restrong}, we have $\L \leq \gen{\P}$ for each $\L \in \sre \setminus \gen{\gen{\X}'}$, we conclude that $\gen{\P} \in \Q_1$, by the right-closedness of $\Q_1$.
		This implies that $\gen{\X}' \notin \Q_k$ for each $2 \leq k \leq \Delta$, as otherwise there would (again) exist some choice $(\B_1, \dots, \B_{\Delta}) \in \Q_1 \times \dots \times \Q_{\Delta}$ containing $\gen{\P}$ and $\gen{\X}'$, which leads to a contradiction as already seen above.
		From $\Q_k \subseteq \gen{\gen{\X}'}$ and $\gen{\X}' \notin \Q_k$ for each $2 \leq k \leq \Delta$, we infer that $\Q_k \subseteq \gen{\gen{\U}'}$ for each $2 \leq k \leq \Delta$, by Corollary~\ref{cor:regen}.
		Since, by Corollary~\ref{cor:regen}, $\Q_1 \subseteq \sre = \gen{\gen{\X}}$, it follows that $\fQ$ can be relaxed to the configuration $\gen{\gen{\U}'}^{\Delta - 1} \s \gen{\gen{\X}} \in \nstar$, yielding the desired contradiction.

		Now consider the other case, namely that $\gen{\X}' \notin \Q_k$ for each $1 \leq k \leq \Delta$.
		Observe that this implies
		\begin{equation}\label{eq:ggu}
			\Q_k \subseteq \gen{\gen{\U}} \text{ for each } 1 \leq k \leq \Delta,
		\end{equation}
		by Corollary~\ref{cor:regen} and the fact that $\gen{\gen{\X}} < \gen{\gen{\X}'}$ (which follows from Lemma~\ref{lem:restrong}).
		It follows that
		\begin{equation}\label{eq:allc}
			\gen{\ell(\ccs(z))}' \in \Q_k \text{ for each } 1 \leq k \leq \Delta,
		\end{equation}
		as otherwise there would be some $1 \leq k \leq \Delta$ satisfying $\Q_k \subseteq \gen{\P}$ (by Lemma~\ref{lem:restrong}, Corollary~\ref{cor:regen} and the right-closedness of the $\Q_k$), which in turn would imply that $\fQ$ can be relaxed to $\gen{\gen{\P}} \s \gen{\gen{\U}}^{\Delta - 1} \in \nstar$, yielding a contradiction.
		
		While, so far, we used (only) that $\fQ$ cannot be relaxed to any of the configurations $\gen{\gen{\X}'}^{\Delta}$, $\gen{\gen{\P}} \s \gen{\gen{\U}}^{\Delta - 1}$, and $\gen{\gen{\U}'}^{\Delta - 1} \s \gen{\gen{\X}}$, we will now also take advantage of the fact that $\fQ$ cannot be relaxed to the remaining (conceptually more complex) configurations from $\nstar$.
		For each $i \in \ccs(z)$, define $R_i := \ccs(z) \setminus \{ i \}$.
		Moreover, we will make use of the bipartite graph $\overline{G} = (\overline{V} \cup \overline{W}, \overline{E})$ obtained by defining $\overline{V} := \{ \Q_1, \dots, \Q_\Delta \}$ and $\overline{W} := \{ R_i \mid i \in \ccs(z) \}$, and setting $\overline{E}$ to be the set of all edges $\{ \Q_k, R_i \}$ satisfying
		\begin{itemize}
			\item $\gen{\ell(R_i)}' \in \Q_k$ if $i \neq q$, and
			\item $\gen{\ell(R_i)} \in \Q_k$ if $i = q$.
		\end{itemize}
		Note that, for any distinct $k, k'$, we consider $\Q_k$ and $\Q_{k'}$ to be different vertices in $\overline{V}$ even if $\Q_k = \Q_{k'}$.
		
		Consider any arbitrary subset $\emptyset \neq \ccc \subseteq \ccs(z)$ and recall that $\ell(\emptyset) := \U$.
		We distinguish two cases, based on whether $\ccc$ contains $q$.

		If $q \notin \ccc$, then, by the definition of $\nstar$, configuration $\fQ$ cannot be relaxed to the configuration $\gen{\gen{\ell(\ccc)}}^{\Delta - |\ccc| + 1} \s \gen{\gen{\U}}^{|\ccc| - 1}$, which, by (\ref{eq:ggu}), implies that there are at least $|\ccc|$ indices $k \in \{ 1, \dots, \Delta \}$ satisfying $\Q_k \nsubseteq \gen{\gen{\ell(\ccc)}}$.
		Observe that for each label $\L \in \gen{\gen{\U}} \setminus \gen{\gen{\ell(\ccc)}}$, we have $\L \in \{ \gen{\U}, \gen{\U}' \}$ or $\L \in \{ \gen{\ell(\overline{\ccc})}, \gen{\ell(\overline{\ccc})}' \}$ for some $\overline{\ccc} \nsupseteq \ccc$, by Corollary~\ref{cor:regen}.
		By Lemma~\ref{lem:restrong}, this implies that for each label $\L \in \gen{\gen{\U}} \setminus \gen{\gen{\ell(\ccc)}}$ there exists some $i \in \ccc$ satisfying $\L \leq \gen{\ell(R_i)}'$.
		By (\ref{eq:ggu}) and the right-closedness of the $\Q_k$, it follows that there are at least $|\ccc|$ indices $k \in \{ 1, \dots, \Delta \}$ such that there exists some $i \in \ccc$ satisfying $\gen{\ell(R_i)}' \in \Q_k$.
		
		Similarly, if $q \in \ccc$, then $\fQ$ cannot be relaxed to the configuration $\gen{\gen{\ell(\ccc)}, \gen{\ell(\ccc \setminus \{ q \})}'}^{\Delta - |\ccc| + 1} \s \gen{\gen{\U}}^{|\ccc| - 1}$, which, by (\ref{eq:ggu}), implies that there are at least $|\ccc|$ indices $k \in \{ 1, \dots, \Delta \}$ satisfying $\Q_k \nsubseteq \gen{\gen{\ell(\ccc)}, \gen{\ell(\ccc \setminus \{ q \})}'}$.
		Observe that for each label $\L \in \gen{\gen{\U}} \setminus \gen{\gen{\ell(\ccc)}, \gen{\ell(\ccc \setminus \{ q \})}'}$, we have $\L \in \{ \gen{\U}, \gen{\U}' \}$, $\L = \gen{\ell(\overline{\ccc})}$ for some $\overline{\ccc} \nsupseteq \ccc$, or $\L = \gen{\ell(\overline{\ccc})}'$ for some $\overline{\ccc} \nsupseteq \ccc \setminus \{ q \}$, by Lemma~\ref{lem:restrong} and Corollary~\ref{cor:regen}.
		By Lemma~\ref{lem:restrong}, this implies\footnote{Note that the implication is also correct in the special case $\ccc = \{ q \}$.} that for each label $\L \in \gen{\gen{\U}} \setminus \gen{\gen{\ell(\ccc)}, \gen{\ell(\ccc \setminus \{ q \})}'}$ we have $\L \leq \gen{\ell(R_q)}$ or there exists some $i \in \ccc \setminus \{ q \}$ satisfying $\L \leq \gen{\ell(R_i)}'$.
		By (\ref{eq:ggu}) and the right-closedness of the $\Q_k$, it follows that there are at least $|\ccc|$ indices $k \in \{ 1, \dots, \Delta \}$ such that $\gen{\ell(R_q)} \in \Q_k$ or there exists some $i \in \ccc$ satisfying $\gen{\ell(R_i)}' \in \Q_k$.
		
		Recall the definition of $\overline{G}$.
		By the above discussion, we conclude that for any arbitrary subset $\emptyset \neq \ccc \subseteq \ccs(z)$, the vertex set $\{ R_i \mid i \in\ccc \} \subseteq \overline{W}$ has at least $|\ccc|$ neighbors in $\overline{V}$.
		Hence, we can apply Theorem~\ref{thm:hall} (i.e., Hall's marriage theorem) to $\overline{G}$ and  obtain a function $f \colon \overline{W} \rightarrow \overline{V}$ such that $f(R_i) \neq f(R_{i'})$ for any $i \neq i'$, $\gen{\ell(R_q)} \in f(R_q)$, and $\gen{\ell(R_i)}' \in f(R_i)$ for any $i \neq q$.
		This implies that there is some choice $(\B_1, \dots, \B_{\Delta}) \in \Q_1 \times \dots \times \Q_{\Delta}$ that contains $\gen{\ell(R_q)}$ as well as $\gen{\ell(R_i)}'$ for all $i \in \ccs(z) \setminus \{ q \}$.
		By (\ref{eq:allc}), it follows that there is some choice $(\B_1, \dots, \B_{\Delta}) \in \Q_1 \times \dots \times \Q_{\Delta}$ such that $Y = \B_1 \s \dots \s \B_{\Delta}$ is a permutation of $\gen{\ell(R_1)}' \s \dots \s \gen{\ell(R_{q - 1})}' \s \gen{\ell(R_q)} \s \gen{\ell(R_{q + 1})}' \s \dots \s \gen{\ell(R_{\len(z)})}' \s (\gen{\ell(\ccs(z))}')^{\Delta - \len(z)}$.
		
		Consider an arbitrary choice $(\A_1, \dots, \A_{\Delta}) \in \B_1 \times \dots \times \B_{\Delta}$, and set $\fA := \A_1 \s \dots \s \A_{\Delta}$.
		Using the aforementioned characterization of $Y$ together with Lemma~\ref{lem:pistrong} and Corollary~\ref{cor:pigen}, we collect some properties of $\fA$ in the following.
		Since $\P \notin \B_k$ for each $1 \leq k \leq \Delta$, we have $\fA \neq \P \s \U^{\Delta - 1}$.
		Since $\M \notin \gen{\ell(R_q)}$, we have $\fA \neq \M^{\Delta}$.
		Since $\D \notin \gen{\ell(R_q)}$, $\D \notin \gen{\ell(R_s)}'$, and $s \neq q$ (by the definition of $\Pi^*_q$), we have $\fA \neq \D^{\Delta - 1} \s \X$.
		Finally, for each $\emptyset \neq \ccc \subseteq \ccs(z)$, we have $\ell(\ccc) \notin \gen{\ell(R_i)}'$ for each $i \in \ccc \setminus \{ q \}$ and, if $q \in \ccc$, additionally $\ell(\ccc) \notin \gen{\ell(R_q)}$, which implies that $\fA \neq \ell(\ccc)^{\Delta - |\ccc| + 1} \s \U^{|\ccc| - 1}$.

		By the definition of $\nzs$, it follows that for any choice $(\A_1, \dots, \A_{\Delta}) \in \B_1 \times \dots \times \B_{\Delta}$, we have $\A_1 \s \dots \s \A_{\Delta} \notin \nzs$, yielding a contradiction to the fact that $\fQ$ is contained in the node constraint of $\rerepizs$ and concluding the proof.
	\end{proof}

	Next, we provide a useful characterization of $\estar$.

	\begin{lemma}\label{lem:estarcharacterization}
		The hyperedge constraint $\estar$ of $\pistar$ consists of all configurations $\L_1 \s \dots \s \L_r$ (with labels from $\sstar$) satisfying at least one of the following three conditions.
	\begin{enumerate}
		\item\label{cond:1++} There is some index $1 \leq j \leq r$ such that $\L_j = \gen{\gen{\X}'}$ and $\L_{j'} \notin \{ \gen{\gen{\U}'}, \gen{\gen{\X}'} \}$ for each $j' \neq j$.
		\item\label{cond:2++} There are two distinct indices $j, j' \in \{ 1, \dots, r \}$ such that $\L_j = \gen{\gen{\X}}$, $\L_{j'}$ is arbitrary, and $\L_{j''} \notin \{ \gen{\gen{\U}'}, \gen{\gen{\X}'} \}$ for each $j'' \in \{ 1, \dots, r \} \setminus \{ j, j' \}$.
		\item\label{cond:3++} All of the following properties hold.
		\begin{enumerate}
			\item\label{prop:a++} $\L_j \neq \gen{\gen{\P}}$ for each $1 \leq j \leq r$.
			\item\label{prop:b++} There is at most one index $j \in \{1, \dots, r\}$ such that $\L_j \in \{ \gen{\gen{\U}'}, \gen{\gen{\X}'} \}$.
			\item\label{prop:c++} There are at most $z_q + 1$ indices $j \in \{1, \dots, r\}$ such that $\L_j = \gen{\gen{\U}'}$ or $\L_j = \gen{\gen{\ell(\ccc)}, \gen{\ell(\ccc \setminus \{ q \} )}'}$ for some color set $\ccc$ containing color $q$.
			\item\label{prop:d++} For each $i \in \ccs(z)$ satisfying $i \neq q$, there are at most $z_i$ indices $j \in \{1, \dots, r\}$ such that $\L_j = \gen{\gen{\ell(\ccc)}}$ for some color set $\ccc$ containing color $i$.
		\end{enumerate}
	\end{enumerate}
	\end{lemma}

	\begin{proof}
		We start by showing that any configuration (with labels from $\sstar$) satisfying at least one of the three conditions stated in the lemma is contained in $\estar$.
		Let $\fQ = \Q_1 \s \dots \s \Q_r$ be an arbitrary configuration satisfying at least one of the conditions.
		If $\fQ$ satisfies Condition~\ref{cond:1++}, then, by Corollary~\ref{cor:regen}, we know that some permutation of $(\gen{\X}', \gen{\P}, \dots, \gen{\P})$ is contained in $\Q_1 \times \dots \times \Q_r$.
		Since Condition~\ref{cond:easy} in Lemma~\ref{lem:repidef} implies that $\gen{\X}' \s \gen{\P}^{r - 1} \in \ere$, we obtain $\fQ \in \estar$, by the definition of $\estar$.
		If $\fQ$ satisfies Condition~\ref{cond:2++}, then, by Corollary~\ref{cor:regen}, we know that some permutation of $(\gen{\X}, \gen{\P}', \dots, \gen{\P})$ is contained in $\Q_1 \times \dots \times \Q_r$.
		Again, Condition~\ref{cond:easy} in Lemma~\ref{lem:repidef} implies that $\gen{\X} \s \gen{\P}' \s \gen{\P}^{r - 2} \in \ere$, and we obtain $\fQ \in \estar$.

		Hence, consider the case that $\fQ$ satisfies Condition~\ref{cond:3++}.
		If there is some $1 \leq j \leq r$ satisfying $\Q_j = \gen{\gen{\X}'}$, then, by Condition~\ref{prop:b++}, $\fQ$ satisfies also Condition~\ref{cond:1++} and we are done.
		Thus, assume in the following that $\Q_j \neq \gen{\gen{\X}'}$ for all $1 \leq j \leq r$.
		Choose some $(\B_1, \dots, \B_r) \in \Q_1 \times \dots \times \Q_r$ as follows.
		If there is some $1 \leq j \leq r$ satisfying $\Q_j = \gen{\gen{\U}'}$, then, for any $1 \leq j' \leq r$,
		\begin{itemize}
			\item if $\Q_{j'} = \gen{\gen{\X}}$ choose $\B_{j'} = \gen{\U}$,
			\item if $\Q_{j'} = \gen{\gen{\L}}$ for some $\L \in \{ \U \} \cup \{ \ell(\ccc) \mid q \notin \ccc \subseteq \ccs(z) \}$, choose $\B_{j'} = \gen{\L}$,
			\item if $\Q_{j'} = \gen{\gen{\U}'}$, choose $\B_{j'} = \gen{\U}'$, and
			\item if $\Q_{j'} = \gen{\gen{\ell(\ccc)}, \gen{\ell(\ccc \setminus \{ q \} )}'}$ for some $\ccc \subseteq \ccs(z)$ containing $q$, choose $\B_{j'} = \gen{\ell(\ccc)}$.
		\end{itemize}
		If there is no $1 \leq j \leq r$ satisfying $\Q_j = \gen{\gen{\U}'}$, then choose the $\B_{j'}$ as above, with the only difference that if there is at least one $j''$ with $\Q_{j''} = \gen{\gen{\ell(\ccc)}, \gen{\ell(\ccc \setminus \{ q \} )}'}$ for some $\ccc \subseteq \ccs(z)$ containing $q$, then choose one such $j''$ arbitrarily, and set $\B_{j''} = \gen{\ell(\ccc \setminus \{ q \} )}'$ for this $j''$.
		By Corollary~\ref{cor:regen} and the definition of $\gen{\cdot}$, each chosen $\B_{j'}$ is indeed contained in $\Q_{j'}$; by the fact that $\Q_{j'} \notin \{ \gen{\gen{\P}}, \gen{\gen{\X}'} \}$ for all $1 \leq j' \leq r$ (due to Condition~\ref{prop:a++} and the above discussion) and the definition of $\sstar$, it is also guaranteed that we chose some $\B_{j'}$ for each $1 \leq j' \leq r$.

		The configuration $\fB := \B_1 \s \dots \s \B_r$ is not necessarily contained in $\estar$; however, we will show that with some additional changes (that preserve containment of the $\B_k$ in the respective $\Q_k$), we can transform $\fB$ into a configuration that is contained in $\estar$.
		Before explaining this transformation, we collect some properties of $\fB$ in the following.
		From the definition of $\fB$, it follows that
		\begin{equation}\label{eq:pp}
			\B_j \notin \{ \gen{\P}, \gen{\P}', \gen{\X}, \gen{\X}' \} \text{ for each } 1 \leq j \leq r.
		\end{equation}
		From Condition~\ref{prop:b++} and the definition of $\sre^1$, it follows that
		\begin{equation}\label{eq:one}
			\text{there is at most one index $1 \leq j \leq r$ such that $\B_j \in \sre^1$.}
		\end{equation}
		From Condition~\ref{prop:c++}, it follows that
		\begin{equation}\label{eq:zq}
			\text{there are at most $z_q$ indices $j \in \{ 1, \dots, r \}$ such that $\B_j \in \{ \gen{\ell(\ccc)}, \gen{\ell(\ccc)}' \mid q \in \ccc \subseteq \ccs(z) \}$.}
		\end{equation}
		From Condition~\ref{prop:d++}, it follows that
		\begin{equation}\label{eq:zi}
			\begin{split}
				\text{for each $1 \leq i \leq z$ satisfying $i \neq q$, there are at most $z_i$ indices}\\ \text{$j \in \{ 1, \dots, r \}$ such that $\B_j \in \{ \gen{\ell(\ccc)}, \gen{\ell(\ccc)}' \mid i \in \ccc \subseteq \ccs(z) \}$.}
			\end{split}
		\end{equation}

		Now, we will perform the aforementioned transformation on $\fB$.
		Observe that, by interpreting $\U$ as $\ell(\emptyset)$, we have $\B_j \in \{ \gen{\ell(\ccc)}, \gen{\ell(\ccc)}' \mid \ccc \subseteq \ccs(z) \}$ for each $1 \leq j \leq r$, by (\ref{eq:pp}) and the definition of $\sstar$.
		Now, iterate through the colors in $\ccs(z)$.
		When processing color $i \in \ccs(z)$, repeat the following until there are exactly $z_i$ indices $j \in \{ 1, \dots, r \}$ such that $\B_j \in \{ \gen{\ell(\ccc)}, \gen{\ell(\ccc)}' \mid i \in \ccc \subseteq \ccs(z) \}$: choose some $\B_j$ satisfying $\B_j = f(\ell(\ccc))$ for some $f \in \{ \gen{\cdot}, \gen{\cdot}' \}$ and some $\ccc$ that does not contain $i$, and replace it with $f(\ell(\ccc \cup \{ i \}))$.
		This is possible due to (\ref{eq:zq}), (\ref{eq:zi}), and the fact that $z_i \leq r - 1$ for each $i \in \ccs(z)$.
		Note that each step in the transformation preserves the correctness of (\ref{eq:pp}), (\ref{eq:one}), (\ref{eq:zq}), and (\ref{eq:zi}).
		Observe also that in each replacement performed during the transformation we replace some label with a stronger label, by Lemma~\ref{lem:restrong}.
		By the right-closedness of all labels in $\sstar$ (which implies that all $\Q_j$ are right-closed), it follows that the transformation preserves that $\B_j \in \Q_j$ for each $1 \leq j \leq r$.
		
		From the above construction and discussion, we conclude that there is some $(\B_1, \dots, \B_r) \in \Q_1 \times \dots \times \Q_r$ such that $\B_1 \s \dots \s \B_r$ is a permutation of some configuration $\gen{\L_1}' \s \gen{\L_2} \s \gen{\L_3} \dots \gen{\L_r}$ satisfying $\L_j \in \szs \setminus \{ \D, \M \}$ for all $1 \leq j \leq r$ and Condition~\ref{cond:hard} in Lemma~\ref{lem:repidef}, i.e., such that $\B_1 \s \dots \s \B_r \in \estar$.
		This concludes the proof that any configuration satisfying at least one of the three conditions stated in the lemma is contained in $\estar$.

		Now we will show the other direction, i.e., that any configuration from $\estar$ satisfies at least one of the three conditions stated in the lemma.
		We will do so by showing the contrapositive, i.e., that any configuration (with labels from $\sstar$) that violates all three conditions is not contained in $\estar$.
		Let $\fQ = \Q_1 \s \dots \s \Q_r$ be an arbitrary configuration violating all three conditions, i.e., $\fQ$ violates Conditions~\ref{cond:1++} and \ref{cond:2++}, and at least one of Conditions~\ref{prop:a++}, \ref{prop:b++}, \ref{prop:c++}, and \ref{prop:d++}.

		We will first make use of the fact that $\fQ$ violates Conditions~\ref{cond:1++} and \ref{cond:2++}.
		If there is some $1 \leq j \leq r$ such that $\Q_j = \gen{\gen{\X}'}$, then the fact that $\fQ$ violates Condition~\ref{cond:1++} implies that there is some $j' \neq j$ satisfying $\Q_{j'} \in \{ \gen{\gen{\U}'}, \gen{\gen{\X}'} \}$, which in turn implies that for any choice $(\B_1, \dots, \B_r) \in \Q_1 \times \dots \times \Q_r$, we have $\B_j, \B_{j'} \in \sre^1$, by Corollary~\ref{cor:regen}.
		Since, by definition, each configuration from $\ere$ contains at most one label from $\sre^1$, it follows, by the definition of $\estar$, that $\fQ \notin \estar$, and we are done.
		Similarly, if there is some $1 \leq j \leq r$ such that $\Q_j = \gen{\gen{\X}}$, then the fact that $\fQ$ violates Condition~\ref{cond:2++} implies that there are two distinct indices $j', j'' \in \{ 1, \dots, r \} \setminus \{ j \}$ such that $\Q_{j'}, \Q_{j''} \in \{ \gen{\gen{\U}'}, \gen{\gen{\X}'} \}$, which analogously to above implies that $\fQ \notin \estar$, and we are done.
		Hence, assume in the following that $\Q_j \in \sstar \setminus \{ \gen{\gen{\X}}, \gen{\gen{\X}'} \}$, which by Corollary~\ref{cor:regen} and the definition of $\gen{\cdot}$ implies that
		\begin{equation}\label{eq:from12}
			 \Q_j \subseteq \sre \setminus \{ \gen{\X}, \gen{\X}' \} \text{ for each } 1 \leq j \leq r.
		\end{equation}

		Now, we will make use of the fact that $\fQ$ violates at least one of Conditions~\ref{prop:a++}, \ref{prop:b++}, \ref{prop:c++}, and \ref{prop:d++}.		
		Consider first the case that $\fQ$ violates Condition~\ref{prop:a++}, i.e., there is some $1 \leq j \leq r$ satisfying $\Q_j = \gen{\gen{\P}}$.
		It follows, by Corollary~\ref{cor:regen} and (\ref{eq:from12}), that for any choice $(\B_1, \dots, \B_r) \in \Q_1 \times \dots \times \Q_r$, we have $\B_j \in \{ \gen{\P}, \gen{\P}' \}$ for some $1 \leq j \leq r$ and $\B_{j'} \notin \{ \gen{\X}, \gen{\X}' \}$ for each $1 \leq j' \leq r$.
		By Conditions~\ref{cond:easy} and \ref{prop:p} in Lemma~\ref{lem:repidef}, it follows that $\B_1 \s \dots \s \B_r \notin \ere$ for all $(\B_1, \dots, \B_r) \in \Q_1 \times \dots \times \Q_r$, which in turn implies $\fQ \notin \estar$, as desired.

		Now, consider the case that $\fQ$ violates Condition~\ref{prop:b++}, i.e., there are at least two indices $j, j' \in \{ 1, \dots, r \}$ such that $\Q_j \in \{ \gen{\gen{\U}'}, \gen{\gen{\X}'} \}$.
		As already seen above, this implies $\fQ \notin \estar$, as desired.

		Next, consider the case that $\fQ$ violates Condition~\ref{prop:c++}, i.e., there are at least $z_q + 2$ indices $j \in \{1, \dots, r\}$ such that $\Q_j = \gen{\gen{\U}'}$ or $\Q_j = \gen{\gen{\ell(\ccc)}, \gen{\ell(\ccc \setminus \{ q \} )}'}$ for some color set $\ccc$ containing color $q$.
		Consider an arbitrary choice $(\B_1, \dots, \B_r) \in \Q_1 \times \dots \times \Q_r$.
		If there exists some index $1 \leq j \leq r$ satisfying $\B_j \in \{ \gen{\P}, \gen{\P}' \}$, then (\ref{eq:from12}) implies that $\B_1 \s \dots \s \B_r \notin \ere$, by Conditions~\ref{cond:easy} and \ref{prop:p} in Lemma~\ref{lem:repidef}; it follows that $\fQ \notin \estar$, and we are done.
		Hence, assume in the following that $\B_j \notin \{ \gen{\P}, \gen{\P}' \}$ for each $1 \leq j \leq r$.
		Since, by definition, $\ere$ contains at most one label from $\sre^1$, it follows, by Lemma~\ref{lem:restrong} and Corollary~\ref{cor:regen}, that there are at least $(z_q + 2) - 1$ indices $j \in \{1, \dots, r\}$ such that $\B_j = \gen{\ell(\ccc)}$ for some color set $\ccc \subseteq \ccs(z)$ that is a superset of a set $\ccc' \subseteq \ccs(z)$ containing color $q$.
		In other words, there are at least $z_q + 1$ indices $j \in \{1, \dots, r\}$ such that $\B_j = \gen{\ell(\ccc)}$ for some color set $\ccc \subseteq \ccs(z)$ containing color $q$.
		Combining this conclusion with the fact that, by (\ref{eq:from12}), $\B_{j'} \notin \{ \gen{\X}, \gen{\X}' \}$ for each $1 \leq j' \leq r$, we obtain that $\B_1 \s \dots \s \B_r \notin \ere$, by Conditions~\ref{cond:easy} and \ref{prop:i} in Lemma~\ref{lem:repidef}.
		It follows that $\fQ \notin \estar$, as desired.

		Finally, consider the case that $\fQ$ violates Condition~\ref{prop:d++}, i.e., there exists some color $i \in \ccs(z)$ satisfying $i \neq q$ such that there are at least $z_i + 1$ indices $j \in \{1, \dots, r\}$ such that $\Q_j = \gen{\gen{\ell(\ccc)}}$ for some color set $\ccc$ containing color $i$.
		Consider an arbitrary choice $(\B_1, \dots, \B_r) \in \Q_1 \times \dots \times \Q_r$.
		Analogously to the previous case, we obtain that the existence of an index $1 \leq j \leq r$ satisfying $\B_j \in \{ \gen{\P}, \gen{\P}' \}$ implies that $\fQ \notin \estar$.
		On the other hand, the nonexistence of such an index implies, by Corollary~\ref{cor:regen}, that there are at least $z_i + 1$ indices $j \in \{1, \dots, r\}$ such that $\B_j \in \{ \gen{\ell(\ccc)}, \gen{\ell(\ccc)}' \}$ for some color set $\ccc \subseteq \ccs(z)$ containing color $i$; by (\ref{eq:from12}) and Conditions~\ref{cond:easy} and \ref{prop:i} in Lemma~\ref{lem:repidef}, it follows that $\B_1 \s \dots \s \B_r \notin \ere$, which in turn implies that $\fQ \notin \estar$, as desired.
	\end{proof}

	\subsection{Renaming}
	We now show that $\pistar$ can be relaxed to  $\Pi(z',s')))$ for some parameters $z'$ and $s'$.
	\onestep*
	\begin{proof}
		We show that $\pistar$, which we proved to be a relaxation of $\rere(\re(\Pi(z,s)))$, is equivalent to $\Pi(z',q)$, if we rename the labels correctly.
		Consider the following renaming:
		\begin{align*}
			\gen{\gen{\X}'} &\hspace{0.5cm}\rightarrow\hspace{0.5cm} \M \\
			\gen{\gen{\P}}  &\hspace{0.5cm}\rightarrow\hspace{0.5cm} \P \\
			\gen{\gen{\U}}  &\hspace{0.5cm}\rightarrow\hspace{0.5cm} \U \\
			\gen{\gen{\U}'} &\hspace{0.5cm}\rightarrow\hspace{0.5cm} \D \\
			\gen{\gen{\X}}  &\hspace{0.5cm}\rightarrow\hspace{0.5cm} \X \\
			\gen{\gen{\ell(\ccc)}} \text{ for each } \emptyset \neq \ccc \subseteq \ccs(z) \text{ satisfying } q \notin \ccc  &\hspace{0.5cm}\rightarrow\hspace{0.5cm} \ell(\ccc) \\
			\gen{\gen{\ell(\ccc)}, \gen{\ell(\ccc \setminus \{ q \})}'} \text{ for each } \emptyset \neq \ccc \subseteq \ccs(z) \text{ satisfying }q \in \ccc  &\hspace{0.5cm}\rightarrow\hspace{0.5cm} \ell(\ccc) \\
		\end{align*}
		Observe that, under this renaming, $\nstar$, as defined in \Cref{def:pistar}, becomes equal to  $\nodeconst(z',q)$, and $\estar$, as characterized in \Cref{lem:estarcharacterization}, becomes equal to $\edgeconst(z',q)$.
	\end{proof}

\section{Deterministic Upper Bounds for Hypergraph MIS}\label{sec:ub}

In this section, we present deterministic algorithms for solving the maximal independent set problem on hypergraphs of maximum degree $\Delta$ and rank $r$. We start with a straightforward algorithm that solves the problem in $O(\Delta r + \log^* n)$ rounds. We then describe an algorithm that has a quadratic dependency on $\Delta$, but only a logarithmic dependency on $r$, achieving a time complexity of roughly $O(\Delta^2 \log r + \log^* n)$ rounds. Finally, we present an algorithm that, perhaps surprisingly, almost does not depend at all on $r$, and that runs in $O(f(\Delta)\cdot\log^*(\Delta r) + \log^* n)$ deterministic rounds, for some function $f$.

\paragraph{Straightforward Algorithm.} We start by describing a simple algorithm that solves the MIS problem on hypergraphs in $O(\Delta r + \log^* n)$ deterministic rounds. The algorithm works as follows: we first compute a coloring of the nodes using $O(\Delta r)$ colors; then, we go through color classes and add a node to the MIS if and only if it does not have any incident hyperedge $e$ with $\rank(e)-1$ incident nodes already in the MIS.

It is easy to see that, for computing the coloring, we can compute a distance-$2$ $(\bar{\Delta} + 1)$-coloring on the incidence graph $G$ of $H$, where $\bar{\Delta}\leq\Delta r$ is the maximum degree in the power graph $G^2$. This can be done in $\widetilde{O}(\sqrt{\Delta r} + \log^* n)$ rounds \cite{MausTonoyan20,fraigniaud16local}. Now we can go through color classes and safely add a node (of $H$) in the set if, by doing so, we still satisfy the desired constraints of the problem, and this requires $O(\Delta r)$ rounds. Hence, in total, we get a time complexity of $O(\Delta  r + \log^* n)$ rounds. Therefore, we obtain the following theorem.

\begin{theorem}\label{thm:mis-straightforward}
	The MIS problem on hypergraphs can be solved in $O(\Delta r + \log^* n)$ deterministic rounds on hypergraphs with maximum degree $\Delta$ and rank $r$.
\end{theorem}

\subsection{Slow-in-$\Delta$ Algorithm}
We now introduce an algorithm that solves the MIS problem on hypergraphs in a time that has a slow dependency on $\Delta$, but only a logarithmic dependency on $r$. More precisely, we prove the following.

\MISalgI*

In this algorithm, each node is either \emph{active} or \emph{inactive}. The algorithm terminates when all nodes are inactive. At the beginning, each node is in the active state. The algorithm is as follows.
\begin{enumerate}
	\item Compute an $O((\Delta r)^2)$ coloring of the nodes of $H$.\label{initial}
	\item Let $H'$ be the hypergraph induced by active nodes. Active nodes construct a virtual graph $G$ by splitting each hyperedge $e$ of $H'$ into $\lceil \rank_{H'}(e)/2\rceil$ virtual hyperedges of rank at most $2$.\footnote{This is a standard procedure, where nodes can easily agree on which nodes are connected to a given virtual edge: for example, there is a virtual edge created from the hyperedge $e$ that connects the two nodes that, among all nodes incident to $e$, are the ones with the smallest ID; then there is another virtual edge that connects the second pair of nodes with smallest ID, and so on.}\label{virtualG}
	
	\item Color nodes of $G$ with $O(\Delta)$ colors.\label{coloring}
	
	\item In the hypergraph $H'$, go through color classes, and while processing a node $v$, do the following.\label{outOrIn}
	
	\begin{enumerate}
		\item If $v$ has at least one incident hyperedge $e$ having $\rank_{H'}(e)-1$ nodes already in the MIS, $v$ considers itself outside the MIS and becomes inactive.\label{outside}
		
		\item If, for each hyperedge $e$ incident to $v$, there exists another node $u$ incident to $e$ such that $\colour(v)\neq \colour(u)$ and $u$ is not in the MIS (either because it got processed and decided to not enter the MIS, or because it has not been processed yet), then $v$ enters the MIS and becomes inactive.\label{inside}
		
		\item Otherwise, $v$ does nothing and remains active.\label{undecided}
	\end{enumerate}
	
	\item Repeat from \Cref{virtualG} until all nodes are inactive.\label{repeat}
\end{enumerate}

\paragraph{Correctness and Time Complexity.}
Nodes can spend $O(\log^* n)$ rounds to compute the initial coloring in \Cref{initial} \cite{Linial1992}.
 
Nodes can create the virtual graph $G$ as in \Cref{virtualG} in constant time. Then, the coloring described in \Cref{coloring} can be done in $\widetilde{O}(\sqrt{\Delta} + \log^* (\Delta r))$ rounds, by exploiting the precomputed initial coloring \cite{MausTonoyan20,fraigniaud16local}. Notice that, while this is a proper coloring of the virtual graph, this results into a defective coloring of the hypergraph $H'$, and this defective coloring has some desired properties. In fact, since in $G$ we have a proper coloring, then every pair of nodes connected to the same virtual edge must have different colors, implying that, for each hyperedge $e$ of $H'$, at most $\lceil \rank_{H'}(e)/2 \rceil$ incident nodes have the same color.
 
 In \Cref{outOrIn} we process nodes going through color classes, and based on some conditions that can be verified in constant time, nodes decide what to do: enter or not the MIS and become inactive, or be still undecided and remain  active. Hence, since we have $O(\Delta)$ many colors, this part can be executed in $O(\Delta)$ rounds. While the coloring in the virtual graph $G$ is a proper one, the coloring in $H'$ is not a proper coloring, and therefore neighboring nodes that have the same color get processed at the same time, but this is not an issue. In fact, \Cref{outside} and \Cref{inside} guarantee that, for each hyperedge $e$, at most $\rank(e) - 1$ incident nodes enter the MIS. Also, \Cref{inside} guarantees that, as long as it is safe to enter the MIS, a node will do so. All nodes that remain still active in \Cref{undecided} are basically nodes that still have to decide whether to join or not the MIS. A node $v$ remains active if both of the following conditions hold:
 \begin{itemize}
 	\item for each incident hyperedge $e$ it holds that at most $\rank_{H'}(e) - 2$ nodes incident to $e$ are in the MIS (otherwise \Cref{outside} would apply);
 	\item there must exist a hyperedge $e$ incident to $v$ such that all active nodes incident to $e$ have the same color as $v$ (otherwise \Cref{inside} would apply).
 \end{itemize} 
 Since, as discussed above, for each hyperedge at most half of its incident nodes have the same color, this means that, in the hypergraph induced by active nodes, for each node it holds that there exists at least one incident hyperedge that in $H'$ has rank $x$ and in the hypergraph induced by active nodes of the next phase has rank at most $\lceil{x/2\rceil}$. Hence, for each node, in each phase, the rank of one incident hyperedge at least halves. This means that, after repeating this process for $O(\Delta\log r)$ many times, we get that all nodes are inactive, achieving a total runtime of $O(\Delta\log r (\Delta + \log^* (\Delta r))+ \log^* n) = O(\Delta\log r (\Delta + \log^* r)+ \log^* n)$, and hence proving \Cref{thm:mis-slowInDelta}.

\subsection{(Almost) Independent-of-$r$ Algorithm}
In this section, we present an algorithm for solving the MIS problem on hypergraphs, the running time of which almost does not depend at all on the rank $r$. We will first start by presenting an algorithm that, assuming we are given an $O((\Delta r)^2)$-coloring of the nodes, computes an MIS in $O(\log^* (\Delta r))$ rounds, in the case where the maximum degree of the nodes is $2$. Then we show how to generalize these ideas and present an algorithm for the general case of $\Delta>2$.

\subsubsection{The Case $\Delta=2$.}
Let $H$ be a hypergraph with maximum degree $2$ and rank $r$, where we are given in input an $O((\Delta r)^2)$-coloring of the nodes (we will later get rid of this assumption and actually compute this coloring, but for now, assume we have it for free). We assume that all nodes have degree exactly $2$. In fact, nodes of degree $1$ can be easily handled: at the beginning we remove them from the hypergraph, then we solve the problem, and then we add all of them in the solution (except for the case in which a hyperedge is only incident to such nodes, in which case we leave one out).

Since nodes in $H$ have degree $2$, we can see them as edges, and create a virtual graph $G$ with maximum degree $r$, where nodes are the hyperedges of $H$, and there is an edge between two nodes in $G$ if and only if the hyperedge that they correspond to share a node in $H$. Note that $G$ may have parallel edges (if $H$ is not a linear hypergraph).
 Observe that any $T$-rounds algorithm that is designed to run in $G$ can be simulated by nodes in $H$ in $O(T)$-rounds. In the following, we show an algorithm that computes a subset $S$ of the edges of $G$ and then show that this set is an MIS for $H$. On a high level, an edge $e$ of $G$ will result outside the set $S$ if and only if there exists one endpoint $v$ of $e$ such that it has already $\deg(v) - 1$ incident edges in $S$. Initially, let $S=\emptyset$. The algorithm works as follows.

\begin{enumerate}
	\item Discard parallel edges. That is, for any pair of nodes $\{u,v\}$ connected by at least one edge, we keep exactly one edge.\label{noparallel}
	
	\item Compute a $(2,2)$-ruling edge set of $G$ (i.e., a $(2,2)$-ruling set of the line graph of $G$).\label{rulingset}

	\item Each node $v$ marks itself with its distance to the nearest edge in the ruling set, resulting with nodes marked with a number in $\{0, 1, 2\}$. More precisely: \label{marking}
	\begin{itemize}
		\item if $v$ is an endpoint of a ruling-set edge, then $v$ is marked with $0$;\label{marked0}
		\item otherwise, if there is a neighbor of $v$ that is an endpoint of a ruling-set edge, then $v$ is marked with $1$;\label{marked1}
		
		\item otherwise, $v$ is marked with $2$ (note that, since we have a $(2,2)$-ruling set of the edges, there is no node at distance larger than $2$ from a ruling-set edge, and nodes marked $2$ form an independent set).\label{marked2}
	\end{itemize}
	
	\item Each node marked with $2$ proposes to an arbitrary neighbor (which is marked with $1$) to add the edge between them in the ruling set.\label{2sPropose}
	
	\item Each node marked with $1$ that receives at least one proposal from a $2$-marked node, accepts exactly one of them and rejects the others (choosing arbitrarily).\label{no2s}
	
	\item Add to the ruling set the edges over which the proposals are accepted (we still obtain a $(2,2)$-ruling set).\label{modifiedRS}
	
	\item Nodes recompute the distances to the nearest ruling-set edge (note that each node now gets marked with either $0$ or $1$).\label{Re-marking}
	
	\item Put back removed parallel edges. None of them is added to the ruling set. Observe that we still have a $(2,2)$-ruling set and the distances of the nodes from edges of the ruling set does not change. \label{addparallel}
	
	\item The set $S$ contains the edges selected according to the following rules.\label{setS}
	
	\begin{enumerate}
		\item Add to $S$ all edges not in the ruling set.\label{notRS-into-S}
		
		\item Each $1$-marked node removes from $S$ one incident edge towards a neighbor marked $0$, breaking ties arbitrarily (note that there must exist at least one such edge).\label{1-0-outside-S}
		
		\item For each edge in the ruling set, add it to $S$ if and only if both endpoints have at least two incident edges not in $S$.\label{rs-in-S}
	\end{enumerate}
	
\end{enumerate}

\paragraph{Correctness and Time Complexity.} We now show that the set $S$ constructed by the above algorithm satisfies the following claim.

\begin{claim}\label{claim:setS}
	Each node $u$ has at most $\deg(u) - 1$ incident edges in $S$, and for each edge not in $S$ it holds that at least one of its endpoints has all other incident edges in $S$.
\end{claim}
Firstly, we argue that, if we prove \cref{claim:setS}, then it directly implies that our solution computes an MIS in $H$. Recall that, by construction, nodes of $G$ correspond to the hyperedges of $H$, and edges of $G$ correspond to the nodes of $H$. Hence, if the above claim is true, it means that the algorithm returns a subset of nodes of $H$ such that, for each hyperedge $e$, at most $\rank(e) - 1$ incident nodes are selected, and for each non-selected node $v$, there exists a hyperedge $e$ incident to $v$ that has $\rank(e) - 1$ incident selected nodes. This, by definition, means that the set $S$ in $G$ corresponds to an MIS in $H$. Hence, in order to show the correctness, it is enough to prove the above claim, which we do in the following.

A $(2,2)$-ruling set of the edges (\Cref{rulingset}) can be computed using the algorithm of \cite{KuhnMW18}, which runs in $O(\log^* (\Delta r))$ rounds if we are given in input an $O((\Delta r)^2)$-coloring. Let $X$ be the set of all nodes that are an endpoint of an edge in the ruling set. Nodes can spend $O(1)$ rounds and mark themselves with the distance from the nearest node in $X$ (\Cref{marking}). It is easy to see that a $(2,2)$-ruling set of the edges of $G$ implies that 
\begin{itemize}
	\item each node in $G$ is at distance at most $2$ from a node in $X$, and
	\item nodes marked with $2$ form an independent set (that is, all their neighbors are marked with $1$).
\end{itemize}
In \Cref{2sPropose}, \Cref{no2s}, and \Cref{modifiedRS}, we modify the ruling set obtained in \Cref{rulingset} such that each node is at distance at most $1$ from a node in $X$. For this, nodes marked with $2$ propose to a $1$-marked neighbor to put the edge between them into the ruling set, and $1$-marked nodes accept only one proposal (if they receive any). Note that, since edges between $1$-marked and $2$-marked nodes are not incident to the edges of the ruling set, and since $1$-marked nodes accept only one proposal, then we still obtain a $(2,2)$-ruling set. Let $u$ be the node marked $2$ that performed a proposal to its neighbor $v$ marked with $1$. If the proposal of $u$ gets accepted, then the edge $(u,v)$ enters the ruling set and nodes $u$ and $v$ will be marked with $0$ in \Cref{Re-marking}. Otherwise, if the proposal gets rejected, it means that $v$ accepted another proposal from another $2$-marked neighbor, meaning that, in \Cref{Re-marking}, node $v$ will be marked with $0$ and hence node $u$ will be marked with $1$. All these operations can be done in $O(1)$ rounds. Therefore, we remain with a $(2,2)$-ruling set of the edges where each node is either marked with $0$ or with $1$. 

Now, in \Cref{addparallel}, we put back the parallel edges removed in \Cref{noparallel}, and we put none of them in the ruling set. Observe that we still have a $(2,2)$-ruling set, and the distances of the nodes from the nearest ruling set edge does not change.

Then, we start constructing our set $S$ according to the rules in \Cref{setS}, which can all be accomplished in constant time. Since each $0$-marked node $u$ has exactly one incident edge in the ruling set, \Cref{notRS-into-S} guarantees that at most $\deg(u) - 1$ incident edges are in $S$. On the other hand, since each $1$-marked node $v$ has at least one incident edge towards a $0$-marked neighbor, \Cref{1-0-outside-S} guarantees that exactly $\deg(v) - 1$ incident edges are in $S$. Hence, nodes satisfy \Cref{claim:setS}. Also, edges incident to $1$-marked nodes satisfy \Cref{claim:setS}, since each marked $1$ node has exactly $\deg(v) - 1$ incident edges in $S$. The only remaining edges to be analyzed are the ones in the ruling set, and according to \Cref{rs-in-S}, they enter the set $S$ if they can. Hence, the set $S$ satisfies \Cref{claim:setS}. Regarding the running time, except for the first item, all other items can be performed in $O(1)$ rounds, achieving a total runtime of $O(\log^* (\Delta r))$ rounds. This implies the following lemma.

\begin{lemma}\label{lem:mis-Delta-2}
	The MIS problem can be solved in $O(\log^* (\Delta r))$ deterministic rounds on hypergraphs with maximum degree $2$ and rank $r$, assuming that an $O((\Delta r)^2)$-coloring is given in input.
\end{lemma}

\subsubsection{Generalization for $\Delta > 2$}\label{subsubsec:Delta>2}
On a high level, we show that, if we have an algorithm that solves MIS on $O((\Delta r)^2)$-colored hypergraphs of maximum degree $\Delta -1$ and rank $r$, then we can use it as a black box for constructing another algorithm that solves MIS on $O((\Delta r)^2)$-colored hypergraphs of maximum degree $\Delta$ and rank $r$. Hence, by starting from our base-case algorithm for hypergraphs of degree at most $2$, and applying this reasoning in an iterative way, we obtain an algorithm for solving MIS on $O((\Delta r)^2)$-colored hypergraphs of any maximum degree $\Delta$ and any rank $r$.

Therefore, let $\mathcal{A}_{\Delta -1}$ be an algorithm that is able to solve MIS on any $O((\Delta r)^2)$-colored hypergraph $H_{\Delta-1}=(V_{\Delta - 1}, E_{\Delta - 1})$ of maximum degree $\Delta-1$ and rank $r$. In the following, we present an algorithm that solves MIS on $O((\Delta r)^2)$-colored hypergraphs $H_{\Delta}=(V_\Delta, E_\Delta)$ with maximum degree $\Delta$ and rank $r$. At the beginning, all nodes are \emph{active}.

\begin{enumerate}
	\item Compute a $(2,\Delta +4)$-ruling set of $H_{\Delta}$.\label{item:hypergraph-rs}
	
	\item Mark hyperedges with their distance from the nearest node in the ruling set, that is, hyperedges get marked with a number in $\{0, 1, \dotsc, \Delta +4\}$. Let $E^i$ be the set of hyperedges marked $i$.\label{item:mark-edges}
	
	\item Initialize $\mbox{MIS}:=\emptyset$.\label{item:initialize-S}
	
	\item For $i=\Delta +4$ to $0$ do:\label{item:for-loop}
	\begin{enumerate}	

		\item Let $S^i$ be the set of active nodes of $V_\Delta$ incident to \emph{at least one} $i$-marked hyperedge, and let $X\subseteq S^i$ be the set of active nodes that have \emph{all} incident hyperedges marked with $i$.\label{item:defsi}
		
		\item Consider the hypergraph $H^i = (S^i\setminus X,\hat{E}^i)$, where $\hat{E}^i = \{ e \cap (S^i \setminus X) ~|~ e \in E^i \}$, that is, the hypergraph induced by nodes in $S^i \setminus X$ and hyperedges in $E^i$.
		This hypergraph has maximum degree $\Delta-1$. On $H^i$, simulate $\mathcal{A}_{\Delta -1}$.\label{item:simulate-algo}
		
		\item Nodes that are part of the MIS of $H^i$ join the MIS. \label{item:output-to-mis}
		
		\item All nodes of $X$ join the MIS.\label{item:x-in-mis}
		
		\item All nodes of $S^i$ become passive. \label{item:becomepassive}
		
		\item All active nodes that are incident to at least one hyperedge $e$ having $\rank_{H_\Delta}(e) - 1$ incident nodes in the MIS become passive.\label{item:out-of-MIS}
	
	\end{enumerate}
\end{enumerate}

\paragraph{Correctness.}
In the algorithm, after performing the ruling set on the hypergraph (\Cref{item:hypergraph-rs}), and after marking the hyperedges with their distance to the nearest ruling-set node (\Cref{item:mark-edges}), we start to construct our solution. 
We go through the marked hyperedges in order, from the largest to the smallest, and \Cref{item:defsi,item:out-of-MIS} guarantee that all nodes incident to at least one hyperedge marked at least $i+1$ are now passive.

Let $H^{\markedi}$ be the hypergraph induced by $E^i$, that is, the hypergraph containing all hyperedges marked $i$ and active nodes incident to at least one of them. This hypergraph has maximum degree $\Delta$, but then, we construct $H^i$ by removing from it all nodes of degree exactly $\Delta$ (these nodes are then part of $X$). Hence, $H^i$ has maximum degree $\Delta-1$.  Therefore we can compute an MIS of $H^i$ using algorithm $\mathcal{A}_{\Delta-1}$ (\Cref{item:simulate-algo}), and obtain that, in $H^i$, each hyperedge $e$ has at most $\rank_{H^i}(e)-1$ incident nodes in the MIS. 

We claim that the operations performed in \Cref{item:output-to-mis,item:x-in-mis} produce an MIS in $H^{\markedi}$. By construction, the set is an MIS in $H^i$, and hence an independent set in $H^{\markedi}$. Since we add all the nodes of $X$, we only need to show that we do not ruin the independence property.
Observe that, for each $i$-marked hyperedge, it holds that there exists at least one incident node connected to an $(i-1)$-marked hyperedge. Hence, each hyperedge $e$ is connected to at most $\rank_{H^{\markedi}}(e) -1$ nodes in $X$, or in other words, it is never the case that a hyperedge $e$, in $H^{\markedi}$, has only neighbors that are part of $X$. Also, for each $i$-marked hyperedge incident to $x>0$ nodes in $X$, it holds that at most $\rank_{H^{\markedi}}(e) - x - 1$ incident nodes are selected. Hence, by adding the nodes in $X$, we still obtain that at most $\rank_{H^{\markedi}}(e) - 1$ incident nodes are selected. 

Now we need to argue that the solution that we get is indeed an MIS of $H_\Delta$, and we show this by induction. As a base case of our inductive argument, as just showed, our algorithm computes an MIS of the hypergraph induced by hyperedges marked with $\Delta + 4$. Hence, suppose that we have an MIS $M^{>i}$ on the hypergraph induced by hyperedges marked at least $i+1$. We show that, after computing an MIS $M^i$ of the hypergraph induced by $i$-marked hyperedges, we get an MIS on the hypergraph induced by hyperedges marked at least $i$. In order to see that this holds it is enough to notice that, at each phase of the for-loop, we consider all nodes that could potentially be part of the MIS, taking into consideration the already computed partial solution. More precisely, when we execute \Cref{item:out-of-MIS} in phase $i+1$, we ensure that, in phase $i$, we construct the MIS set $M^i$ by considering \emph{all} nodes that are incident to at least one $i$-marked hyperedge, and that can potentially be part of the MIS. Therefore, the set $M^{>i}\cup M^{i}$ is an MIS of the hypergraph induced by hyperedges marked at least $i$. Hence, after the $\Delta + 5$ steps performed in the for-loop, we get an MIS of our hypergraph $H_{\Delta}$.

\paragraph{Time Complexity.}
We now provide a recursive formula for the above algorithm and prove the following lemma.

\begin{lemma}\label{lemma:from-Delta-1-to-Delta}
	Let $\mathcal{A}_{\Delta-1}$ be a $T$-rounds algorithm that solves MIS on $O((\Delta r)^2)$-colored hypergraphs of maximum degree $\Delta - 1$ and maximum rank $r$. Then, the algorithm described in \Cref{subsubsec:Delta>2} solves MIS on $O((\Delta r)^2)$-colored hypergraphs of maximum degree $\Delta$ and maximum rank $r$ in time $O(\Delta (1+T) + \log^* (\Delta r))$.
\end{lemma}
\begin{proof}
	\Cref{item:hypergraph-rs} of the algorithm in \Cref{subsubsec:Delta>2} can be done in $O(\Delta + \log^*(\Delta r))$ rounds by using the algorithm of \cite[Corollary 1.6]{KuhnMW18}\footnote{Their algorithm is phrased as a $(2,\ell+4)$-ruling hyperedge set on hypergraphs of rank at most $\ell$. This algorithm can be trivially converted, in the \LOCAL model, into an algorithm that finds a $(2,\ell+4)$-ruling set on hypergraphs of degree at most $\ell$, by reversing the roles of nodes and hyperedges.}. Then, nodes spend $O(\Delta)$ rounds to perform \Cref{item:mark-edges}. In the for-loop, we then execute $O(\Delta)$ times $\mathcal{A}_{\Delta - 1}$, while the other items of the for-loop can be done in constant time. Hence, we get a total runtime of $O(\Delta (1+T) + \log^* (\Delta r))$ rounds.
\end{proof}

\paragraph{Putting Things Together.}
Finally, we show that the recursive formula that expresses the runtime of our algorithm obtained in \Cref{lemma:from-Delta-1-to-Delta} can be upper bounded by $2^{O(\Delta\log\Delta)}\log^* r$. We prove the following lemma.

\begin{lemma}\label{lem:solve-recursive}
	The MIS problem can be solved in $2^{O(\Delta\log\Delta)}\log^*(\Delta r) = 2^{O(\Delta\log\Delta)}\log^* r$ deterministic rounds on $O((\Delta r)^2)$-colored hypergraphs of maximum degree $\Delta$ and maximum rank $r$.
\end{lemma}
\begin{proof}
	Let $T_d$ be the runtime of an algorithm that solves MIS on hypergraphs of maximum degree $d$ and maximum rank $r$, given an $O((\Delta r)^2)$-coloring in input. By \Cref{lem:mis-Delta-2}, we know that $T_2 \le c_1\log^*(\Delta r)$, and by \Cref{lemma:from-Delta-1-to-Delta}, $T_{d+1}\le c_1\Delta (1+T_d) + c_1 \log^*(\Delta r)$, for some constant $c_1$. We show, by induction, that this results in the runtime stated in this lemma.
	
	As a base case we consider hypergraphs with maximum degree $2$, that clearly holds by \Cref{lem:mis-Delta-2}. Hence, assuming that $T_\Delta \le 2^{c\Delta\log\Delta}\log^*(\Delta r)$, we show that $T_{\Delta+1} \le 2^{c(\Delta+1)\log(\Delta+1)}\log^*(\Delta r)$.
	\begin{align*}
		T_{\Delta + 1} &\le c_1\Delta (1+T_\Delta) + c_1\log^*(\Delta r) \mbox{ (by \Cref{lemma:from-Delta-1-to-Delta})}\\
		&\le c_1\Delta (1+2^{c\Delta\log\Delta}\log^*(\Delta r)) + c_1\log^*(\Delta r) \mbox{ (by inductive hypothesis)}\\
		&\le c_1 (\Delta 2^{c\Delta\log\Delta + 1} + 1)\log^*(\Delta r)\\
		&\le 2^{c\Delta\log\Delta + \log c_1 + \log\Delta + 2}\log^*(\Delta r)\\
		&\le 2^{c\Delta\log\Delta + c_1\log\Delta}\log^*(\Delta r)\\
		&\le 2^{c(\Delta+1)\log(\Delta+1)}\log^*(\Delta r), \mbox{ for a large enough } c.
	\end{align*}
\end{proof} 
Up until now we have assumed an $O((\Delta r)^2)$-coloring of the nodes in input. We can compute this coloring in $O(\log^* n)$ rounds using the algorithm of \cite{Linial1992}, and together with \Cref{lem:solve-recursive}, we obtain \Cref{thm:MIS-independent-of-r}, which is restated here for completeness.

\MISalgII*

\section{Open Questions}\label{sec:open}
In this work, we showed that, for hypergraph maximal matching, the $O(\Delta r + \log^* n)$ algorithm is optimal, in the sense that in order to improve the $\Delta r$ dependency it is required to spend much more as a function of $n$. Also, we showed that the same does not hold for hypergraph MIS, by providing two algorithms, one with complexity $O(\Delta^2 \log r  + \Delta \log r \log^* r + \log^* n)$, and the other with complexity   $2^{O(\Delta\log\Delta)}\log^* r + O(\log^* n)$. Hence, if $\Delta$ is constant but $r$ very large, we can still solve hypergraph MIS in just $O(\log^* n)$ rounds. Unfortunately, these algorithms do not match the lower bound, since the only known lower bound comes from standard MIS, and says that this problem requires $\Omega(\min\{\Delta,\log_\Delta n\})$ rounds \cite{Balliu2019,hideandseek}. Observe that this lower bound is tight for MIS on graphs, that can in fact be solved in $O(\Delta + \log^* n)$ rounds, and there is the possibility that also hypergraph MIS could be solved in this time (i.e., there is no lower bound that prevents this).

We now show a possible direction that could lead to solving hypergraph MIS in just $O(\Delta + \log^* n)$ rounds. Consider the following variant of hypergraph coloring, for which our lower bounds do \emph{not} apply. 
\begin{definition}
	A $c$-unique-maximum coloring is a labeling of the nodes of the hypergraph, such that each node has a label in $\{1,\ldots,c\}$ and each hyperedge $e$ satisfies that, if $C$ is the set of colors used by the nodes incident to $e$, then the maximum element of $C$ is used by only one node incident to $e$.
\end{definition}
We informally state a curios fact about hypergraph MIS. By applying round elimination on it for $k$ times, it seems that we obtain a problem that can be decomposed into three parts:
\begin{itemize}
	\item The \emph{original} problem;
	\item A \emph{natural} part, the $k$-unique-maximum coloring;
	\item An \emph{unnatural} part, of size roughly equal to a power tower of height $k$, that cannot be easily understood.
\end{itemize}
We do not formally prove that this is indeed the result that we get by applying round elimination on hypergraph MIS, but we instead use this observation as a suggestion for a possible algorithm. In fact, we now prove that we can solve hypergraph MIS fast if we are given a suitable unique maximum coloring.

\begin{theorem}
	If the $O(\Delta)$-unique-maximum coloring problem can be solved in $O(\Delta + \log^* n)$ rounds, then also hypergraph MIS can be solved in $O(\Delta + \log^* n)$ rounds. 
\end{theorem}
\begin{proof}
	We show that, after spending $O(\Delta + \log^* n)$ rounds to compute an $O(\Delta)$-unique-maximum coloring, then we can spend only $O(\Delta)$ rounds to solve hypergraph MIS.
	
	The algorithm is exactly the same as the trivial algorithm: we process nodes by color classes, from the smallest to the largest, and we add the nodes to the MIS, if possible. This clearly requires $O(\Delta)$ rounds. The unique-maximum property of the coloring guarantees that, for each hyperedge $e$, when processing the node $v$ of the largest color incident to it, all other nodes in $e\setminus v$ have been already processed, and hence there could be two scenarios:
	\begin{itemize}
		\item either there are some nodes incident to $e$ that did not join the MIS, and in that case $v$ joins the MIS if allowed by the other hyperedges incident to $v$, or
		\item all the other incident nodes to $e$ have already been processed and joined the MIS, and hence $v$ decides not to join.
	\end{itemize}
	The independence property is guaranteed by the fact that, for each hyperedge $e$, among all nodes in $e$, the last one to be processed is $v$, which does not enter the MIS if all nodes in $e\setminus v$ are in the MIS. The maximality property is obtained for the following reason. If a node does not have the largest color in any of its incident hyperedges, then it always enters the MIS, since, for each incident hyperedge $e$, there is at least one node that is not (potentially yet) in the MIS, and that is the node with the largest color in $e$. This means that, if there is a node $v$ that decides not to join the MIS, then there must exist a hyperedge $e$ such that $v$ is the node with the unique maximum color among all nodes in $e$, and all nodes in $e\setminus v$ already joined the MIS, hence guaranteeing maximality.
\end{proof}

We find very fascinating that a problem that has been studied in very different contexts (see, e.g., \cite{uniquemaximum, CFKLMMPSSWW07, Har-PeledS05}) appears as a natural subproblem when applying round elimination on hypergraph MIS. Unfortunately, we do not know what is the complexity of $O(\Delta)$-unique-maximum coloring, and we leave it as an open question to determine its complexity. Note that a $(\Delta + 1)$-unique-maximum coloring always exists and that it can be computed distributedly in $O(\Delta \cdot T_{\mathrm{hmis}})$ rounds, where $T_{\mathrm{hmis}}$ is the time required to compute a hypergraph MIS. In fact, it can be computed by doing the following for $\Delta+1$ times: at step $i$ compute a hypergraph MIS, color the nodes in the obtained set with color $i$, remove the colored nodes and the hyperedges of rank $1$, and continue in the residual graph. The reason why it takes at most $\Delta+1$ iterations is that for every node $v$ that remains it must hold that for at least one incident hyperedge all incident nodes except $v$ joined the set, and hence the degree of $v$ decreases by at least $1$.

\begin{oq}
	What is the distributed complexity of $O(\Delta)$-unique-maximum coloring?
\end{oq}

It could be that our algorithms are actually optimal, and that the current lower bound is not tight. Another possibility is that there is a fast algorithm, that is not based on first computing a unique-maximum coloring. Hence, we leave as open question determining the exact complexity of hypergraph MIS.
\begin{oq}
	Is it possible to solve hypergraph MIS in $O(\Delta + \log^* n)$? What are the possible tradeoffs between $r$ and $\Delta$ in the complexity of hypergraph MIS, when restricting the $n$ dependency to be $O(\log^* n)$?
\end{oq}

Finally, proofs based on round elimination are getting harder and harder. We believe that there is still a lot to be understood about the round elimination technique. Specifically, it would be nice to find a way to prove round elimination statements without writing a tedious case analysis. We leave as an open question finding simpler ways to apply round elimination.
\begin{oq}
	Is there a way to apply round elimination that does not require to write complicated proofs?
\end{oq}

% \clearpage
% \setcounter{page}{0}
% \thispagestyle{empty}
% \tableofcontents

%\clearpage

\urlstyle{same}
\bibliographystyle{alpha}
\bibliography{references}

\newcommand{\etalchar}[1]{$^{#1}$}
\begin{thebibliography}{GKMU18}

\bibitem[ABI86]{Alon1986}
Noga Alon, L{\'{a}}szl{\'{o}} Babai, and Alon Itai.
\newblock {A fast and simple randomized parallel algorithm for the maximal
  independent set problem}.
\newblock {\em Journal of Algorithms}, 7(4):567--583, 1986.

\bibitem[AGLP89]{Awerbuch89}
Baruch Awerbuch, Andrew~V. Goldberg, Michael Luby, and Serge~A. Plotkin.
\newblock Network decomposition and locality in distributed computation.
\newblock In {\em Proc.\ 30th Symp.\ on Foundations of Computer Science
  (FOCS)}, pages 364--369, 1989.

\bibitem[AH05]{strongcoloring}
Geir Agnarsson and Magn{\'u}s~M. Halld{\'o}rsson.
\newblock Strong colorings of hypergraphs.
\newblock In {\em Approximation and Online Algorithms}, 2005.

\bibitem[BBE{\etalchar{+}}20]{binary}
Alkida Balliu, Sebastian Brandt, Yuval Efron, Juho Hirvonen, Yannic Maus,
  Dennis Olivetti, and Jukka Suomela.
\newblock Classification of distributed binary labeling problems.
\newblock In {\em Proc.\ 34th Symp.\ on Distributed Computing ({DISC})}, 2020.

\bibitem[BBH{\etalchar{+}}19]{Balliu2019}
Alkida Balliu, Sebastian Brandt, Juho Hirvonen, Dennis Olivetti, Mika{\"{e}}l
  Rabie, and Jukka Suomela.
\newblock Lower bounds for maximal matchings and maximal independent sets.
\newblock In {\em Proc.\ 60th {IEEE} Symp.\ on Foundations of Computer Science
  (FOCS)}, pages 481--497, 2019.

\bibitem[BBKO21]{BBKOmis}
Alkida Balliu, Sebastian Brandt, Fabian Kuhn, and Dennis Olivetti.
\newblock Improved distributed lower bounds for {MIS} and bounded (out-)degree
  dominating sets in trees.
\newblock In {\em Proc.\ 40th ACM Symposium on Principles of Distributed
  Computing (PODC)}, 2021.

\bibitem[BBKO22]{hideandseek}
Alkida Balliu, Sebastian Brandt, Fabian Kuhn, and Dennis Olivetti.
\newblock Distributed {\(\Delta\)}-coloring plays hide-and-seek.
\newblock In {\em Proceedings of the 54th Annual {ACM} {SIGACT} Symposium on
  Theory of Computing, (STOC 2018)}, 2022.

\bibitem[BBO20]{balliurules}
Alkida Balliu, Sebastian Brandt, and Dennis Olivetti.
\newblock Distributed lower bounds for ruling sets.
\newblock In {\em Proc.\ 61st {IEEE} Symp.\ on Foundations of Computer Science
  (FOCS)}, pages 365--376, 2020.

\bibitem[BEG18]{BEG18}
Leonid Barenboim, Michael Elkin, and Uri Goldenberg.
\newblock Locally-iterative distributed {$(\Delta+1)$}-coloring below
  {Szegedy-Vishwanathan} barrier, and applications to self-stabilization and to
  restricted-bandwidth models.
\newblock In {\em Proc.\ 37th {ACM} Symp.\ on Principles of Distributed
  Computing (PODC)}, pages 437--446, 2018.

\bibitem[BEK14]{barenboim14distributed}
Leonid Barenboim, Michael Elkin, and Fabian Kuhn.
\newblock {Distributed ($\Delta$+1)-Coloring in Linear (in $\Delta$) Time}.
\newblock {\em SIAM Journal on Computing}, 43(1):72--95, 2014.

\bibitem[BEPS12]{BEPS}
Leonid Barenboim, Michael Elkin, Seth Pettie, and Johannes Schneider.
\newblock The locality of distributed symmetry breaking.
\newblock In {\em Proc.\ 53rd {IEEE} Symp.\ on Foundations of Computer Science
  (FOCS)}, pages 321--330, 2012.

\bibitem[BFH{\etalchar{+}}16]{Brandt2016}
Sebastian Brandt, Orr Fischer, Juho Hirvonen, Barbara Keller, Tuomo
  Lempi{\"{a}}inen, Joel Rybicki, Jukka Suomela, and Jara Uitto.
\newblock {A lower bound for the distributed Lov{\'{a}}sz local lemma}.
\newblock In {\em Proceedings of the 48th ACM Symposium on Theory of Computing
  (STOC 2016)}, pages 479--488. ACM Press, 2016.

\bibitem[BGHS17]{phmis}
Ioana~O. Bercea, Navin Goyal, David~G. Harris, and Aravind Srinivasan.
\newblock On computing maximal independent sets of hypergraphs in parallel.
\newblock {\em ACM Trans. Parallel Comput.}, 3(1), 2017.

\bibitem[BGKO22]{BGKO2022}
Alkida Balliu, Mohsen Ghaffari, Fabian Kuhn, and Dennis Olivetti.
\newblock {Node and Edge Averaged Complexities of Local Graph Problems}.
\newblock In {\em Proceedings of the 2022 ACM Symposium on Principles of
  Distributed Computing (PODC)}, 2022.

\bibitem[BKR{\etalchar{+}}21]{spaa21_tokendropping}
Sebastian Brandt, Barbara Keller, Joel Rybicki, Jukka Suomela, and Jara Uitto.
\newblock Efficient load-balancing through distributed token dropping.
\newblock In {\em Proc.\ 33rd ACM Symp.\ on Parallelism in Algorithms and
  Architectures (SPAA)}, pages 129--139, 2021.

\bibitem[BL90]{BeameL90}
Paul Beame and Michael Luby.
\newblock Parallel search for maximal independence given minimal dependence.
\newblock In {\em Proc.\ 1st {ACM-SIAM} Symp. on Discrete Algorithms (SODA)},
  pages 212--218, 1990.

\bibitem[BO20]{trulytight}
Sebastian Brandt and Dennis Olivetti.
\newblock Truly tight-in-{$\Delta$} bounds for bipartite maximal matching and
  variants.
\newblock In {\em Proc.\ 39th {ACM} Symp.\ on Principles of Distributed
  Computing (PODC)}, pages 69--78, 2020.

\bibitem[Bra19]{Brandt2019}
Sebastian Brandt.
\newblock An automatic speedup theorem for distributed problems.
\newblock In {\em Proceedings of the 2019 {ACM} Symposium on Principles of
  Distributed Computing, {PODC} 2019, Toronto, ON, Canada, July 29 - August 2,
  2019}, pages 379--388, 2019.

\bibitem[CFK{\etalchar{+}}07]{CFKLMMPSSWW07}
Ke~Chen, Amos Fiat, Haim Kaplan, Meital Levy, Ji\v{r}\'{\i} Matou\v{s}ek,
  Elchanan Mossel, J\'{a}nos Pach, Micha Sharir, Shakhar Smorodinsky, Uli
  Wagner, and Emo Welzl.
\newblock Online conflict‐free coloring for intervals.
\newblock {\em SIAM Journal on Computing}, 36(5):1342--1359, 2007.

\bibitem[CFP{\etalchar{+}}21]{CastanedaFPRRT21}
Armando Casta{\~{n}}eda, Pierre Fraigniaud, Ami Paz, Sergio Rajsbaum, Matthieu
  Roy, and Corentin Travers.
\newblock A topological perspective on distributed network algorithms.
\newblock {\em Theor. Comput. Sci.}, 849:121--137, 2021.

\bibitem[CH03]{czygrinow03}
A.~Czygrinow and M.~Ha\'{n}\'{c}kowiak.
\newblock Distributed algorithm for better approximation of the maximum
  matching.
\newblock In {\em Proc.\ 9th Annual Int.\ Computing and Combinatorics Conf.\
  (COCOON)}, pages 242--251, 2003.

\bibitem[CHSW12]{CzygrinowHSW12}
Andrzej Czygrinow, Michal Han{\'{c}}kowiak, Edyta Szyma{\'{n}}ska, and Wojciech
  Wawrzyniak.
\newblock Distributed 2-approximation algorithm for the semi-matching problem.
\newblock In {\em Proc.\ 26th Symp.\ on Distributed Computing (DISC)}, pages
  210--222, 2012.

\bibitem[CKP12]{uniquemaximum}
Panagiotis Cheilaris, Bal{\'a}zs Keszegh, and D{\"o}m{\"o}t{\"o}r
  P{\'a}lv{\"o}lgyi.
\newblock Unique-maximum and conflict-free coloring for hypergraphs and tree
  graphs.
\newblock In {\em SOFSEM 2012: Theory and Practice of Computer Science}, 2012.

\bibitem[CR12]{CastanedaR12}
Armando Casta{\~{n}}eda and Sergio Rajsbaum.
\newblock New combinatorial topology bounds for renaming: The upper bound.
\newblock {\em J. {ACM}}, 59(1):3:1--3:49, 2012.

\bibitem[FGK17]{FischerGK17}
Manuela Fischer, Mohsen Ghaffari, and Fabian Kuhn.
\newblock Deterministic distributed edge-coloring via hypergraph maximal
  matching.
\newblock In {\em Proc.\ 58th {IEEE} Symp.\ on Foundations of Computer Science
  (FOCS)}, pages 180--191, 2017.

\bibitem[FHK16]{fraigniaud16local}
Pierre Fraigniaud, Marc Heinrich, and Adrian Kosowski.
\newblock Local conflict coloring.
\newblock In {\em Proc.\ 57th IEEE Symp.\ on Foundations of Computer Science
  (FOCS)}, pages 625--634, 2016.

\bibitem[Fis17]{fischer17improved}
Manuela Fischer.
\newblock Improved deterministic distributed matching via rounding.
\newblock In {\em Proc.\ 31st Symp.\ on Distributed Computing (DISC)}, pages
  17:1--17:15, 2017.

\bibitem[GGR21]{GGR2020}
Mohsen Ghaffari, Christoph Grunau, and V{\'{a}}clav Rozhon.
\newblock Improved deterministic network decomposition.
\newblock In {\em Proc.\ 32nd {ACM-SIAM} Symposium on Discrete Algorithms
  (SODA)}, pages 2904--2923, 2021.

\bibitem[Gha16]{ghaffari16improved}
Mohsen Ghaffari.
\newblock An improved distributed algorithm for maximal independent set.
\newblock In {\em Proceedings of the 27th Annual ACM-SIAM Symposium on Discrete
  Algorithms (SODA 2016)}, pages 270--277, 2016.

\bibitem[GHK18]{FOCS18-derand}
Mohsen Ghaffari, David~G. Harris, and Fabian Kuhn.
\newblock On derandomizing local distributed algorithms.
\newblock In {\em Proc.\ 59th Symp.\ on Foundations of Computer Science
  (FOCS)}, pages 662--673, 2018.

\bibitem[GKM17]{GhaffariKM17-slocal}
Mohsen Ghaffari, Fabian Kuhn, and Yannic Maus.
\newblock On the complexity of local distributed graph problems.
\newblock In {\em Proc.\ 49th {ACM} Symp.\ on Theory of Computing (STOC)},
  pages 784--797, 2017.

\bibitem[GKMU18]{GhaffariKMU18}
Mohsen Ghaffari, Fabian Kuhn, Yannic Maus, and Jara Uitto.
\newblock Deterministic distributed edge-coloring with fewer colors.
\newblock In {\em Proc.\ 50th {ACM} Symposium on Theory of Computing (STOC)},
  pages 418--430, 2018.

\bibitem[Har19]{phmis2}
David~G. Harris.
\newblock Derandomized concentration bounds for polynomials, and hypergraph
  maximal independent set.
\newblock {\em ACM Trans. Algorithms}, 15(3), 2019.

\bibitem[Har20]{Harris20}
David~G. Harris.
\newblock Distributed local approximation algorithms for maximum matching in
  graphs and hypergraphs.
\newblock {\em {SIAM} J. Comput.}, 49(4):711--746, 2020.

\bibitem[HK73]{HopcroftK73}
J.~E. Hopcroft and R.~M. Karp.
\newblock An n\({}^{\mbox{5/2}}\) algorithm for maximum matchings in bipartite
  graphs.
\newblock {\em {SIAM} J. Comput.}, 2(4):225--231, 1973.

\bibitem[HKP98]{Hanckowiak1998}
Michal Hanckowiak, Michal Karonski, and Alessandro Panconesi.
\newblock On the distributed complexity of computing maximal matchings.
\newblock In {\em Proc.\ 9th {ACM-SIAM} Symp.\ on Discrete Algorithms (SODA)},
  pages 219--225, 1998.

\bibitem[HKP01]{Hanckowiak2001}
Michal Hanckowiak, Michal Karonski, and Alessandro Panconesi.
\newblock On the distributed complexity of computing maximal matchings.
\newblock {\em SIAM Journal on Discrete Mathematics}, 15(1):41--57, 2001.

\bibitem[HKR13]{topology_book}
Maurice Herlihy, Dimitry~N. Kozlov, and Sergio Rajsbaum.
\newblock {\em Distributed Computing Through Combinatorial Topology}.
\newblock Morgan Kaufmann, 2013.

\bibitem[HMP{\etalchar{+}}16]{HarrisMPRS16}
David~G. Harris, Ehab Morsy, Gopal Pandurangan, Peter Robinson, and Aravind
  Srinivasan.
\newblock Efficient computation of sparse structures.
\newblock {\em Random Struct. Algorithms}, 49(2):322--344, 2016.

\bibitem[HS99]{HerlihyS99}
Maurice Herlihy and Nir Shavit.
\newblock The topological structure of asynchronous computability.
\newblock {\em J. {ACM}}, 46(6):858--923, 1999.

\bibitem[HS05]{Har-PeledS05}
Sariel Har{-}Peled and Shakhar Smorodinsky.
\newblock Conflict-free coloring of points and simple regions in the plane.
\newblock {\em Discret. Comput. Geom.}, 34(1):47--70, 2005.

\bibitem[HV06]{HougardyV06}
S.\ Hougardy and D.\ E.~Drake Vinkemeier.
\newblock Approximating weighted matchings in parallel.
\newblock {\em Inf. Process. Lett.}, 99(3):119--123, 2006.

\bibitem[II86]{Israeli1986}
Amos Israeli and A.~Itai.
\newblock {A fast and simple randomized parallel algorithm for maximal
  matching}.
\newblock {\em Information Processing Letters}, 22(2):77--80, 1986.

\bibitem[Kel92]{Kelsen92}
Pierre Kelsen.
\newblock On the parallel complexity of computing a maximal independent set in
  a hypergraph.
\newblock In {\em Proc.\ 24th {ACM} Symp.\ on Theory of Computing (STOC)},
  pages 339--350, 1992.

\bibitem[KMW04]{KuhnMW04}
Fabian Kuhn, Thomas Moscibroda, and Roger Wattenhofer.
\newblock What cannot be computed locally!
\newblock In {\em Proc.\ 23rd {ACM} Symp.\ on Principles of Distributed
  Computing (PODC)}, pages 300--309, 2004.

\bibitem[KMW16]{KuhnMW16}
Fabian Kuhn, Thomas Moscibroda, and Roger Wattenhofer.
\newblock Local computation: Lower and upper bounds.
\newblock {\em Journal of {ACM}}, 63:17:1--17:44, 2016.

\bibitem[KMW18]{KuhnMW18}
Fabian Kuhn, Yannic Maus, and Simon Weidner.
\newblock Deterministic distributed ruling sets of line graphs.
\newblock In {\em Structural Information and Communication Complexity - 25th
  International Colloquium, {SIROCCO} 2018, Ma'ale HaHamisha, Israel, June
  18-21, 2018, Revised Selected Papers}, pages 193--208, 2018.

\bibitem[KNPR14]{KuttenNPR14}
Shay Kutten, Danupon Nanongkai, Gopal Pandurangan, and Peter Robinson.
\newblock Distributed symmetry breaking in hypergraphs.
\newblock In {\em Proc.\ 28th Symp.\ on Distributed Computing (DISC)}, pages
  469--483, 2014.

\bibitem[Kuh09]{Kuhn2009}
Fabian Kuhn.
\newblock Local weak coloring algorithms and implications on deterministic
  symmetry breaking.
\newblock In {\em Proc.\ 21st ACM Symp.\ on Parallelism in Algorithms and
  Architectures (SPAA)}, 2009.

\bibitem[KUW88]{KarpUW88}
Richard~M. Karp, Eli Upfal, and Avi Wigderson.
\newblock The complexity of parallel search.
\newblock {\em J. Comput. Syst. Sci.}, 36(2):225--253, 1988.

\bibitem[KZ18]{KuhnZheng}
Fabian Kuhn and Chaodong Zheng.
\newblock Efficient distributed computation of {MIS} and generalized {MIS} in
  linear hypergraphs.
\newblock {\em CoRR}, abs/1805.03357, 2018.

\bibitem[Lin87]{Linial1987}
Nathan Linial.
\newblock {Distributive graph algorithms – Global solutions from local data}.
\newblock In {\em Proc.\ 28th Symp.\ on Foundations of Computer Science (FOCS
  1987)}, pages 331--335. IEEE, 1987.

\bibitem[Lin92]{Linial1992}
Nathan Linial.
\newblock {Locality in Distributed Graph Algorithms}.
\newblock {\em SIAM Journal on Computing}, 21(1):193--201, 1992.

\bibitem[LPP15]{lotker15}
Z.~Lotker, B.~Patt{-}Shamir, and S.~Pettie.
\newblock Improved distributed approximate matching.
\newblock {\em J.\ ACM}, 62(5):38:1--38:17, 2015.

\bibitem[LS93]{LinialS93}
N.~Linial and M.~Saks.
\newblock Low diameter graph decompositions.
\newblock {\em Combinatorica}, 13(4):441--454, 1993.

\bibitem[LS97]{LuczakS97}
Tomasz Luczak and Edyta Szymanska.
\newblock A parallel randomized algorithm for finding a maximal independent set
  in a linear hypergraph.
\newblock {\em J. Algorithms}, 25(2):311--320, 1997.

\bibitem[Lub86]{Luby1986}
Michael Luby.
\newblock A simple parallel algorithm for the maximal independent set problem.
\newblock {\em {SIAM} Journal on Computing}, 15(4):1036--1053, 1986.

\bibitem[LW11]{LenzenW11}
Christoph Lenzen and Roger Wattenhofer.
\newblock {MIS} on trees.
\newblock In {\em Proceedings of the 2011 Annual {ACM} Symposium on Principles
  of Distributed (PODC)}, pages 41--48, 2011.

\bibitem[MT20]{MausTonoyan20}
Yannic Maus and Tigran Tonoyan.
\newblock Local conflict coloring revisited: Linial for lists.
\newblock In {\em 34th International Symposium on Distributed Computing,
  {DISC}}, pages 16:1--16:18, 2020.

\bibitem[Nao91]{Naor1991}
Moni Naor.
\newblock {A lower bound on probabilistic algorithms for distributive ring
  coloring}.
\newblock {\em SIAM Journal on Discrete Mathematics}, 4(3):409--412, 1991.

\bibitem[Oli19]{Olivetti2019}
Dennis Olivetti.
\newblock {Round Eliminator: a tool for automatic speedup simulation}, 2019.

\bibitem[PR01]{panconesi01simple}
Alessandro Panconesi and Romeo Rizzi.
\newblock {Some simple distributed algorithms for sparse networks}.
\newblock {\em Distributed Computing}, 14(2):97--100, 2001.

\bibitem[PS96]{panconesi96decomposition}
Alessandro Panconesi and Aravind Srinivasan.
\newblock {On the Complexity of Distributed Network Decomposition}.
\newblock {\em Journal of Algorithms}, 20(2):356--374, 1996.

\bibitem[RG20]{Rozhon2020}
V{\'{a}}clav Rozho\v{n} and Mohsen Ghaffari.
\newblock Polylogarithmic-time deterministic network decomposition and
  distributed derandomization.
\newblock In {\em Proc.\ 52nd ACM Symp.\ on Theory of Computing (STOC)}, pages
  350--363, 2020.

\bibitem[SZ00]{SaksZ00}
Michael~E. Saks and Fotios Zaharoglou.
\newblock Wait-free k-set agreement is impossible: The topology of public
  knowledge.
\newblock {\em {SIAM} J. Comput.}, 29(5):1449--1483, 2000.

\end{thebibliography}

\end{document}